\documentclass[journal,twoside,twocolumn]{IEEEtran}

\usepackage{cite}

\usepackage[cmex10]{amsmath}
\usepackage{array}

\usepackage[pdftex]{graphicx}

\usepackage{amsfonts,amssymb,amsthm,bbm,verbatim}
\allowdisplaybreaks 
\usepackage{mathrsfs,dsfont} 
\usepackage{hyperref}
\hypersetup{colorlinks = true, linkcolor = blue, citecolor = blue}
\usepackage[all]{hypcap} 
\usepackage[usenames,dvipsnames]{color}
\usepackage[]{inputenc}
\usepackage{physics} 
\usepackage{eqparbox} 
\def\X{{\mathcal{X}}}
\def\Y{{\mathcal{Y}}}
\def\M{{\mathcal{M}}} 
\def\Cperm{{C_{\mathsf{perm}}}} 
 
\def\Simplex{{\mathcal{P}}}
\def\E{{\mathbb{E}}}
\def\VAR{{\mathbb{V}\mathbb{A}\mathbb{R}}}
\def\R{{\mathbb{R}}}
\def\P{{\mathbb{P}}}
\def\N{{\mathbb{N}}}
\def\Z{{\mathbb{Z}}}
\def\1{{\textbf{1}}}
\def\0{{\textbf{0}}}
\def\I{{\mathds{1}}}
\def\Ber{{\mathsf{Ber}}}
\def\bin{{\mathsf{bin}}}
\def\BSC{{\mathsf{BSC}}}
\def\BEC{{\mathsf{BEC}}}
\def\qEC{{q\text{-}\mathsf{EC}}}
\def\qSC{{q\text{-}\mathsf{SC}}}
\DeclareMathOperator*{\argmin}{arg\,min}
\DeclareMathOperator*{\argmax}{arg\,max}
\DeclareMathOperator*{\ext}{ext}
\newcommand{\T}{\mathrm{T}}

\newtheorem{theorem}{Theorem}
\newtheorem{lemma}{Lemma}
\newtheorem{proposition}{Proposition}

\theoremstyle{definition}
\newtheorem{definition}{Definition}
\newtheorem{conjecture}{Conjecture}

\begin{document}

\bstctlcite{IEEEexample:BSTcontrol} 

\title{Coding Theorems for Noisy Permutation Channels}

\author{Anuran~Makur%
\thanks{This work was presented in part at the 2020 IEEE International Symposium on Information Theory \cite{Makur2020Conf}, and a very preliminary version of this work was presented in part at the 2018 56th Annual Allerton Conference on Communication, Control, and Computing \cite{Makur2018}.}%
\thanks{A. Makur is with the Department of Electrical Engineering and Computer Science, Massachusetts Institute of Technology, Cambridge, MA 02139, USA (e-mail: a\_makur@mit.edu).}%
\thanks{Copyright (c) 2020 IEEE. Personal use of this material is permitted. However, permission to use this material for any other purposes must be obtained from the IEEE by sending a request to pubs-permissions@ieee.org.}}%

\maketitle

\begin{abstract}
In this paper, we formally define and analyze the class of noisy permutation channels. The noisy permutation channel model constitutes a standard discrete memoryless channel (DMC) followed by an independent random permutation that reorders the output codeword of the DMC. While coding theoretic aspects of this model have been studied extensively, particularly in the context of reliable communication in network settings where packets undergo transpositions, and closely related models of DNA based storage systems have also been analyzed recently, we initiate an information theoretic study of this model by defining an appropriate notion of \textit{noisy permutation channel capacity}. Specifically, on the achievability front, we prove a lower bound on the noisy permutation channel capacity of any DMC in terms of the rank of the stochastic matrix of the DMC. On the converse front, we establish two upper bounds on the noisy permutation channel capacity of any DMC whose stochastic matrix is strictly positive (entry-wise). Together, these bounds yield coding theorems that characterize the noisy permutation channel capacities of every strictly positive and ``full rank'' DMC, and our achievability proof yields a conceptually simple, computationally efficient, and capacity achieving coding scheme for such DMCs. Furthermore, we also demonstrate the relation between the well-known output degradation preorder over channels and noisy permutation channel capacity. In fact, the proof of one of our converse bounds exploits a degradation result that constructs a symmetric channel for any DMC such that the DMC is a degraded version of the symmetric channel. Finally, we illustrate some examples such as the special cases of binary symmetric channels and (general) erasure channels. Somewhat surprisingly, our results suggest that noisy permutation channel capacities are generally quite agnostic to the parameters that define the DMCs.
\end{abstract}


\tableofcontents
\hypersetup{linkcolor = red}

\section{Introduction}
\label{Introduction}

In this paper, we initiate an information theoretic study of the problem of reliable communication through noisy permutation channels by defining and analyzing a pertinent notion of information capacity for such channels. Noisy permutation channels refer to discrete memoryless channels (DMCs) followed by independent random permutation transformations that are applied to the entire blocklength of the output codeword. Such channels can be perceived as models of communication links in networks where packets are not delivered in sequence, and hence, the ordering of the packets does not carry any information. Moreover, they also bear a close resemblance to recently introduced models of deoxyribonucleic acid (DNA) based storage systems. The main contributions of this work are the following:
\begin{enumerate}
\item We formalize the notion of ``noisy permutation channel capacity'' of a DMC in Definition \ref{Def: Permutation Channel Capacity}, which captures, up to first order, the maximum number of messages than can be transmitted through a noisy permutation channel model with vanishing probability of error as the blocklength tends to infinity. (Although our formalism is quite natural, it has not appeared in the literature to our knowledge.)
\item We establish an achievability bound on the noisy permutation channel capacity of any DMC in terms of the rank of the DMC in Theorem \ref{Thm: Achievability Bound} by analyzing a conceptually simple and computationally tractable randomized coding scheme. Moreover, we also demonstrate an alternative proof of our achievability bound for DMCs that have rank $2$ in Proposition \ref{Prop: Achievability Bound for DMCs with Rank 2} by using the so called second moment method (in Lemmata \ref{Lemma: Second Moment Method} and \ref{Lemma: Testing between Converging Hypotheses}).
\item We prove two converse bounds on the noisy permutation channel capacity of any DMC that is strictly positive (entry-wise). The first bound, in Theorem \ref{Thm: Converse Bound I}, is in terms of the output alphabet size of the DMC, and the second bound, in Theorem \ref{Thm: Converse Bound II}, is in terms of the ``effective input alphabet'' size of the DMC. (Neither bound is uniformly better than the other.)
\item Using the aforementioned achievability and converse bounds, we exactly characterize the noisy permutation channel capacity of all strictly positive and ``full rank'' DMCs in Theorem \ref{Thm: Strictly Positive and Full Rank Channels}. Furthermore, we propound a candidate solution for the noisy permutation channel capacity of general strictly positive DMCs (regardless of their rank) in Conjecture \ref{Conj: Strictly Positive DMC Conjecture}.
\item To complement these results and assist in understanding them, we derive an intuitive monotonicity relation between the degradation preorder over channels and noisy permutation channel capacity in Theorem \ref{Thm: Comparison Bound via Degradation} (also see Theorem \ref{Thm: Comparison Bounds via Degradation}). Furthermore, we also construct symmetric channels that dominate given DMCs in the degradation sense in Proposition \ref{Prop: Degradation by Symmetric Channels}. This construction is utilized in the proof of Theorem \ref{Thm: Converse Bound II}.
\item Finally, we present exact characterizations of the noisy permutation channel capacities of several specific families of channels, e.g., binary symmetric channels in Proposition \ref{Prop: Permutation Channel Capacity of BSC} (cf. \cite[Theorem 3]{Makur2018}), channels with unit rank transition kernels in Proposition \ref{Prop: Unit Rank Stochastic Matrices}, and channels with permutation matrices as transition kernels in Proposition \ref{Prop: Permutation Stochastic Matrices}. Furthermore, we present bounds on the noisy permutation channel capacities of (general) erasure channels in Proposition \ref{Prop: Permutation Channel Capacity of q-EC}, and also propose related conjectures (see, e.g., Conjecture \ref{Conj: BEC Conjecture}). In particular, although Theorem \ref{Thm: Achievability Bound} yields our achievability bound for erasure channels, we show an alternative achievability proof in Proposition \ref{Prop: Permutation Channel Capacity of q-EC} by exploiting the classical notion of Doeblin minorization.
\end{enumerate}

The ensuing subsections provide some background literature to motivate our study, a formal description of the noisy permutation channel model, some additional notation that will be used throughout the paper, and an outline of the remainder of the paper.

\subsection{Related Literature and Motivation}

The setting of channel coding with transpositions, where the output codeword undergoes some reordering of its symbols, has been widely studied in the coding theory, communication networks, and molecular and biological communications communities. We briefly discuss some relevant literature from these three disciplines, each of which provides a compelling incentive to study noisy permutation channels. 

Firstly, in the coding theory literature, one earlier line of work concerned the construction of error-correcting codes that achieve capacity of the \textit{random deletion channel}, cf. \cite{DiggaviGrossglauser2001,Mitzenmacher2006,Metzner2009}. The random deletion channel operated on the codeword space by deleting each input codeword symbol independently with some probability $p \in (0,1)$, and copying it otherwise. As expounded in \cite[Section I]{Mitzenmacher2006}, with sufficiently large alphabet size $2^b$, where each symbol of the alphabet was construed as a packet with $b$ bits and $b = \Omega(\log(n))$ depended on the blocklength $n$, ``embedding sequence numbers into the transmitted symbols [turned] the deletion channel [into a memoryless] erasure channel.'' Since coding for erasure channels was well-understood, the intriguing question became to construct (nearly) capacity achieving codes for the random deletion channel using sufficiently large packet length $b$ (depending on $n$), but without embedding sequence numbers (see, e.g., \cite{Mitzenmacher2006,DiggaviGrossglauser2001}, and the references therein).\footnote{We also refer readers to the recent work \cite{Kudekaretal2017}, which proves that \emph{Reed-Muller codes} achieve capacity for erasure channels, and the references therein.} In particular, the author of \cite{Mitzenmacher2006} demonstrated that \textit{low density parity check (LDPC) codes} with verification-based decoding formed a computationally tractable coding scheme with these properties. Notably, this coding scheme also tolerated transpositions of packets that were not deleted in the process of transmission. Therefore, it was equivalently a coding scheme for a memoryless erasure channel followed by a random permutation block, albeit with an alphabet size that grew polynomially with the blocklength.

Broadly speaking, the results of \cite{Mitzenmacher2006,DiggaviGrossglauser2001} can be perceived as preliminary steps towards analyzing the fundamental limits of reliable communication through noisy permutation channels where the DMCs are erasure channels. Several other coding schemes for erasure permutation channels with sufficiently large alphabet size have also been developed in the literature. We refer readers to \cite{Metzner2009}, which builds upon the key conceptual ideas in \cite{Mitzenmacher2006}, and the references therein for other examples of such coding schemes.

Secondly, this discussion concerning the random deletion channel has a patent counterpart in the (closely related) communication networks literature. Indeed, in the context of the well-known \textit{store-and-forward} transmission scheme for packet networks, packet losses (or deletions) were typically corrected using \textit{Reed-Solomon codes} which assumed that each packet carried a header with a sequence number\textemdash see, e.g., \cite{XuZhang2002}, \cite[Section I]{GadouleauGoupil2010}, and the references therein. Akin to the random deletion channel setting, this simplified the error correction problem since packet losses could be treated as erasures. However, ``motivated by networks whose topologies change over time, or where several routes with unequal delays are available to transmit the data,'' the authors of \cite{GadouleauGoupil2010} illustrated that packet errors and losses could also be corrected using binary codes under a channel model where the impaired or lost packets were randomly permuted, and the packets were not indexed with sequence numbers. Such work can also be construed as developing codes for specific kinds of noisy permutation channels. 

\begin{figure}[t]
\centering
\includegraphics[trim = 20mm 50mm 10mm 30mm, clip, width=\linewidth]{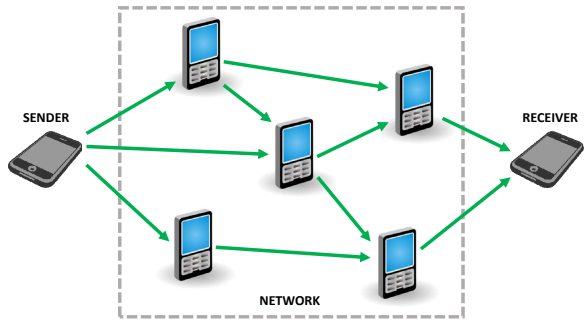} 
\caption{Illustration of point-to-point communication between a sender and a receiver through a mobile ad hoc network (MANET).}
\label{Figure: Real Network}
\end{figure}

In general, the noisy permutation channel model in subsection \ref{Permutation Channel Model} is a simple and useful abstraction for point-to-point communication between a source and a receiver in various network settings. For instance, when information is transmitted using a lower level \textit{multipath routed network} (see, e.g., Figure \ref{Figure: Real Network}), the set of all possible packets make up the channel input alphabet with each packet representing a different symbol (as mentioned earlier), and any context specific packet impairments are represented by the DMC in the model. Furthermore, since the packets (or symbols) can take different paths to the receiver in such a network, they may arrive at the destination out-of-order due to different delay profiles in the different paths. This out-of-order delivery of packets is captured by the random permutation transformation in the model. Several other aspects of noisy permutation channels have also been investigated in the communication networks literature. For example, the authors of \cite{WalshWeberMaina2009} established rate-delay tradeoffs for multipath routed networks, although they neglected to account for packet impairments, such as deletions, in their analysis for simplicity.

More recently, inspired by packet networks such as mobile ad hoc networks (where the network topology changes over time)\textemdash see Figure \ref{Figure: Real Network}, and heavily loaded datagram-based networks (where packets are often re-routed for load balancing purposes), the authors of \cite{KovacevicVukobratovic2013,KovacevicVukobratovic2015,KovacevicTan2018a} have considered the general problem of coding in channels where the codeword undergoes a random permutation and is subjected to impairments such as insertions, deletions, substitutions, and erasures. As stated in \cite[Section I]{KovacevicVukobratovic2013}, the basic strategy to reliably communicate across a channel that applies a transformation to its codewords is to ``encode the information in an object that is invariant under the given transformation.'' In the case of noisy permutation channels, the appropriate codes are the so called \textit{multiset codes}, where the codewords are characterized by their empirical distribution over the underlying alphabet. The existence of certain perfect multiset codes is established in \cite{KovacevicVukobratovic2015}, and several other multiset code constructions based on lattices and \textit{Sidon sets} are analyzed in \cite{KovacevicTan2018a}. 

Thirdly, an alternative motivation for analyzing noisy permutation channels stems from research at the intersection of computational biology and information theory on \textit{DNA based storage systems}, cf. \cite{Yazdietal2015,KiahPuleoMilenkovic2016,Heckeletal2017,KovacevicTan2018b,ShomoronyHeckel2019}. For example, the authors of \cite{Heckeletal2017} examined the storage capacity of systems where the source is encoded using DNA molecules. In their model, source data was encoded into codeword strings (or DNA molecules) consisting of letters from an alphabet of four nucleobases, and short fragments of these codewords were then cached in an unordered fashion akin to the effect of the random permutation in our noisy permutation channel model. The receiver read the encoded data by shotgun sequencing, or equivalently, by randomly sampling the stored and unordered fragments. While the unordered caching aspect of this model resembles our model, as stated in \cite[Section I-B]{KovacevicTan2018a}, this storage model also differs from our model since the receiver samples the stored codewords with replacement and without errors.

A very closely related DNA based storage model to \cite{Heckeletal2017}, known as the \emph{noisy shuffling channel}, is investigated in \cite{ShomoronyHeckel2019}. Specifically, in order to represent the corruption of DNA molecules during ``synthesis, sequencing, and\dots storage,'' the authors of \cite{ShomoronyHeckel2019} studied the storage capacity of a model where DNA codewords first experienced the deleterious effects of a DMC (e.g., a binary symmetric channel), and were then fragmented, and subsequently, the fragments were randomly permuted. (Unlike \cite{Heckeletal2017}, the receiver had access to all the permuted fragments in this model for simplicity.) The DNA based storage model in \cite{ShomoronyHeckel2019} is much closer to our noisy permutation channel model than the model in \cite{Heckeletal2017}. However, in contrast to our model, both \cite{Heckeletal2017} and \cite{ShomoronyHeckel2019} assume that the lengths of the codeword fragments (which are permuted) grow logarithmically with the number of fragments. We refer readers to \cite{Yazdietal2015} for a broader overview of DNA based storage systems, and to \cite{KiahPuleoMilenkovic2016,KovacevicTan2018b}, and the references therein for other examples of codes for such systems. Moreover, we also refer readers to the comprehensive bibliography in \cite{KovacevicTan2018a} for further related literature on noisy permutation channels.  

Finally, it is worth re-emphasizing that the noisy permutation channel model in subsection \ref{Permutation Channel Model} may be regarded as a variant or generalization of the models described above. More precisely, the analysis in \cite{Mitzenmacher2006,DiggaviGrossglauser2001} pertains to erasure permutation channels where the alphabet size grows polynomially with the blocklength, the work in \cite{GadouleauGoupil2010,KovacevicVukobratovic2013,KovacevicVukobratovic2015,KovacevicTan2018a} is concerned with various codes for specific noisy permutation channels, and the focus of \cite{ShomoronyHeckel2019} is on noisy shuffling channels that randomly permute fragments of codewords whose lengths scale logarithmically with the number of fragments. In comparison, our results on noisy permutation channels in this paper consider much broader classes of DMCs, and assume that alphabet sizes are constant with respect to the blocklength, or alternatively, that fragment lengths are constant with respect to the number of fragments. (Note that this latter assumption ensures that we cannot add sequence numbers to packets in order to transform our problem into one of classical coding.) 

\begin{figure*}[ht]
\centering
\includegraphics[trim = 0mm 100mm 0mm 80mm, clip, width=\linewidth]{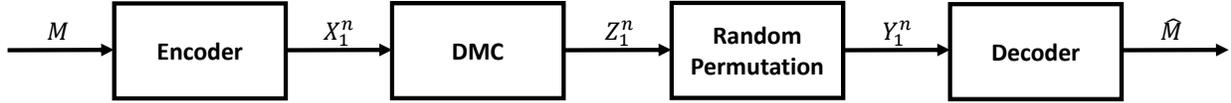} 
\caption{Illustration of a communication system with a DMC followed by a random permutation.}
\label{Figure: Permutation Channel}
\end{figure*}

Furthermore, as the discussion heretofore reveals, the majority of the literature on noisy permutation channels analyzes its coding theoretic aspects. In contrast, we approach these channels from a purely information theoretic perspective. To our knowledge, such a systematic analysis has not been undertaken until now, and thus, there are no known results on the information capacity of the noisy permutation channel model described in the next subsection. (Indeed, while the aforementioned references \cite{DiggaviGrossglauser2001}, \cite{Heckeletal2017}, and \cite{ShomoronyHeckel2019} have a more information theoretic focus, they analyze different models to ours.) In this paper, we will take some first steps towards a complete understanding of the information capacity of noisy permutation channels. Rather interestingly, our main achievability proof will automatically yield computationally tractable codes for reliable communication through certain noisy permutation channels, thereby rendering the need to develop conceptually sophisticated coding schemes for these channels futile when (theoretically) achieving noisy permutation channel capacity is the sole objective. 

\subsection{Noisy Permutation Channel Model}
\label{Permutation Channel Model}

We define the point-to-point noisy permutation channel model in analogy with standard information theoretic definitions, cf. \cite[Section 7.5]{CoverThomas2006}. Let $n \in \N \triangleq \{1,2,3,\dots\}$ denote a fixed blocklength, $M \in \M \triangleq \{1,\dots,|\M|\}$ be a \textit{message} random variable that is drawn uniformly from the message set $\M$, $f_n:\M \rightarrow \X^n$ be a (possibly randomized) \textit{encoder}, where $\X$ is the finite input alphabet of the channel with $|\X| \geq 2$, and $g_n:\Y^n \rightarrow \M \cup\! \{\mathsf{e}\}$ be a (possibly randomized) \textit{decoder}, where $\Y$ is the finite output alphabet of the channel with $|\Y| \geq 2$ and $\mathsf{e}$ denotes an additional ``error symbol.'' The message $M$ is first encoded into a codeword $X_1^n = f_n(M)$, where each $X_i \in \X$, and we use the notation $X_i^j \triangleq (X_i,\dots,X_j)$ for $i < j$. This codeword is transmitted through a (stationary) \textit{discrete memoryless channel} defined by the conditional probability distributions $\{P_{Z|X}(\cdot|x) \in \Simplex_{\Y} : x \in \X\}$ to produce $Z_1^n \in \Y^n$, where each $Z_i \in \Y$, and $\Simplex_{\Y}$ denotes the probability simplex in $\R^{|\Y|}$ of all probability distributions on $\Y$. Note that in later sections, we will often treat a DMC $P_{Z|X}$ as a row stochastic transition probability matrix $P_{Z|X} \in \R^{|\X| \times |\Y|}$ whose rows are given by $\{P_{Z|X}(\cdot|x) \in \Simplex_{\Y} : x \in \X\}$, and vice versa, since the two perspectives are equivalent. (In particular, for every $x \in \X$ and $y \in \Y$, the conditional probability $P_{Z|X}(y|x)$ is also the $(x,y)$th element of the matrix $P_{Z|X}$, i.e., $P_{Z|X}(y|x) = \big[P_{Z|X}\big]_{x,y}$ using the notation in subsection \ref{Additional Notation}. Likewise, for every $x \in \X$, the conditional distribution $P_{Z|X}(\cdot|x) \in \Simplex_{\Y}$ forms the $x$th row of the matrix $P_{Z|X}$.) The memorylessness property of the DMC implies that
\begin{equation}
\label{Eq: Memorylessness}
P_{Z_1^n | X_1^n}(z_1^n | x_1^n) = \prod_{i = 1}^{n}{P_{Z|X}(z_i|x_i)}
\end{equation}
for every $x_1^n \in \X^n$ and every $z_1^n \in \Y^n$. The noisy codeword $Z_1^n$ is then passed through an independent \textit{random permutation} transformation to generate $Y_1^n \in \Y^n$. Specifically, the random permutation channel $\Pi \triangleq \{\Pi(\cdot|z_1^n) = P_{Y_1^n|Z_1^n}(\cdot|z_1^n) : z_1^n \in \Y^n\}$ is defined as
\begin{equation}
\label{Eq: The Random Perm Channel}
\begin{aligned}
\Pi(y_1^n|z_1^n) & = P_{Y_1^n|Z_1^n}(y_1^n|z_1^n) \\
& = \frac{1}{n!} \sum_{\lambda \in \mathcal{S}_n}{\I\!\left\{\forall i \in \{1,\dots,n\}, \, y_{\lambda(i)} = z_i \right\}}
\end{aligned}
\end{equation}
for every $y_1^n,z_1^n \in \Y^n$, where the sum is over all permutations $\lambda:\{1,\dots,n\} \rightarrow \{1,\dots,n\}$ in the symmetric group $\mathcal{S}_n$ over the set $\{1,\dots,n\}$, and $\I\{\cdot\}$ is the indicator function defined in subsection \ref{Additional Notation}. Alternatively, we can describe the action of $\Pi$ in \eqref{Eq: The Random Perm Channel} as follows:
\begin{enumerate}
\item First, randomly draw a bijection (or permutation) $\lambda:\{1,\dots,n\} \rightarrow \{1,\dots,n\}$ uniformly, and independently of everything else, from the symmetric group $\mathcal{S}_n$ over $\{1,\dots,n\}$, 
\item Then, generate $Y_1^n$ from $Z_1^n$ using the permutation $\lambda$ so that $Y_{\lambda(i)} = Z_i$ for all $i \in \{1,\dots,n\}$.
\end{enumerate}
Throughout this paper, we will refer to random permutation channels on different alphabets, such as the one defined above, as ``random permutations'' without any further clarification. Finally, the received codeword $Y_1^n$ is decoded to produce an estimate $\hat{M} = g_n(Y_1^n)$ of $M$. Figure \ref{Figure: Permutation Channel} illustrates this communication system.

Let the \textit{average probability of error} in this model be
\begin{equation}
P_{\mathsf{error}}^{(n)} \triangleq \P\big(M \neq \hat{M}\big) \, ,
\end{equation}
where we assume that any decoder $g_n$ always makes an error when it outputs the error symbol $\mathsf{e}$.\footnote{Under an average probability of error criterion, the sequence of encoders $f_n$ that minimize $P_{\mathsf{error}}^{(n)}$ are deterministic, and the corresponding sequence of decoders $g_n$ that minimize $P_{\mathsf{error}}^{(n)}$ are the \emph{maximum a posteriori} decoders (or maximum likelihood decoders, since $M$ is uniformly distributed), which are also deterministic without loss of generality \cite[Section 16.2.1]{PolyanskiyWu2017Notes}. In contrast, under a maximal probability of error criterion, randomized encoders and decoders can be useful \cite[Section 16.2.1]{PolyanskiyWu2017Notes}.} The ``\textit{rate}'' of the encoder-decoder pair $(f_n,g_n)$ is defined as
\begin{equation}
\label{Eq: Rate}
R \triangleq \frac{\log(|\M|)}{\log(n)} \, ,
\end{equation}
where $\log(\cdot)$ is the binary logarithm (with base $2$) throughout this paper, and all \textit{Shannon entropy} $H(\cdot)$, \textit{mutual information} $I(\cdot;\cdot)$, and \textit{Kullback-Leibler (KL) divergence} (or relative entropy) $D(\cdot||\cdot)$ terms are measured in bits.\footnote{The notion of rate defined in \eqref{Eq: Rate} is analogous to the so called \emph{third-order coding rate} in the finite blocklength analysis literature; see, e.g., \cite{KosutSankar2014}.} So, we can also write $|\M| = n^R$. (Strictly speaking, $n^R$ should be an integer, but we will often neglect this detail since it will not affect our results.) We will say that a rate $R \geq 0$ is \textit{achievable} if there exists a sequence of encoder-decoder pairs $\{(f_n,g_n)\}_{n \in \N}$ such that $\lim_{n \rightarrow \infty}{P_{\mathsf{error}}^{(n)}} = 0$. Lastly, we operationally define the noisy permutation channel capacity as follows.

\begin{definition}[Noisy Permutation Channel Capacity]
\label{Def: Permutation Channel Capacity}
For any DMC $P_{Z|X}$, its \emph{noisy permutation channel capacity} is given by
$$ \Cperm(P_{Z|X}) \triangleq \sup\!\left\{R \geq 0 : R \text{ is achievable}\right\} . $$
\end{definition}

It is straightforward to verify that the scaling in \eqref{Eq: Rate} is indeed $\log(n)$ rather than the standard $n$. As mentioned earlier, due to the independent random permutation in the model, all information embedded in the ordering within codewords is lost. (In fact, canonical fixed composition codes cannot carry more than one message in this setting.) So, the maximum number of decodable messages is (intuitively) upper bounded by the number of possible empirical distributions of $Y_1^n$, i.e.,
\begin{equation}
\label{Eq: Number of Emp Dists}
n^R = |\M| \leq \binom{n+|\Y|-1}{|\Y|-1} \leq (n+1)^{|\Y|-1} \, ,  
\end{equation}
where taking $\log$'s and letting $n \rightarrow \infty$ yields $\Cperm(P_{Z|X}) \leq |\Y|-1$ (at least non-rigorously). This justifies that $\log(n)$ is the correct scaling in \eqref{Eq: Rate}, i.e., the maximum number of messages that can be reliably communicated is polynomial in the blocklength (rather than exponential).

\subsection{Additional Notation}
\label{Additional Notation}

In this subsection, we define some additional notation that will be utilized throughout the paper. We begin with some probabilistic notation. The standard expressions $\P(\cdot)$, $\E[\cdot]$, and $\VAR(\cdot)$ represent the probability, expectation, and variance operators, where the underlying probability measures will be clear from context. Moreover, we will write $X \sim P_X$ when the random variable $X$ has probability law $P_X$. We let $\I\{\cdot\}$ denote the indicator function which equals $1$ if its input proposition is true and $0$ otherwise. Given any sequence $x_1^n \in \X^n$ with $n \in \N$, we define the \textit{empirical distribution} (histogram or \textit{type}) of $x_1^n$ as
\begin{equation}
\hat{P}_{x_1^n} = \left(\hat{P}_{x_1^n}(x^{\prime}) : x^{\prime} \in \X\right) \in \Simplex_{\X} \, ,
\end{equation}
where $\hat{P}_{x_1^n}$ is a probability distribution on $\X$, and for every $x^{\prime} \in \X$,
\begin{equation}
\hat{P}_{x_1^n}(x^{\prime}) \triangleq \frac{1}{n} \sum_{i = 1}^{n}{\I\!\left\{x_i = x^{\prime}\right\}} \, . 
\end{equation} 
For convenience, we will use the notation
\begin{equation}
\binom{n}{n \hat{P}_{x_1^n}} \triangleq \frac{n!}{\displaystyle{\prod_{x^{\prime} \in \X}{\big(n\hat{P}_{x_1^n}(x^{\prime})\big)!}}} 
\end{equation}
for the multinomial coefficient. Furthermore, for any $k \in \N$ and $p \in [0,1]$, we let $\Ber(p)$ denote a Bernoulli distribution with success probability $p$, and $\bin(k,p)$ denote a binomial distribution with $k$ trials and success probability $p$.

Next, we introduce some linear algebraic notation. Fix any $m,n \in \N$. Given any matrix $A \in \R^{m \times n}$, we let $\left[A\right]_{i,j}$ denote the $(i,j)$th element of $A$, $\|A\|_{\mathsf{op}}$ denote the operator or spectral norm of $A$ (which is the largest singular value of $A$), $\sigma_{\mathsf{min}}(A)$ denote the smallest of the $\min\{m,n\}$ singular values of $A$, $\rank(A)$ denote the rank of $A$, $A^{\T} \in \R^{n \times m}$ denote the transpose or adjoint of $A$, $A^{\dagger} \in \R^{n \times m}$ denote the \textit{Moore-Penrose pseudoinverse} of $A$, and $A^{-1} \in \R^{n \times n}$ denote the inverse of $A$ when $m = n$ and $A$ is non-singular. Furthermore, when the rows of $A$ are linearly independent, then $A^{\dagger} = A^{\T} \big(A A^{\T}\big)^{-1}$ is a \textit{right inverse} of $A$ such that $A A^{\dagger} = I$, where $I$ is the identity matrix of appropriate dimension. For any row stochastic matrix $A \in \R^{m \times n}$, we let $\ext(A)$ denote the number of extreme points of the convex hull of the rows of $A$, and it is straightforward to verify that
\begin{equation}
\rank(A) \leq \ext(A) \leq m \, .
\end{equation}
In the sequel, we refer to a row stochastic matrix $A$ as \textit{full rank} if $\rank(A) = \min\{\ext(A),n\}$,\footnote{This is in contrast to standard usage where $A$ is said to be ``full rank'' if $\rank(A) = \min\{m,n\}$. Our alternative usage of the phrase ``full rank'' is motivated by information theoretic contexts, such as in the proof of Theorem \ref{Thm: Converse Bound II} in subsection \ref{Converse Bounds for Strictly Positive DMCs}, where the effective number of rows (or input alphabet) of a row stochastic matrix (or channel) $A$ can often be reduced to $\ext(A)$ due to the convexity of KL divergence. The resulting sub-matrix, which has $\ext(A)$ rows, is full rank in the standard sense when $\rank(A) = \min\{\ext(A),n\}$.} and \textit{strictly positive} if the elements of $A$ are all strictly positive. 

Finally, we present some miscellaneous analysis notation. We let the customary $\left\| \cdot \right\|_{p}$ notation denote the $\mathcal{L}^p$-norm for $p \in [1,\infty]$. We let $\exp(\cdot)$ denote the natural exponential function (with base $e$), and $\lfloor \cdot \rfloor$ denote the floor function. Throughout this paper, we will use the standard \textit{Bachmann-Landau asymptotic notation}, e.g., $O(\cdot)$, $\Theta(\cdot)$, $o(\cdot)$, and $\omega(\cdot)$, with the understanding that the parameter $n \rightarrow \infty$ and all other parameters are held constant with respect to $n$.

\subsection{Outline}
\label{Outline}

In closing section \ref{Introduction}, we briefly delineate the organization of the rest of this paper. In section \ref{Main Results}, we present all of our main results, which were described at the outset of section \ref{Introduction}. In section \ref{Achievability and Converse Bounds}, we prove our main achievability and converse bounds using several auxiliary lemmata. Then, we illustrate several examples of noisy permutation channel capacities for different families of channels in section \ref{Permutation Channel Capacity}. Furthermore, we also establish the connection between the degradation preorder over channels and noisy permutation channel capacity in section \ref{Permutation Channel Capacity}. Finally, we conclude our discussion and propose future research directions in section \ref{Conclusion}. On a separate note, it is worth mentioning that throughout this paper, theorems, propositions, and lemmata are stated according to the following convention: If the result is known in the literature, we provide references in the header, and if the result is new, we (obviously) do not provide any references.

\section{Main Results}
\label{Main Results}

In this section, we present our main results under the setup of subsection \ref{Permutation Channel Model}, very briefly mention the important ideas in the corresponding proofs, and discuss any related literature where appropriate.

\subsection{Achievability Bound}
\label{Achievability Bound}

Our first main result is a lower bound on the noisy permutation channel capacity of any DMC in terms of the rank of the DMC. 

\begin{theorem}[Achievability Bound]
\label{Thm: Achievability Bound}
The noisy permutation channel capacity of a DMC $P_{Z|X}$ is lower bounded by 
$$ \Cperm(P_{Z|X}) \geq \frac{\rank(P_{Z|X}) - 1}{2} \, . $$
\end{theorem}

Theorem \ref{Thm: Achievability Bound} is proved in subsection \ref{Achievability Bounds for DMCs} using a simple (randomized) code which enables a basic concentration of measure inequality based argument. We also present an alternative proof of Theorem \ref{Thm: Achievability Bound} for the special case of row stochastic matrices with rank $2$ in subsection \ref{Achievability Bounds for DMCs}, which employs the so called \textit{second moment method} for total variation distance.

\subsection{Converse Bounds}
\label{Converse Bounds}

Our second main result is an upper bound on the noisy permutation channel capacity of any strictly positive DMC in terms of the output alphabet size of the DMC.

\begin{theorem}[Converse Bound I]
\label{Thm: Converse Bound I}
The noisy permutation channel capacity of a strictly positive DMC $P_{Z|X}$, which means that $P_{Z|X}(y|x) > 0$ for all $x \in \X$ and $y \in \Y$, is upper bounded by
$$ \Cperm(P_{Z|X}) \leq \frac{|\Y| - 1}{2} \, . $$
\end{theorem}

Theorem \ref{Thm: Converse Bound I} is established in subsection \ref{Converse Bounds for Strictly Positive DMCs}. The proof of Theorem \ref{Thm: Converse Bound I} uses a Fano's inequality argument followed by a careful application of a \textit{central limit theorem} (CLT) based approximation of the entropy of a binomial random variable. Intuitively, we also expect to have a converse bound in terms of the input alphabet size, because when $|\X|$ is much smaller than $|\Y|$, there are at most $O\big(n^{|\X|-1}\big)$ distinguishable empirical distributions (rather than $O\big(n^{|\Y|-1}\big)$, as suggested by \eqref{Eq: Number of Emp Dists}). Our third main result addresses this intuition by providing an alternative upper bound on the noisy permutation channel capacity of any strictly positive DMC in terms of the number of extreme points of the convex hull of the conditional probability distributions defining the DMC.

\begin{theorem}[Converse Bound II]
\label{Thm: Converse Bound II}
The noisy permutation channel capacity of a strictly positive DMC $P_{Z|X}$ is upper bounded by 
$$ \Cperm(P_{Z|X}) \leq \frac{\ext(P_{Z|X}) - 1}{2} \, . $$
\end{theorem}

Theorem \ref{Thm: Converse Bound II} is also proved in subsection \ref{Converse Bounds for Strictly Positive DMCs}. Its proof layers a degradation argument, based on Proposition \ref{Prop: Degradation by Symmetric Channels} (which will be presented in due course), over the derivation of Theorem \ref{Thm: Converse Bound I}. We remark that the quantity $\ext(P_{Z|X})$ can be perceived as an ``effective input alphabet'' size. Indeed, as elucidated in the proof of Theorem \ref{Thm: Converse Bound II}, the input alphabet of $P_{Z|X}$ can be reduced to a subset of $\X$ corresponding to the extreme points of the convex hull of the rows of $P_{Z|X}$ without loss of generality (due, essentially, to the convexity of KL divergence).  

Together, the bounds in Theorems \ref{Thm: Converse Bound I} and \ref{Thm: Converse Bound II} yield the following corollary that for any strictly positive DMC $P_{Z|X}$,
\begin{equation}
\label{Eq: Combined Converse Bound}
\Cperm(P_{Z|X}) \leq \frac{\min\{\ext(P_{Z|X}),|\Y|\} - 1}{2} \, .  
\end{equation}
On the other hand, for a general DMC $P_{Z|X}$, which may have zero entries, we can show that
\begin{equation}
\label{Eq: General Combined Converse Bound}
\Cperm(P_{Z|X}) \leq \min\{\ext(P_{Z|X}),|\Y|\} - 1 \, .  
\end{equation}
To see this, note that the bound $\Cperm(P_{Z|X}) \leq |\Y| - 1$ is already intuitively justified by \eqref{Eq: Number of Emp Dists}, and a rigorous argument follows along the same lines as the converse proof in subsection \ref{Permutation Transition Matrices}. Moreover, the bound $\Cperm(P_{Z|X}) \leq \ext(P_{Z|X}) - 1$ can be established by following the proof of Theorem \ref{Thm: Converse Bound II} in subsection \ref{Converse Bounds for Strictly Positive DMCs}. (Indeed, the derivation of \eqref{Eq: Fano step 3} in this proof also holds for DMCs $P_{Z|X}$ with zero entries, in which case, $P_{\tilde{Z}|\tilde{X}}$ is the identity channel. The converse proof in subsection \ref{Permutation Transition Matrices} can then be applied to yield the desired bound.) We omit these proofs for the sake of brevity.

\subsection{Strictly Positive and Full Rank Channels}

Theorem \ref{Thm: Achievability Bound} and \eqref{Eq: Combined Converse Bound} portray that for any strictly positive DMC $P_{Z|X}$, the noisy permutation channel capacity satisfies the bounds 
\begin{equation}
\label{Eq: Combined Bounds}
\begin{aligned}
\frac{\rank(P_{Z|X}) - 1}{2} & \leq \Cperm(P_{Z|X}) \\
& \leq \frac{\min\{\ext(P_{Z|X}),|\Y|\} - 1}{2} \, .
\end{aligned}
\end{equation}
Based on the inequalities in \eqref{Eq: Combined Bounds}, we now state (perhaps) the most important result of this paper, which characterizes the noisy permutation channel capacity of the family of strictly positive and full rank channels. 

\begin{theorem}[$\Cperm$ of Strictly Positive and Full Rank Channels]
\label{Thm: Strictly Positive and Full Rank Channels}
The noisy permutation channel capacity of a strictly positive and full rank DMC $P_{Z|X}$ with rank $r \triangleq \rank(P_{Z|X}) = \min\{\ext(P_{Z|X}),|\Y|\}$ is given by
$$ \Cperm(P_{Z|X}) = \frac{r - 1}{2} \, . $$
\end{theorem}

\begin{proof}
Recalling the definition of ``full rank'' from subsection \ref{Additional Notation}, this is an immediate corollary of \eqref{Eq: Combined Bounds} (i.e., of Theorems \ref{Thm: Achievability Bound}, \ref{Thm: Converse Bound I}, and \ref{Thm: Converse Bound II}).
\end{proof}

\subsection{Degradation and Noisy Permutation Channel Capacity}
\label{Degradation and Permutation Channel Capacity}

To complement the aforementioned results, we next present another main result that relates the notion of noisy permutation channel capacity with the so called (output) \textit{degradation} preorder over channels, which was defined in information theory to study broadcast channels in \cite{Cover1972, Bergmans1973}. (It is worth mentioning that in this paper, we are concerned with the notion of \textit{stochastic} degradation as opposed to \textit{physical} degradation, cf. \cite[Section 5.4]{ElGamalKim2011}.)

\begin{definition}[Degradation Preorder]
\label{Def: Degradation Preorder}
For any two DMCs (or row stochastic matrices) $P_{Z_1|X} \in \R^{|\X| \times |\mathcal{Z}_1|}$ and $P_{Z_2|X} \in \R^{|\X| \times |\mathcal{Z}_2|}$ with common input alphabet $\X$ and output alphabets $\mathcal{Z}_1$ and $\mathcal{Z}_2$, respectively, we say that $P_{Z_2|X}$ is a \emph{degraded} version of $P_{Z_1|X}$ if $P_{Z_2|X} = P_{Z_1|X} P_{Z_2|Z_1}$ for some channel $P_{Z_2|Z_1} \in \R^{|\mathcal{Z}_1| \times |\mathcal{Z}_2|}$. 
\end{definition}

The degradation preorder has a long and intriguing history that is worth elaborating on. Its study actually originated in the statistics literature \cite{Blackwell1951,Sherman1951,Stein1951}, where it is also known as the \textit{Blackwell order}. Indeed, the channels $P_{Z_1|X}$ and $P_{Z_2|X}$ can be construed as statistical experiments (or observation models) of the parameter space $\X$. In this statistical decision theoretic context, the celebrated \textit{Blackwell-Sherman-Stein theorem} states that $P_{Z_2|X}$ is a degraded version of $P_{Z_1|X}$ if and only if for every prior distribution $P_X \in \Simplex_{\X}$, and every real-valued loss function with domain $\X \times \X$, the minimum Bayes risk corresponding to $P_{Z_1|X}$ is less than or equal to the minimum Bayes risk corresponding to $P_{Z_2|X}$ \cite{Blackwell1951,Sherman1951,Stein1951} (also see \cite{LeshnoSpector1992} for a simple proof of this result using the separating hyperplane theorem). Furthermore, degradation has beautiful ties with non-Bayesian binary hypothesis testing as well. When $|\X| = 2$, the channels $P_{Z_1|X}$ and $P_{Z_2|X}$ can be construed as \textit{dichotomies} of likelihoods, and it can be shown that $P_{Z_2|X}$ is a degraded version of $P_{Z_1|X}$ if and only if the \textit{Neyman-Pearson function}, or receiver operating characteristic curve, of $P_{Z_1|X}$ dominates the Neyman-Pearson function of $P_{Z_2|X}$ pointwise (cf. \cite[Theorem 5.3]{Torgersen1991} and \cite[Section 9.3]{Torgersen1991book}, where equivalent characterizations using $f$-divergences and majorization are also given). Moreover, for the special case where $P_{Z_1|X}$ and $P_{Z_2|X}$ are binary input \textit{symmetric} channels, other majorization and stochastic domination based characterizations of degradation can be found in \cite[Sections 4.1.14--4.1.16]{RichardsonUrbanke2008}. Finally, we note that degradation is also equivalent to the notion of \textit{matrix majorization} in \cite[Chapter 15, Definition C.8]{MarshallOlkinArnold2011} (also see \cite{Dahl1999a} and \cite{Dahl1999b}). We refer readers to the author's doctoral thesis \cite[Section 3.1.1]{Makur2019} and \cite[Section I-B]{MakurPolyanskiy2018} for further discussion and references.

The next theorem conveys an intuitive comparison result that if one DMC dominates another DMC in the degradation sense, then the noisy permutation channel capacity of the dominating DMC is larger than the noisy permutation channel capacity of the degraded DMC.   

\begin{theorem}[Comparison Bound via Degradation]
\label{Thm: Comparison Bound via Degradation}
Consider any two DMCs $P_{Z_1|X} \in \R^{|\X| \times |\mathcal{Z}_1|}$ and $P_{Z_2|X} \in \R^{|\X| \times |\mathcal{Z}_2|}$, with common input alphabet $\X$ and output alphabets $\mathcal{Z}_1$ and $\mathcal{Z}_2$, respectively. If $P_{Z_2|X}$ is a degraded version of $P_{Z_1|X}$, then we have
$$ \Cperm(P_{Z_2|X}) \leq \Cperm(P_{Z_1|X}) \, . $$
\end{theorem}

Theorem \ref{Thm: Comparison Bound via Degradation} is derived in subsection \ref{Erasure Channels}. As with the setting of traditional channel capacity, the proof of Theorem \ref{Thm: Comparison Bound via Degradation} proceeds by verifying that a noisy permutation channel capacity achieving coding scheme for $P_{Z_2|X}$ can be used to achieve the same rate and vanishing probability of error when communicating through $P_{Z_1|X}$. Furthermore, a specialization of Theorem \ref{Thm: Comparison Bound via Degradation} for erasure channels turns out to correspond to the concept of Doeblin minorization, and we use this connection in subsection \ref{Erasure Channels} to provide an alternative achievability bound on the noisy permutation channel capacity of erasure channels.

Lastly, while we are on the topic of degradation, we present another seemingly disparate result which constructs symmetric channels that dominate given DMCs in the degradation sense. To state this result, we first recall the definition of symmetric channels, cf. \cite[Equation (10)]{MakurPolyanskiy2018}. 

\begin{definition}[$q$-ary Symmetric Channel]
\label{Def: Symmetric Channel}
Under the formalism presented in subsection \ref{Permutation Channel Model}, we define a \emph{$q$-ary symmetric channel} with total crossover probability $\delta \in [0,1]$, and input and output alphabet $\X = \Y$ with $|\X| = q \in \N \backslash \!\{1\}$, denoted $\qSC(\delta)$, using the doubly stochastic matrix
\begin{equation}
\label{Eq: q-SC Matrix}
S_{\delta} \triangleq \left[
\begin{array}{ccccc}
1-\delta & \frac{\delta}{q-1} & \cdots & \frac{\delta}{q-1} & \frac{\delta}{q-1} \\
\frac{\delta}{q-1} & 1-\delta & \cdots & \frac{\delta}{q-1} & \frac{\delta}{q-1} \\
\vdots & \vdots & \ddots & \vdots & \vdots \\
\frac{\delta}{q-1} & \frac{\delta}{q-1} & \cdots & 1-\delta & \frac{\delta}{q-1} \\
\frac{\delta}{q-1} & \frac{\delta}{q-1} & \cdots & \frac{\delta}{q-1} & 1 - \delta
\end{array} \right] \in \R^{q \times q}
\end{equation} 
which has $1-\delta$ along its principal diagonal, and $\frac{\delta}{q-1}$ in all other entries. (The rows and columns of $S_{\delta}$ are both indexed consistently by $\X$.)
\end{definition}

We note that in the special case where $q = 2$, $\X = \Y = \{0,1\}$, and $\delta$ is the probability that the input bit flips, we refer to the $2\text{-}\mathsf{SC}(\delta)$ as a \textit{binary symmetric channel} (BSC), denoted $\BSC(\delta)$.

The ensuing proposition portrays a sufficient condition for degradation by $q$-ary symmetric channels.  

\begin{proposition}[Degradation by Symmetric Channels]
\label{Prop: Degradation by Symmetric Channels}
Suppose we are given a DMC (or row stochastic matrix) $P_{Z|X} \in \R^{|\X| \times |\Y|}$ with minimum entry
$$ \nu = \min_{x \in \X, \, y \in \Y}{P_{Z|X}(y|x)} \, , $$ 
and a $q$-ary symmetric channel, $\qSC(\delta)$, which has a common input alphabet $\X$ such that $|\X| = q$. If the total crossover probability parameter satisfies
$$ 0 \leq \delta \leq \frac{\nu}{1 - \nu + \frac{\nu}{q-1}} \, , $$
then $P_{Z|X}$ is a degraded version of $\qSC(\delta)$. 
\end{proposition}

\begin{figure*}[ht]
\centering
\includegraphics[trim = 0mm 100mm 0mm 80mm, clip, width=\linewidth]{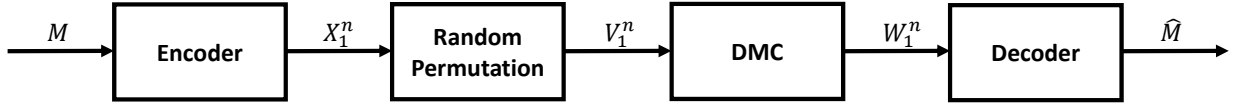} 
\caption{Illustration of a communication system with a random permutation followed by a DMC.}
\label{Figure: Permutation Channel 2}
\end{figure*}

Proposition \ref{Prop: Degradation by Symmetric Channels} is proved in appendix \ref{Proof of Proposition Degradation by Symmetric Channels}. Although it appears to be unrelated to our thrust towards understanding noisy permutation channel capacity, it turns out to be indispensable in the proof of Theorem \ref{Thm: Converse Bound II}. We state Proposition \ref{Prop: Degradation by Symmetric Channels} here as a main result because we believe it can have many applications in information theory and statistics beyond the context of noisy permutation channels. We refer readers to \cite{MakurPolyanskiy2017,MakurPolyanskiy2018} for further insight regarding the value of studying channel domination by symmetric channels. It is also worth making a few remarks about related results in the literature. Indeed, Proposition \ref{Prop: Degradation by Symmetric Channels} establishes a result analogous to \cite[Theorem 2]{MakurPolyanskiy2018} (also see \cite[Theorem 2]{MakurPolyanskiy2017}) that holds for general rectangular row stochastic matrices $P_{Z|X}$ (rather than square row stochastic matrices as in \cite[Theorem 2]{MakurPolyanskiy2018}). However, Proposition \ref{Prop: Degradation by Symmetric Channels} is weaker than \cite[Theorem 2]{MakurPolyanskiy2018} for square row stochastic matrices (i.e., the upper bound on $\delta$ in \cite[Theorem 2]{MakurPolyanskiy2018} is larger than that in Proposition \ref{Prop: Degradation by Symmetric Channels} when $q > 2$), because the proof of \cite[Theorem 2]{MakurPolyanskiy2018} exploits more sophisticated majorization arguments. We also remark that other sufficient conditions for degradation of square row stochastic matrices (or Markov kernels) by $q$-ary symmetric channels, which either use more information than the minimum entries of the matrices (see \cite[Proposition 8.1, Equations (8.1) and (8.2)]{MosselOleszkiewiczSen2013}), or assume further structure on the matrices such as additive noise over Abelian groups (see \cite[Theorem 3, Proposition 10]{MakurPolyanskiy2018} or \cite[Theorem 3]{MakurPolyanskiy2017}), have been derived in the literature.

\section{Achievability and Converse Bounds}
\label{Achievability and Converse Bounds}

We prove the achievability result in Theorem \ref{Thm: Achievability Bound} and the converse results in Theorems \ref{Thm: Converse Bound I} and \ref{Thm: Converse Bound II} in this section. We commence by presenting some useful lemmata in subsection \ref{Auxiliary Lemmata}, and then proceed to establishing the aforementioned theorems in subsections \ref{Achievability Bounds for DMCs} and \ref{Converse Bounds for Strictly Positive DMCs}, respectively. 

\subsection{Auxiliary Lemmata}
\label{Auxiliary Lemmata}

First, to establish our converse bounds in Theorems \ref{Thm: Converse Bound I} and \ref{Thm: Converse Bound II}, we will present two lemmata. The first lemma we will exploit is the following useful estimate of the entropy of a binomial distribution from the literature.

\begin{lemma}[Approximation of Binomial Entropy {\cite[Equation (7)]{AdellLekuonaYu2010}}]
\label{Lemma: Approximation of Binomial Entropy}
Given a binomial random variable $X \sim \mathsf{bin}(n,p)$ with $n \in \N$ and $p \in (0,1)$, we have
$$ \left|H(X) - \frac{1}{2} \log(2\pi e n p (1-p))\right| \leq \frac{c(p)}{n} $$
for some constant $c(p) \geq 0$ (that depends on $p$).
\end{lemma}

The second lemma we will utilize illustrates that swapping the order of the DMC and the random permutation block in the communication system in Figure \ref{Figure: Permutation Channel} produces the statistically equivalent communication system in Figure \ref{Figure: Permutation Channel 2}. 

\begin{lemma}[Equivalent Model]
\label{Lemma: Equivalent Model}
Consider the channel $P_{W_1^n|X_1^n}$ shown in Figure \ref{Figure: Permutation Channel 2}, where the input codeword $X_1^n \in \X^n$ passes through an independent random permutation to produce $V_1^n \in \X^n$, and $V_1^n$ then passes through a DMC $P_{W|V}$ to produce the output codeword $W_1^n \in \Y^n$ so that (much like \eqref{Eq: Memorylessness})
$$ \forall v_1^n \in \X^n, \, w_1^n \in \Y^n, \enspace P_{W_1^n | V_1^n}(w_1^n | v_1^n) = \prod_{i = 1}^{n}{P_{W|V}(w_i|v_i)} \, . $$
If the DMC $P_{W|V}$ is equal to the DMC $P_{Z|X}$ entry-wise, i.e.,
\begin{equation}
\label{Eq: Equivalence Condition}
\forall x \in \X, \, z \in \Y, \enspace P_{W|V}(z|x) = P_{Z|X}(z|x) \, , 
\end{equation}
then the channel $P_{W_1^n|X_1^n}$ is equivalent to the channel $P_{Y_1^n|X_1^n}$ (described in subsection \ref{Permutation Channel Model} and Figure \ref{Figure: Permutation Channel}), i.e.,
$$ \forall x_1^n \in \X^n, \, y_1^n \in \Y^n, \enspace P_{W_1^n | X_1^n}(y_1^n | x_1^n) = P_{Y_1^n | X_1^n}(y_1^n | x_1^n) \, . $$
\end{lemma}

\begin{proof}
This follows from direct calculation. Fix any $x_1^n \in \X^n$ and $y_1^n \in \Y^n$. Observe that
\begin{align}
P_{W_1^n | X_1^n}(y_1^n | x_1^n) & = \sum_{\substack{v_1^n \in \X^n : \\ \hat{P}_{v_1^n} = \hat{P}_{x_1^n}}}{P_{W_1^n | V_1^n}(y_1^n | v_1^n) P_{V_1^n | X_1^n}(v_1^n | x_1^n)} \nonumber \\
& = \binom{n}{n \hat{P}_{x_1^n}}^{\! -1} \sum_{\substack{v_1^n \in \X^n : \\ \hat{P}_{v_1^n} = \hat{P}_{x_1^n}}}{\prod_{i = 1}^{n}{P_{W|V}(y_i|v_i)}} \nonumber \\
& = \binom{n}{n \hat{P}_{x_1^n}}^{\! -1} \sum_{\substack{v_1^n \in \X^n : \\ \hat{P}_{v_1^n} = \hat{P}_{x_1^n}}}{\prod_{i = 1}^{n}{P_{Z|X}(y_i|v_i)}} \nonumber \\
& = \binom{n}{n \hat{P}_{y_1^n}}^{\! -1} \binom{n}{n \hat{P}_{x_1^n}}^{\! -1} \cdot \nonumber \\
& \enspace \sum_{\substack{\tilde{y}_1^n \in \Y^n : \\ \hat{P}_{\tilde{y}_1^n} = \hat{P}_{y_1^n}}}{\sum_{\substack{v_1^n \in \X^n : \\ \hat{P}_{v_1^n} = \hat{P}_{x_1^n}}}{\prod_{i = 1}^{n}{P_{Z|X}(\tilde{y}_i|v_i)}}} \, ,
\label{Eq: Product Channel Form 1}
\end{align}
where the first equality uses the Markov property $X_1^n \rightarrow V_1^n \rightarrow W_1^n$, the third equality follows from \eqref{Eq: Equivalence Condition}, and the fourth equality holds because
$$ P_{W_1^n|X_1^n}(y_1^n|x_1^n) = P_{W_1^n|X_1^n}(\tilde{y}_1^n|x_1^n) $$
for every $\tilde{y}_1^n \in \Y^n$ that is a permutation of $y_1^n$, which follows from the expression in the third equality. Likewise, we have
\begin{align}
P_{Y_1^n | X_1^n}(y_1^n | x_1^n) & = \sum_{\substack{z_1^n \in \Y^n : \\ \hat{P}_{z_1^n} = \hat{P}_{y_1^n}}}{P_{Y_1^n | Z_1^n}(y_1^n | z_1^n) P_{Z_1^n | X_1^n}(z_1^n | x_1^n)} \nonumber \\
& = \binom{n}{n \hat{P}_{y_1^n}}^{\! -1} \sum_{\substack{z_1^n \in \Y^n : \\ \hat{P}_{z_1^n} = \hat{P}_{y_1^n}}}{\prod_{i = 1}^{n}{P_{Z|X}(z_i|x_i)}} \nonumber \\
& = \binom{n}{n \hat{P}_{x_1^n}}^{\! -1} \binom{n}{n \hat{P}_{y_1^n}}^{\! -1} \cdot \nonumber \\
& \enspace \sum_{\substack{\tilde{x}_1^n \in \X^n: \\ \hat{P}_{\tilde{x}_1^n} = \hat{P}_{x_1^n}}}{\sum_{\substack{z_1^n \in \Y^n : \\ \hat{P}_{z_1^n} = \hat{P}_{y_1^n}}}{\prod_{i = 1}^{n}{P_{Z|X}(z_i|\tilde{x}_i)}}}\,,
\label{Eq: Product Channel Form 2}
\end{align}
where the first equality uses the Markov property $X_1^n \rightarrow Z_1^n \rightarrow Y_1^n$, and the third equality holds because
$$ P_{Y_1^n|X_1^n}(y_1^n|x_1^n) = P_{Y_1^n|X_1^n}(y_1^n|\tilde{x}_1^n) $$
for every $\tilde{x}_1^n \in \X^n$ that is a permutation of $x_1^n$, which follows from the expression in the second equality. Therefore, using \eqref{Eq: Product Channel Form 1} and \eqref{Eq: Product Channel Form 2}, we have
$$ P_{W_1^n|X_1^n}(y_1^n|x_1^n) = P_{Y_1^n|X_1^n}(y_1^n|x_1^n) \, , $$
which completes the proof. 
\end{proof} 

Next, to derive our achievability bound in Theorem \ref{Thm: Achievability Bound}, we will require the following well-known concentration of measure inequality, which is a specialization of \textit{Hoeffding's inequality}.

\begin{lemma}[Hoeffding's Inequality {\cite[Theorems 1 and 2]{Hoeffding1963}}]
\label{Lemma: Hoeffding's Inequality}
Suppose $X_1,\dots,X_n$ are independent and identically distributed (i.i.d.) random variables such that $|X_1| \leq \sigma$ almost surely for some $\sigma > 0$. Then, for every $\gamma \geq 0$,
$$ \P\!\left(\frac{1}{n}\sum_{i = 1}^{n}{X_i} - \E\!\left[X_1\right] \geq \gamma \right) \leq \exp\!\left(-\frac{n \gamma^2}{2 \sigma^2}\right) $$
and
$$ \P\!\left(\frac{1}{n}\sum_{i = 1}^{n}{X_i} - \E\!\left[X_1\right] \leq -\gamma \right) \leq \exp\!\left(-\frac{n \gamma^2}{2 \sigma^2}\right) . $$
\end{lemma}

While Lemma \ref{Lemma: Hoeffding's Inequality} is used to provide exponentially decaying tail bounds on certain conditional probability of error terms in the proof of Theorem \ref{Thm: Achievability Bound} (see \eqref{Eq: Hoeffding Bound 1} and \eqref{Eq: Hoeffding Bound 2} in subsection \ref{Achievability Bounds for DMCs}), we will also show that much weaker tail bounds suffice for proving Theorem \ref{Thm: Achievability Bound} for DMCs with rank $2$. Indeed, our alternative achievability proof of Proposition \ref{Prop: Achievability Bound for DMCs with Rank 2} in subsection \ref{Achievability Bounds for DMCs} uses the two ensuing lemmata pertaining to the following \textit{binary hypothesis testing} problem. 

Fix any $n \in \N$, and two distinct probability distributions $P_X,Q_X \in \Simplex_{\X}$ (which can depend on $n$). Consider the hypothesis random variable $H \sim \Ber\big(\frac{1}{2}\big)$ (i.e., uniform prior), and likelihoods $P_{X|H}(\cdot|0) = P_X(\cdot)$ and $P_{X|H}(\cdot|1) = Q_X(\cdot)$, such that we observe $n$ samples $X_1^n$ that are drawn conditionally i.i.d. given $H$ from the likelihoods, viz.,
\begin{equation}
\label{Eq: BHT Problem}
\begin{aligned}
\text{Given } H = 0 & : X_1^n \stackrel{\text{i.i.d.}}{\sim} P_{X} \, , \\
\text{Given } H = 1 & : X_1^n \stackrel{\text{i.i.d.}}{\sim} Q_{X} \, .
\end{aligned}
\end{equation}
The (classical) objective of binary hypothesis testing is to decode the hypothesis $H$ with minimum probability of error from the observed samples $X_1^n$. It is well-known that the \textit{maximum likelihood (ML) decision rule} for $H$ based on $X_1^n$, $\hat{H}_{\mathsf{ML}}^n : \X^n \rightarrow \{0,1\}$, which is defined by
\begin{equation}
\forall x_1^n \in \X^n, \enspace \hat{H}_{\mathsf{ML}}^n(x_1^n) = \argmax_{h \in \{0,1\}}{\prod_{i=1}^{n}{P_{X|H}(x_i|h)}} \, ,
\end{equation}
or equivalently,
\begin{equation}
\label{Eq: ML decoder}
D(\hat{P}_{x_1^n}||Q_X) \quad \substack{\hat{H}_{\mathsf{ML}}^n(x_1^n) \, = \, 0 \\  \gtreqless \\ \hat{H}_{\mathsf{ML}}^n(x_1^n) \, = \, 1} \quad D(\hat{P}_{x_1^n}||P_X) \, ,
\end{equation}
achieves the minimum probability of error
\begin{equation}
P_{\mathsf{ML}}^{(n)} \triangleq \P\!\left(\hat{H}_{\mathsf{ML}}^n(X_1^n) \neq H\right) ,
\end{equation}
where the tie-breaking rule in \eqref{Eq: ML decoder} (when the likelihoods of $0$ and $1$ are equal) does not affect $P_{\mathsf{ML}}^{(n)}$ (see, e.g., \cite[Chapter 2]{Wornell2017}). Furthermore, \textit{Le Cam's relation} states that the ML decoding probability of error is completely characterized by the \textit{total variation (TV) distance} between the two likelihoods, cf. \cite[proof of Theorem 2.2(i)]{Tsybakov2009}. Recall that the TV distance between two probability measures $P_0$ and $P_1$ on a common measurable space $(\mathcal{U},\mathscr{F})$ is defined as
\begin{align}
\left\|P_0 - P_1\right\|_{\mathsf{TV}} & \triangleq \sup_{\mathcal{A} \in \mathscr{F}}{|P_0(\mathcal{A}) - P_1(\mathcal{A})|} \\
& = \frac{\left\|P_0 - P_1\right\|_{1}}{2} \, ,
\label{Eq: ell^1 norm char}
\end{align}
where \eqref{Eq: ell^1 norm char} is well-known (see, e.g., \cite[Chapter 4]{LevinPeresWilmer2009} for a proof in the discrete case). Then, we have
\begin{equation}
\label{Eq: Le Cam relation}
P_{\mathsf{ML}}^{(n)} = \frac{1}{2}\left(1 - \left\|P_X^{\otimes n}  - Q_X^{\otimes n} \right\|_{\mathsf{TV}}\right) ,
\end{equation}
where $P_X^{\otimes n}$ and $Q_X^{\otimes n}$ denote the $n$-fold product distributions of $X_1^n$ given $H = 0$ and $H = 1$, respectively. The next lemma presents a vector generalization of the so called ``second moment method for TV distance,'' cf. \cite[Lemma 4.2(iii)]{Evansetal2000}, and lower bounds $\big\|P_X^{\otimes n}  - Q_X^{\otimes n} \big\|_{\mathsf{TV}}$.

\begin{lemma}[Second Moment Method]
\label{Lemma: Second Moment Method}
For the binary hypothesis testing problem in \eqref{Eq: BHT Problem}, we have
$$ \left\|P_X^{\otimes n}  - Q_X^{\otimes n} \right\|_{\mathsf{TV}} \geq \frac{\left\|P_X - Q_X\right\|_{2}^2}{\displaystyle{4 \sum_{x \in \X}{\VAR\big(\hat{P}_{X_1^n}(x)\big)}}} \, . $$
\end{lemma}

Lemma \ref{Lemma: Second Moment Method} is proved in appendix \ref{Proof of Lemma Second Moment Method}. Moreover, as mentioned in the remark in appendix \ref{Proof of Lemma Second Moment Method}, Lemma \ref{Lemma: Second Moment Method} can also be construed as a variant of the \textit{Hammersley-Chapman-Robbins (HCR) bound} in statistics \cite{Hammersley1950,ChapmanRobbins1951}.

Our final lemma, Lemma \ref{Lemma: Testing between Converging Hypotheses}, establishes an upper bound on $P_{\mathsf{ML}}^{(n)}$ using Lemma \ref{Lemma: Second Moment Method}. It will be used to derive Proposition \ref{Prop: Achievability Bound for DMCs with Rank 2} in subsection \ref{Achievability Bounds for DMCs}\textemdash a specialization of Theorem \ref{Thm: Achievability Bound} for DMCs with rank $2$.

\begin{lemma}[Testing between Converging Hypotheses]
\label{Lemma: Testing between Converging Hypotheses} 
For the binary hypothesis problem in \eqref{Eq: BHT Problem}, suppose the $\ell^2$-distance between $P_X$ and $Q_X$ is lower bounded by
\begin{equation}
\label{Eq: 2-norm condition}
\left\|P_X - Q_X\right\|_2 \geq \frac{1}{n^{\frac{1}{2} - \epsilon_n}} 
\end{equation}
for some constant $\epsilon_n \in \big(0,\frac{1}{2}\big)$ (which may depend on $n$). Then, we have
$$ P_{\mathsf{ML}}^{(n)} \leq \frac{|\X|}{2|\X| + 2 n^{2\epsilon_n}} \, , $$
which implies that $\displaystyle{\lim_{n \rightarrow \infty}{P_{\mathsf{ML}}^{(n)}} = 0}$ when $\displaystyle{\lim_{n \rightarrow 
\infty}{n^{\epsilon_n}} = +\infty}$.
\end{lemma}

Lemma \ref{Lemma: Testing between Converging Hypotheses} is established in appendix \ref{Proof of Lemma Testing between Converging Hypotheses}. It illustrates that as long as the Euclidean distance between the likelihoods $P_{X|H}(\cdot|0)$ and $P_{X|H}(\cdot|1)$ vanishes slower than $\Theta(1/\!\sqrt{n})$, we can decode the hypothesis $H$ with vanishing probability of error as $n \rightarrow \infty$. Intuitively, when $\|P_X - Q_X\|_2 = \Theta\big(1/n^{\frac{1}{2}-\epsilon_n}\big)$ and we neglect $\epsilon_n$, Lemma \ref{Lemma: Testing between Converging Hypotheses} holds because the sum of the variances of the entries of the sufficient statistic $T_n = \hat{P}_{X_1^n} - \frac{1}{2}P_X - \frac{1}{2}Q_X$ (defined in \eqref{Eq: SS} and \eqref{Eq: Constant Shift} in the proof of Lemma \ref{Lemma: Second Moment Method} in appendix \ref{Proof of Lemma Second Moment Method}) is $O(1/n)$. So, as long as the Euclidean distance between the two likelihoods is $\omega(1/\!\sqrt{n})$, it is possible to distinguish between the two hypotheses. We also remark that tighter upper bounds on $P_{\mathsf{ML}}^{(n)}$ can be obtained using standard exponential concentration of measure inequalities. However, the simpler second moment method approach will suffice for our proof of Proposition \ref{Prop: Achievability Bound for DMCs with Rank 2} (while our proof of Theorem \ref{Thm: Achievability Bound} will in fact use stronger concentration bounds).

\subsection{Achievability Bounds for DMCs}
\label{Achievability Bounds for DMCs}

In this subsection, we first prove our main achievability result in Theorem \ref{Thm: Achievability Bound} and then provide an alternative proof for DMCs with rank $2$. Recall the formalism of subsection \ref{Permutation Channel Model}, which describes the noisy permutation channel model with a DMC $P_{Z|X}$.

\begin{proof}[Proof of Theorem \ref{Thm: Achievability Bound}] 
Since the lower bound in Theorem \ref{Thm: Achievability Bound} trivially holds for the case $\rank(P_{Z|X}) = 1$, we assume without loss of generality that $r \triangleq \rank(P_{Z|X}) \geq 2$. Let $\X^{\prime} \subseteq \X$ denote any (fixed) subset of $\X$ such that $|\X^{\prime}| = r$ and the set of conditional distributions $\{P_{Z|X}(\cdot|x) \in \Simplex_{\Y} : x \in \X^{\prime}\}$ are linearly independent (as vectors in $\R^{|\Y|}$),\footnote{This implies that the extreme points of the convex hull of $\{P_{Z|X}(\cdot|x) \in \Simplex_{\Y} : x \in \X^{\prime}\}$ are precisely $\{P_{Z|X}(\cdot|x) \in \Simplex_{\Y} : x \in \X^{\prime}\}$.} and let $\tilde{P}_{Z|X} \in \R^{r \times |\Y|}$ denote the row stochastic matrix whose rows are given by $\{P_{Z|X}(\cdot|x)\in \Simplex_{\Y} : x \in \X^{\prime}\}$. Furthermore, define
\begin{equation}
\Simplex_{r,k} \triangleq \left\{\left(\frac{p_1}{k},\dots,\frac{p_r}{k}\right) : p_1,\dots,p_r \in \N\cup\!\{0\}, \, \sum_{i = 1}^{r}{p_i} = k \right\} 
\end{equation} 
as the intersection of the scaled integer lattice $\frac{1}{k} \Z^r$ and the probability simplex in $\R^{r}$, where $k \in \N$ is some large constant. Under the setup of subsection \ref{Permutation Channel Model}, for any $\epsilon \in \big(0,\frac{1}{2}\big)$, consider the following message set and encoder-decoder pair:
\begin{enumerate}
\item The message set $\M = \Simplex_{r,k}$ with $k = {\big \lfloor} n^{\frac{1}{2}-\epsilon} {\big \rfloor}$ so that the cardinality of $\M$ is 
\begin{equation}
\label{Eq: Order for Permutation Channel} 
|\M| = \binom{k + r - 1}{r - 1} = \Theta\!\left(n^{\frac{r-1}{2}-\epsilon(r-1)}\right) ,
\end{equation}
where the elements of $\M$ have been re-indexed for convenience.
\item The randomized encoder $f_n : \Simplex_{r,k} \rightarrow \X^n$ is given by
\begin{equation}
\label{Eq: Randomized Encoder 1}
\forall p = \left(\frac{p_1}{k},\dots,\frac{p_r}{k}\right) \in \Simplex_{r,k}, \enspace f_n(p) = X_1^n \stackrel{\text{i.i.d.}}{\sim} P_X\, , 
\end{equation}
where $X_1^n$ are i.i.d. according to a probability distribution $P_X \in \Simplex_{\X}$ such that
\begin{equation}
\label{Eq: Input Message Representation}
P_X(x) = 
\begin{cases}
\frac{p_x}{k}, & \text{for } x \in \X^{\prime} \\ 
0, & \text{for } x \in \X\backslash\X^{\prime}
\end{cases}
\end{equation}
and $\X^{\prime} =  \{1,\dots,r\}$ without loss of generality.
\item Instead of the ML decoder which achieves minimum probability of error, consider the (sub-optimal) element-wise thresholding decoder $g_n : \Y^n \rightarrow \Simplex_{r,k} \cup \{\mathsf{e}\}$ defined by
\begin{equation}
\label{Eq: Decoder 1}
g_n(y_1^n) = 
\begin{cases}
\big(\frac{\hat{p}_1}{k},\dots,\frac{\hat{p}_r}{k}\big), & \text{if } \big(\frac{\hat{p}_1}{k},\dots,\frac{\hat{p}_r}{k}\big) \in \Simplex_{r,k} \\ 
\, \mathsf{e}, & \text{otherwise} 
\end{cases} 
\end{equation}
for every $y_1^n \in \Y^n$, where for each $x \in \X^{\prime}$,
\begin{equation}
\label{Eq: Element-wise Decoder}
\hat{p}_x = \argmin_{j \in \{0,\dots,k\}}{\left|\sum_{y \in \Y}{\hat{P}_{y_1^n}(y) \left[\tilde{P}_{Z|X}^{\dagger}\right]_{y,x}} - \frac{j}{k}\right|} \, ,
\end{equation}
where we choose a minimizer randomly when there are several. 
\end{enumerate}
This encoder-decoder pair completely specifies the communication system model in subsection \ref{Permutation Channel Model}. Intuitively, this decoder performs reasonably well because $\tilde{P}_{Z|X}^{\dagger}$ is a valid right inverse of $\tilde{P}_{Z|X}$ (since the rows of $\tilde{P}_{Z|X}$ are linearly independent). Indeed, conditioned on sending a particular message, $\hat{P}_{Y_1^n}$ is ``close'' to $P_Z$ (which is the true distribution of the $Y_i$'s as shown below) with high probability when $n$ is large. So, $\sum_{y \in \Y}{\hat{P}_{Y_1^n}(y) \big[\tilde{P}_{Z|X}^{\dagger}\big]_{y,x}}$ is ``close'' to the true $P_X(x)$ for all $x \in \X^{\prime}$ with high probability. We now analyze the average probability of error for this coding scheme.

Let us condition on the event $\{M = p\}$ for some $p = \big(\frac{p_1}{k},\dots,\frac{p_r}{k}\big) \in \Simplex_{r,k}$. Then, we have
$$ X_1^n \stackrel{\text{i.i.d.}}{\sim} P_X \, , \quad Z_1^n \stackrel{\text{i.i.d.}}{\sim} P_Z \, , \quad Y_1^n \stackrel{\text{i.i.d.}}{\sim} P_Z \, , $$
where $P_Z$ denotes the output distribution when $P_X$, defined via \eqref{Eq: Input Message Representation}, is ``pushed forward'' through the channel $P_{Z|X}$, $Z_1^n$ are i.i.d. because the channel $P_{Z|X}$ is memoryless, and $Y_1^n$ are i.i.d. because they are the output of passing $Z_1^n$ through an independent random permutation. Let $\P_{p}$ represent the underlying probability measure after conditioning on $\{M = p\}$. The conditional probability that our element-wise thresholding decoder makes an error is upper bounded by
\begin{align}
& \P_p\big(\hat{M} \neq M\big) \nonumber \\
& \overset{\eqmakebox[B][c]{}}{=} \P_p\big(g_n(Y_1^n) \neq p\big) \nonumber \\
& \overset{\eqmakebox[B][c]{}}{=} \P_p\big(\exists\, x \in \X^{\prime}, \, \hat{p}_x \neq p_x\big) \nonumber \\
& \overset{\eqmakebox[B][c]{\footnotesize (a)}}{\leq} \sum_{x \in \X^{\prime}}{\P_p\big(\hat{p}_x \in \{0,\dots,k\}\!\backslash\!\{p_x\}\big)} \nonumber \\
& \overset{\eqmakebox[B][c]{\footnotesize (b)}}{=} \sum_{x \in \X^{\prime}}\!{\left(\sum_{j > p_x}{\P_p(\hat{p}_x = j)} + \sum_{j < p_x}{\P_p(\hat{p}_x = j)}\right)} \nonumber \\
& \overset{\eqmakebox[B][c]{\footnotesize (c)}}{\leq} \sum_{x \in \X^{\prime}}{\sum_{j > p_x}{\P_p\!\left(\sum_{y \in \Y}{\hat{P}_{Y_1^n}(y) \left[\tilde{P}_{Z|X}^{\dagger}\right]_{y,x}} \geq \frac{p_x + j}{2 k} \right)}} \nonumber \\
& \quad \, + \sum_{x \in \X^{\prime}}{\sum_{j < p_x}{\P_p\!\left(\sum_{y \in \Y}{\hat{P}_{Y_1^n}(y) \left[\tilde{P}_{Z|X}^{\dagger}\right]_{y,x}} \leq \frac{p_x + j}{2 k} \right)}}  ,
\label{Eq: Conditional Probability of Error}
\end{align}
where (a) follows from the union bound, (b) splits a summation over $j \in \{0,\dots,k\}\!\backslash\!\{p_x\}$ into two summations (and one of these summations is $0$ if $p_x \in \{0,k\}$), and (c) holds because $\hat{p}_x = j$ implies that $k\sum_{y \in \Y}{\hat{P}_{Y_1^n}(y) \big[\tilde{P}_{Z|X}^{\dagger}\big]_{y,x}}$ is closer to $j$ than $p_x$ due to \eqref{Eq: Element-wise Decoder}, and we count the tie case, where $k \sum_{y \in \Y}{\hat{P}_{Y_1^n}(y) \big[\tilde{P}_{Z|X}^{\dagger}\big]_{y,x}}$ is equally close to $j$ and $p_x$, as an error since this gives us an upper bound on the desired conditional probability of error.  

To show that this upper bound in \eqref{Eq: Conditional Probability of Error} vanishes, observe that for any $x \in \X^{\prime}$ and any $j > p_x$ (assuming $p_x < k$),
\begin{equation}
\begin{aligned}
& \P_p\!\left(\sum_{y \in \Y}{\hat{P}_{Y_1^n}(y) \left[\tilde{P}_{Z|X}^{\dagger}\right]_{y,x}} \geq \frac{p_x + j}{2 k} \right) \\
& \qquad \qquad = \P_p\!\left(\frac{1}{n}\sum_{i = 1}^{n}{\left[\tilde{P}_{Z|X}^{\dagger}\right]_{Y_i,x}} - \frac{p_x}{k} \geq \frac{j - p_x}{2 k} \right) ,
\end{aligned}
\label{Eq: Hoeffding Bound 00}
\end{equation}
which holds because
$$ \sum_{y \in \Y}{\hat{P}_{Y_1^n}(y) \left[\tilde{P}_{Z|X}^{\dagger}\right]_{y,x}} = \frac{1}{n}\sum_{i = 1}^{n}{\left[\tilde{P}_{Z|X}^{\dagger}\right]_{Y_i,x}} $$
almost surely. To bound the right hand side of \eqref{Eq: Hoeffding Bound 00}, we notice three facts:
\begin{enumerate}
\item $\big\{\big[\tilde{P}_{Z|X}^{\dagger}\big]_{Y_i,x} : i \in \{1,\dots,n\}\big\}$ are i.i.d. random variables, because $\big[\tilde{P}_{Z|X}^{\dagger}\big]_{Y_i,x}$ is a deterministic function of $Y_i$ (since $\tilde{P}_{Z|X}^{\dagger}$ is a known deterministic matrix), and $Y_1^n$ are i.i.d. random variables given $M = p$.
\item For each $i \in \{1,\dots,n\}$, $\big[\tilde{P}_{Z|X}^{\dagger}\big]_{Y_i,x}$ is bounded almost surely by
\begin{align*}
\left|\left[\tilde{P}_{Z|X}^{\dagger}\right]_{Y_i,x}\right| & = \left|e_{Y_i}^{\T} \tilde{P}_{Z|X}^{\dagger} e_x\right| \\
& \leq \max_{\substack{u \in \R^{|\Y|}, \, v \in \R^r : \\ \left\|u\right\|_2 = \left\|v\right\|_2 = 1}}{\left|u^{\T} \tilde{P}_{Z|X}^{\dagger} v\right|} \\
& = \left\|\tilde{P}_{Z|X}^{\dagger}\right\|_{\mathsf{op}} \triangleq \sigma \, ,
\end{align*}
where $e_j$ denotes the $j$th standard basis vector with unity at the $j$th position and zero elsewhere, and the last equality follows from the \textit{Courant-Fischer-Weyl min-max theorem}, cf. \cite[Section 3.1, Problem 6, p.155]{HornJohnson1991}, \cite[Lemma 2]{RakocevicWimmer2003}. 
\item For each $i \in \{1,\dots,n\}$, $\big[\tilde{P}_{Z|X}^{\dagger}\big]_{Y_i,x}$ has expected value
\begin{align*}
\E_{p}\!\left[\left[\tilde{P}_{Z|X}^{\dagger}\right]_{Y_i,x}\right] & = \sum_{y \in \Y}{\E_{p}\!\left[\hat{P}_{Y_1^n}(y)\right] \left[\tilde{P}_{Z|X}^{\dagger}\right]_{y,x}} \\
& =  \sum_{y \in \Y}{P_Z(y) \left[\tilde{P}_{Z|X}^{\dagger}\right]_{y,x}} \\
& = P_X(x) \\
& = \frac{p_x}{k} \, ,
\end{align*}
where $\E_p[\cdot]$ represents expectation with respect to the conditional probability distribution of $Y_1^n$ given $M = p$, the third equality crucially uses the fact that $P_{Z|X}$ has a right inverse since its rows are linearly independent,\footnote{The existence of a right inverse of $P_{Z|X}$ ensures that the input distribution $P_X$ can be uniquely recovered from the output distribution $P_Z$ and $P_{Z|X}$. This elucidates why our achievability bound depends on the rank of $P_{Z|X}$.} and the final equality follows from \eqref{Eq: Input Message Representation}.
\end{enumerate}
Using these facts, we can apply Lemma \ref{Lemma: Hoeffding's Inequality} to the right hand side of \eqref{Eq: Hoeffding Bound 00} and obtain
\begin{equation}
\begin{aligned}
& \P_p\!\left(\sum_{y \in \Y}{\hat{P}_{Y_1^n}(y) \left[\tilde{P}_{Z|X}^{\dagger}\right]_{y,x}} \geq \frac{p_x + j}{2 k} \right) \\
& \qquad \qquad \qquad \qquad \qquad \leq \exp\!\left(-\frac{n (j - p_x)^2}{8 \sigma^2 k^2}\right) 
\end{aligned}
\label{Eq: Hoeffding Bound 1}
\end{equation}
for any $x \in \X^{\prime}$ and any $j > p_x$ (assuming $p_x < k$). Likewise, for any $x \in \X^{\prime}$ and any $j < p_x$ (assuming $p_x > 0$), Lemma \ref{Lemma: Hoeffding's Inequality} yields
\begin{equation}
\begin{aligned}
& \P_p\!\left(\sum_{y \in \Y}{\hat{P}_{Y_1^n}(y) \left[\tilde{P}_{Z|X}^{\dagger}\right]_{y,x}} \leq \frac{p_x + j}{2 k} \right) \\
& \qquad \qquad \qquad \qquad \qquad \leq \exp\!\left(-\frac{n (j - p_x)^2}{8 \sigma^2 k^2}\right) .
\end{aligned}
\label{Eq: Hoeffding Bound 2}
\end{equation}
So, bounding the terms in \eqref{Eq: Conditional Probability of Error} using \eqref{Eq: Hoeffding Bound 1} and \eqref{Eq: Hoeffding Bound 2} produces
\begin{align}
\P_p\big(\hat{M} \neq M\big) & \leq \sum_{x \in \X^{\prime}}{\sum_{j \in \{0,\dots,k\}\!\backslash\!\{p_x\}}{\exp\!\left(-\frac{n (j - p_x)^2}{8 \sigma^2 k^2}\right)}} \nonumber \\
& \leq \sum_{x \in \X^{\prime}}{\sum_{j \in \{0,\dots,k\}\!\backslash\!\{p_x\}}{\exp\!\left(-\frac{n}{8 \sigma^2 n^{1-2\epsilon}}\right)}} \nonumber \\
& \leq r \, n^{\frac{1}{2}-\epsilon} \exp\!\left(-\frac{n^{2\epsilon}}{8 \sigma^2}\right)
\label{Eq: Maximal probability of error}
\end{align}
where the second inequality holds because $k = \lfloor n^{\frac{1}{2}-\epsilon} \rfloor \leq n^{\frac{1}{2}-\epsilon}$ and $|j - p_x| \geq 1$ for all $x \in \X^{\prime}$ and all $j \neq p_x$, and the third inequality holds because $|\X^{\prime}| = r$ and $k \leq n^{\frac{1}{2}-\epsilon}$.

Lastly, taking expectations with respect to the law of $M$ in \eqref{Eq: Maximal probability of error} yields
\begin{equation}
\label{Eq: Overall probability of error}
P_{\mathsf{error}}^{(n)} \leq r \, n^{\frac{1}{2}-\epsilon} \exp\!\left(-\frac{n^{2\epsilon}}{8 \sigma^2}\right) ,
\end{equation}
which implies that $\lim_{n \rightarrow \infty}{P_{\mathsf{error}}^{(n)}} = 0$. Therefore, using \eqref{Eq: Order for Permutation Channel}, the rate
$$ R = \lim_{n \rightarrow \infty}{\frac{\displaystyle{\log\!\left(\binom{k + r - 1}{r - 1}\right)}}{\log(n)}} = \frac{r-1}{2} - \epsilon(r-1) $$
is achievable for every $\epsilon \in \big(0,\frac{1}{2}\big)$, and $\Cperm(P_{Z|X}) \geq \frac{r-1}{2}$. This completes the proof.
\end{proof}

We now make two pertinent remarks regarding Theorem \ref{Thm: Achievability Bound}. Firstly, the randomized encoder and element-wise thresholding decoder presented in the achievability proof constitute a computationally tractable coding scheme. Indeed, unlike the random coding argument in traditional channel coding, the element-wise thresholding decoder takes $O(n)$ (i.e., linear) time, because constructing $\hat{P}_{y_1^n}$ from $y_1^n$ requires $O(n)$ time and computing \eqref{Eq: Element-wise Decoder} requires $O\big(\sqrt{n}\big)$ time. Therefore, communication via noisy permutation channels appears to not require the development of conceptually sophisticated coding schemes to theoretically achieve capacity. (Of course, other code constructions could be of utility based on alternative practical considerations.) Furthermore, our achievability proof also implies the existence of a good deterministic code using a simple application of the \textit{probabilistic method} (see, e.g., \cite[Lemma 2.2]{BlochBarros2011}).

Secondly, although we have presented Theorem \ref{Thm: Achievability Bound} under an average probability of error criterion, our achievability proof establishes a lower bound on noisy permutation channel capacity under a maximal probability of error criterion as well; see, e.g., \eqref{Eq: Maximal probability of error}. More generally, the noisy permutation channel capacity of a DMC remains the same under a maximal probability of error criterion. This follows from a straightforward \textit{expurgation} argument similar to \cite[Theorem 18.3, Corollary 18.1]{PolyanskiyWu2017Notes} or \cite[Section 7.7, p.204]{CoverThomas2006}.\footnote{Indeed, $\Cperm$ under a maximal probability of error criterion is clearly upper bounded by $\Cperm$ under an average probability of error criterion, and expurgating the code used to achieve $\Cperm$ under an average probability of error criterion shows that this bound can be met with equality.}

For the special case where $r = \rank(P_{Z|X}) = 2$, we next present an alternative proof of Theorem \ref{Thm: Achievability Bound} which exploits the second moment method bound in Lemma \ref{Lemma: Testing between Converging Hypotheses}. For convenience, we also state the corresponding achievability result in the ensuing proposition.

\begin{proposition}[Achievability Bound for DMCs with Rank $2$]
\label{Prop: Achievability Bound for DMCs with Rank 2}
The noisy permutation channel capacity of a DMC $P_{Z|X}$ with $r \triangleq \rank(P_{Z|X}) = 2$ is lower bounded by 
$$ \Cperm(P_{Z|X}) \geq \frac{1}{2} \, . $$
\end{proposition}

\begin{proof} 
We commence our proof without imposing the $r = 2$ constraint. As in the proof of Theorem \ref{Thm: Achievability Bound}, consider the reduced input alphabet $\X^{\prime} = \{1,\dots,r\} \subseteq \X$ such that the rows of $\tilde{P}_{Z|X}$ are linearly independent, the message set $\M = \Simplex_{r,k}$ where $k = {\big \lfloor} n^{\frac{1}{2}-\epsilon} {\big \rfloor}$ for any $\epsilon \in \big(0,\frac{1}{2}\big)$, and the randomized encoder $f_n : \Simplex_{r,k} \rightarrow \X^n$ given in \eqref{Eq: Randomized Encoder 1} and \eqref{Eq: Input Message Representation}. However, on the receiver end, consider the ML decoder $g_n : \Y^n \rightarrow \Simplex_{r,k}$ such that
$$ \forall y_1^n \in \Y^n, \enspace g_n(y_1^n) = \argmax_{p \in \Simplex_{r,k}}{P_{Y_1^n|M}(y_1^n|p)} \, , $$
where the tie-breaking rule (to choose one maximizer when there are several) does not affect $P_{\mathsf{error}}^{(n)}$. We now analyze the average probability of error for this simple encoding and decoding scheme.

Firstly, as before, we condition on the event $\{M = p\}$ for some $p = \big(\frac{p_1}{k},\dots,\frac{p_r}{k}\big) \in \Simplex_{r,k}$, and note that
$$ X_1^n \stackrel{\text{i.i.d.}}{\sim} P_X \, , \quad Z_1^n \stackrel{\text{i.i.d.}}{\sim} P_Z \, , \quad Y_1^n \stackrel{\text{i.i.d.}}{\sim} P_Z \, , $$
where $P_Z$ denotes the output distribution when $P_X$, defined via \eqref{Eq: Input Message Representation}, is ``pushed forward'' through $P_{Z|X}$. Moreover, as before, we let $\P_{p}$ represent the underlying probability measure after conditioning on $\{M = p\}$. The conditional probability that our ML decoder makes an error is upper bounded by
\begin{align}
& \P_p\big(\hat{M} \neq M \big) \nonumber \\
& \overset{\eqmakebox[C][c]{}}{=} \P_p\big(g_n(Y_1^n) \neq p\big) \nonumber \\
& \overset{\eqmakebox[C][c]{\footnotesize (a)}}{\leq} \P_p\big(\exists \, q \in \Simplex_{r,k}\!\backslash\!\{p\}, \, P_{Y_1^n|M}(Y_1^n|q) \geq P_{Y_1^n|M}(Y_1^n|p)\big) \nonumber \\
& \overset{\eqmakebox[C][c]{\footnotesize (b)}}{\leq} \sum_{q \in \Simplex_{r,k}\!\backslash\!\{p\}}{\P_p\big(P_{Y_1^n|M}(Y_1^n|q) \geq P_{Y_1^n|M}(Y_1^n|p)\big)} 
\label{Eq: Conditional Probability of Error 2}
\end{align}
where (a) is an upper bound because we regard the ML decoding equality case, $P_{Y_1^n|M}(Y_1^n|q) = P_{Y_1^n|M}(Y_1^n|p)$ for $q \neq p$, as an error even though the ML decoder may return the correct message in this scenario, and (b) follows from the union bound. 

Then, to prove that this upper bound in \eqref{Eq: Conditional Probability of Error 2} vanishes, for any message $q = \big(\frac{q_1}{k},\dots,\frac{q_r}{k}\big) \in \Simplex_{r,k}\!\backslash\!\{p\}$, consider a binary hypothesis test with likelihoods given by
\begin{align*}
\text{Given } H = 0 & : Y_1^n \stackrel{\text{i.i.d.}}{\sim} P_Z \, , \\
\text{Given } H = 1 & : Y_1^n \stackrel{\text{i.i.d.}}{\sim} Q_Z \, ,
\end{align*}
where the hypotheses $H = 0$ and $H = 1$ correspond to the messages $M = p$ and $M = q$, respectively, and $Q_Z \in \Simplex_{\Y}$ is the output distribution when the input distribution $Q_X \in \Simplex_{\X}$, defined analogously to \eqref{Eq: Input Message Representation} as
$$ Q_X(x) = 
\begin{cases}
\frac{q_x}{k}, & \text{for } x \in \X^{\prime} \\ 
0, & \text{for } x \in \X\backslash\X^{\prime}
\end{cases} , $$
is ``pushed forward'' through the channel $P_{Z|X}$. Notice that the $\ell^2$-distance between $P_Z$ and $Q_Z$ can be upper and lower bounded using $\|P_X - Q_X\|_{2}$;\footnote{The ensuing lower bound is where we crucially introduce a dependence between the noisy permutation channel capacity and rank of a DMC. Furthermore, the upper and lower bounds together imply that $\|P_Z - Q_Z\|_{2} = \Theta\big(\|P_X - Q_X\|_{2}\big)$.\label{Footnote: Scaling of Message Diff}} indeed,
\begin{equation}
\label{Eq: Scaling of Message Diff}
\begin{aligned}
\sigma_{\mathsf{min}}\big(\tilde{P}_{Z|X}\big) \left\|P_X - Q_X\right\|_{2} & \leq \left\|P_Z - Q_Z\right\|_{2} \\
& \leq \left\|P_{Z|X}\right\|_{\mathsf{op}} \left\|P_X - Q_X\right\|_{2} ,
\end{aligned}
\end{equation}
where $\sigma_{\mathsf{min}}\big(\tilde{P}_{Z|X}\big) > 0$ because the rows of $\tilde{P}_{Z|X}$ are linearly independent, the first inequality follows from the Courant-Fischer-Weyl min-max theorem, cf. \cite[Theorem 7.3.8]{HornJohnson2013}, because $\|P_X - Q_X\|_2 = \|p - q\|_2$, and $P_Z$ and $Q_Z$ can be obtained by pushing $p$ and $q$ forward through the channel $\tilde{P}_{Z|X}$, respectively, and the second inequality follows from the definition of operator norm. So, letting
\begin{equation}
\label{Eq: epsilon Definition}
\epsilon_n = \epsilon + \frac{\log\!\left(\sigma_{\mathsf{min}}\big(\tilde{P}_{Z|X}\big)\right) + \frac{1}{2}\log\!\left(\sum_{i = 1}^{r}{\left(p_i - q_i\right)^2}\right)}{\log(n)} 
\end{equation}
such that $\epsilon_n \in \big(0,\frac{1}{2}\big)$ for all sufficiently large $n$ (depending on $P_{Z|X}$), we have
\begin{align*}
\left\|P_Z - Q_Z\right\|_{2} & \geq \sigma_{\mathsf{min}}\big(\tilde{P}_{Z|X}\big) \left\|p - q\right\|_{2} \\
& = \frac{\sigma_{\mathsf{min}}\big(\tilde{P}_{Z|X}\big)}{{\big \lfloor} n^{\frac{1}{2}-\epsilon} {\big \rfloor}} \sqrt{\sum_{i = 1}^{r}{\left(p_i - q_i\right)^2}} \\
& \geq \frac{1}{n^{\frac{1}{2}-\epsilon_n}} \, ,
\end{align*}
where $p = \big(\frac{p_1}{k},\dots,\frac{p_r}{k}\big) \in \Simplex_{r,k}$ and $q = \big(\frac{q_1}{k},\dots,\frac{q_r}{k}\big) \in \Simplex_{r,k}$ with $k = {\big \lfloor} n^{\frac{1}{2}-\epsilon} {\big \rfloor}$. Using Lemma \ref{Lemma: Testing between Converging Hypotheses} (which is based on the second moment method in Lemma \ref{Lemma: Second Moment Method}), if $H \sim \Ber\big(\frac{1}{2}\big)$, i.e., the hypotheses are equiprobable, then the ML decoding probability of error for our binary hypothesis testing problem, $P_{\mathsf{ML}}^{(n)} = \P\big(\hat{H}_{\mathsf{ML}}^n(Y_1^n) \neq H\big)$, satisfies
\begin{align*}
P_{\mathsf{ML}}^{(n)} & = \frac{1}{2}\P_p\big(\hat{H}_{\mathsf{ML}}^n(Y_1^n) = 1\big) + \frac{1}{2}\P_q\big(\hat{H}_{\mathsf{ML}}^n(Y_1^n) = 0\big) \\
& \leq \frac{|\Y|}{2|\Y| + 2 n^{2\epsilon_n}} \, . 
\end{align*}
This implies that the false alarm probability satisfies
\begin{align}
\P_p\big(\hat{H}_{\mathsf{ML}}^n(Y_1^n) = 1\big) & = \P_p\big(P_{Y_1^n|M}(Y_1^n|q) \geq P_{Y_1^n|M}(Y_1^n|p)\big) \nonumber \\
& \leq \frac{|\Y|}{|\Y| + n^{2\epsilon_n}} \, ,
\label{Eq: Bound on conditional probability}
\end{align}
where the equality follows from breaking ties in ML decoding, i.e., in cases where we get $P_{Y_1^n|M}(Y_1^n|q) = P_{Y_1^n|M}(Y_1^n|p)$, by assigning $\hat{H}_{\mathsf{ML}}^n(Y_1^n) = 1$ (which does not affect the analysis of $P_{\mathsf{ML}}^{(n)}$ in Lemma \ref{Lemma: Testing between Converging Hypotheses}).

Next, combining \eqref{Eq: Conditional Probability of Error 2} and \eqref{Eq: Bound on conditional probability} yields
\begin{align}
\P_p\big(\hat{M} \neq M\big) & \leq \sum_{q \in \Simplex_{r,k}\!\backslash\!\{p\}}{\frac{|\Y|}{|\Y| + n^{2\epsilon_n}}} \nonumber \\
& \leq |\Y| \sum_{q \in \Simplex_{r,k}\!\backslash\!\{p\}}{\frac{1}{n^{2\epsilon_n}}} \nonumber \\
& = \frac{|\Y|}{\sigma_{\mathsf{min}}\big(\tilde{P}_{Z|X}\big)^2 n^{2\epsilon}} \sum_{q \in \Simplex_{r,k}\!\backslash\!\{p\}}{\frac{1}{\sum_{i = 1}^{r}{\left(p_i - q_i\right)^2}}} , 
\label{Eq: Intermediate Lattice Bound}
\end{align}
where the last equality follows from substituting \eqref{Eq: epsilon Definition}. At this point, we use the fact that $r = 2$ to simplify \eqref{Eq: Intermediate Lattice Bound} so that
\begin{align}
\P_p\big(\hat{M} \neq M\big) & \overset{\eqmakebox[D][c]{\footnotesize (a)}}{\leq} \frac{|\Y|}{2 \, \sigma_{\mathsf{min}}\big(\tilde{P}_{Z|X}\big)^2 n^{2\epsilon}} \sum_{q \in \Simplex_{r,k}\!\backslash\!\{p\}}{\frac{1}{\left(p_1 - q_1\right)^2}} \nonumber \\
& \overset{\eqmakebox[D][c]{\footnotesize (b)}}{\leq} \frac{|\Y|}{\sigma_{\mathsf{min}}\big(\tilde{P}_{Z|X}\big)^2 n^{2\epsilon}} \sum_{j = 1}^{\infty}{\frac{1}{j^2}} \nonumber \\
& \overset{\eqmakebox[D][c]{\footnotesize (c)}}{=} \frac{|\Y| \pi^2}{6 \, \sigma_{\mathsf{min}}\big(\tilde{P}_{Z|X}\big)^2 n^{2\epsilon}} \, ,
\label{Eq: Maximal probability of error 2}
\end{align}
where (a) holds because $p_1 + p_2 = q_1 + q_2 = k$, (b) holds because $j = p_1 - q_1$ ranges over a subset of all non-zero integers, and (c) utilizes the renowned solution to the Basel problem.

Finally, taking expectations with respect to the law of $M$ in \eqref{Eq: Maximal probability of error 2} produces
\begin{equation}
\label{Eq: Overall probability of error 2}
P_{\mathsf{error}}^{(n)} \leq \frac{|\Y| \pi^2}{6 \, \sigma_{\mathsf{min}}\big(\tilde{P}_{Z|X}\big)^2 n^{2\epsilon}} \, , 
\end{equation}
which implies that $\lim_{n \rightarrow \infty}{P_{\mathsf{error}}^{(n)}} = 0$. Therefore, as argued in the proof of Theorem \ref{Thm: Achievability Bound}, $\Cperm(P_{Z|X}) \geq \frac{r-1}{2} = \frac{1}{2}$, which completes the proof.
\end{proof}

In view of Proposition \ref{Prop: Achievability Bound for DMCs with Rank 2}, some further remarks are in order. Firstly, the high-level proof strategy to establish Proposition \ref{Prop: Achievability Bound for DMCs with Rank 2} parallels the \emph{pairwise error probability analysis} technique used in conventional channel coding problems (see, e.g., \cite{SasonShamai2006}). However, the details of our hypothesis testing formulation and the bounds we use to execute our analysis are different to such classical approaches. 

Secondly, if $r > 2$, then the obvious approach to bounding $P_{\mathsf{error}}^{(n)}$ starting from \eqref{Eq: Intermediate Lattice Bound} (and taking expectations with respect to $M$) yields
\begin{equation}
\label{Eq: General Perror Bound}
P_{\mathsf{error}}^{(n)} \leq \frac{|\Y|}{\sigma_{\mathsf{min}}\big(\tilde{P}_{Z|X}\big)^2 n^{2\epsilon}} \sum_{m \in \Z^r\backslash\! (0,\dots,0)}{\frac{1}{\left\|m\right\|_2^2}} \, , 
\end{equation}
because we can define $m = (m_1,\dots,m_r) \in \Z^r$ such that $m_i = p_i - q_i$ for every $i \in \{1,\dots,r\}$, and then take the summation over additional sequences $m$ whose sums are not necessarily equal to $0$ (i.e., we take the summation over additional $q_1,\dots,q_r$ whose sums are not necessarily equal to $k$, as is the case when $q \in \Simplex_{r,k}$). It is straightforward to verify that the bound in \eqref{Eq: General Perror Bound} diverges when $r \geq 2$. Indeed, notice that
\begin{align}
\sum_{m \in \Z^r\backslash\! (0,\dots,0)}{\frac{1}{\left\|m\right\|_2^2}} & \geq \sum_{m \in \Z^r\backslash\! (0,\dots,0)}{\frac{1}{\left\|m\right\|_1^2}} \nonumber \\
& \geq \frac{1}{(r-1)!} \sum_{d = 1}^{\infty}{\frac{1}{d^2} \prod_{j = 1}^{r-1}{(d+j)}} \nonumber \\
& \geq \frac{1}{(r-1)!} \sum_{d = 1}^{\infty}{\frac{1}{d^{3-r}}} = +\infty \, ,
\end{align}
where the first inequality uses the monotonicity of $\ell^p$-norms in $p \in [1,\infty]$, the second inequality follows from enumerating over all possible $\ell^1$-norms $d$ and noting that there are $\binom{d+r-1}{r-1}$ (entry-wise) non-negative points in the integer lattice $\Z^r$ that have an $\ell^1$-norm of $d$, and the expression in the final inequality is infinity due to the divergent nature of the harmonic series. In the $r = 2$ case, as in the proof of Proposition \ref{Prop: Achievability Bound for DMCs with Rank 2} above, it is possible to tighten \eqref{Eq: General Perror Bound} and obtain a summation over $\Z$ rather than $\Z^2$. However, such a tightening does not ameliorate the divergent situation for $r > 2$. So, the proof technique of Proposition \ref{Prop: Achievability Bound for DMCs with Rank 2} cannot be used for $r > 2$.

Thirdly, as in the earlier proof of Theorem \ref{Thm: Achievability Bound}, the randomized encoder and ML decoder presented in the proof of Proposition \ref{Prop: Achievability Bound for DMCs with Rank 2} also constitute a computationally tractable coding scheme. In particular, the ML decoder requires at most $O\big(\sqrt{n}\big)$ likelihood ratio tests, which means that the decoder operates in polynomial time in $n$. 

Fourthly, for any non-trivial DMCs, we intuitively expect the rate of decay of $P_{\mathsf{error}}^{(n)}$ to be dominated by the rate of decay of the probability of error in distinguishing between two ``consecutive'' messages. Although we do not derive precise error exponents or rates of decay in this paper, \eqref{Eq: Maximal probability of error}, \eqref{Eq: Overall probability of error}, Lemma \ref{Lemma: Testing between Converging Hypotheses}, and \eqref{Eq: Overall probability of error 2} indicate that this intuition is accurate. 

Lastly, when we specialize \eqref{Eq: Overall probability of error 2} for a $\BSC(\delta)$ with $\delta \in \big(0,\frac{1}{2}\big) \cup \big(\frac{1}{2},1\big)$, we get
\begin{equation}
P_{\mathsf{error}}^{(n)} \leq \frac{\pi^2}{3 (1-2\delta)^2 n^{2\epsilon}} \, ,
\end{equation}
because $|\Y| = |\{0,1\}| = 2$ and $\sigma_{\mathsf{min}}\big(\tilde{P}_{Z|X}\big) = |1 - 2\delta|$. This improves the constant in our corresponding bound in \cite[Theorem 3, Equation (19)]{Makur2018} by a factor of $3$.

\subsection{Converse Bounds for Strictly Positive DMCs}
\label{Converse Bounds for Strictly Positive DMCs}

In this subsection, we first prove Theorem \ref{Thm: Converse Bound I}. To this end, once again recall the formalism introduced in subsection \ref{Permutation Channel Model}.

\begin{proof}[Proof of Theorem \ref{Thm: Converse Bound I}]
Suppose we are given a sequence of encoder-decoder pairs $\{(f_n,g_n)\}_{n \in \N}$ on message sets of size $|\M| = n^R$ such that $\lim_{n \rightarrow \infty}{P_{\mathsf{error}}^{(n)}} = 0$. Consider the Markov chain $M \rightarrow X_1^n \rightarrow Z_1^n \rightarrow Y_1^n \rightarrow \hat{P}_{Y_1^n}$. Observe using \eqref{Eq: The Random Perm Channel} that for every $y_1^n \in \Y^n$ and $m \in \M$,
$$ P_{Y_1^n|M}(y_1^n|m) = \binom{n}{n \hat{P}_{y_1^n}}^{\! -1} \P\!\left(\hat{P}_{Z_1^n} = \hat{P}_{y_1^n}\middle|M = m\right) . $$
Since $P_{Y_1^n|M}(y_1^n|m)$ depends on $y_1^n$ through $\hat{P}_{y_1^n}$, the \textit{Fisher-Neyman factorization theorem} implies that $\hat{P}_{Y_1^n}$ is a sufficient statistic of $Y_1^n$ for $M$ \cite[Theorem 3.6]{Keener2010}. Then, following the standard argument from \cite[Section 7.9]{CoverThomas2006}, we have
\begin{align}
R \log(n) & \overset{\eqmakebox[E][c]{\footnotesize (a)}}{=} H(M) \nonumber \\
& \overset{\eqmakebox[E][c]{\footnotesize (b)}}{=} H(M|\hat{M}) + I(M;\hat{M}) \nonumber \\
& \overset{\eqmakebox[E][c]{\footnotesize (c)}}{\leq} 1 + P_{\mathsf{error}}^{(n)} R \log(n) + I(M;Y_1^n) \nonumber \\
& \overset{\eqmakebox[E][c]{\footnotesize (d)}}{=} 1 + P_{\mathsf{error}}^{(n)} R \log(n) + I(M;\hat{P}_{Y_1^n}) \nonumber \\
& \overset{\eqmakebox[E][c]{\footnotesize (e)}}{\leq} 1 + P_{\mathsf{error}}^{(n)} R \log(n) + I(X_1^n;\hat{P}_{Y_1^n}) \, ,
\label{Eq: Fano step}
\end{align}
where (a) holds because $M$ is uniformly distributed, (b) follows from the definition of mutual information, (c) follows from \textit{Fano's inequality} and the \textit{data processing inequality} \cite[Theorems 2.10.1 and 2.8.1]{CoverThomas2006}, (d) holds because $\hat{P}_{Y_1^n}$ is a sufficient statistic, cf. \cite[Section 2.9]{CoverThomas2006}, and (e) also follows from the data processing inequality. 

We now upper bound $I(X_1^n;\hat{P}_{Y_1^n})$. Notice that
\begin{align}
I(X_1^n;\hat{P}_{Y_1^n}) & = H(\hat{P}_{Y_1^n}) - H(\hat{P}_{Y_1^n}|X_1^n) \nonumber \\
& \leq (|\Y| - 1) \log(n+1) \nonumber \\
& \quad \, - \sum_{x_1^n \in \{0,1\}^n}{P_{X_1^n}(x_1^n) H(\hat{P}_{Y_1^n}|X_1^n = x_1^n)} \, ,
\label{Eq: Upper bound on MI}
\end{align}
where we use the upper bound on the number of possible empirical distributions given in \eqref{Eq: Number of Emp Dists}. Given $X_1^n = x_1^n$ for any fixed $x_1^n \in \X^n$, $\{Z_i \sim P_{Z|X}(\cdot|x_i) : i \in \{1,\dots,n\}\}$ are mutually independent and $\hat{P}_{Z_1^n} = \hat{P}_{Y_1^n}$ (almost surely). To lower bound $H(\hat{P}_{Y_1^n}|X_1^n = x_1^n)$, we need some additional definitions. Let $x^* = \argmax_{x \in \X}{\hat{P}_{x_1^n}(x)}$, choosing one maximizer arbitrarily if there are several. For every $y \in \Y$, define the random variable
$$ N_y \triangleq \sum_{i = 1}^{n}{\I\!\left\{x_i = x^*\right\} \I\!\left\{Z_i = y\right\}} \, .  $$
Furthermore, define the empirical conditional distribution
$$ \hat{P}_{Z|X = x}^{n}(y) \triangleq \frac{1}{n \hat{P}_{x_1^n}(x)} \sum_{i = 1}^{n}{\I\!\left\{x_i = x\right\} \I\!\left\{Z_i = y\right\}} $$
for every $y \in \Y$ and every $x \in \X$ such that $\hat{P}_{x_1^n}(x) > 0$, where we let $\hat{P}_{Z|X = x}^{n} = \big(\hat{P}_{Z|X = x}^{n}(y) : y \in \Y \big) \in \Simplex_{\Y}$ so that
\begin{equation}
\label{Eq: Empirical Total Probability Law}
\hat{P}_{Z_1^n} = \sum_{x \in \X}{\hat{P}_{x_1^n}(x) \hat{P}_{Z|X = x}^{n}} \, . 
\end{equation}
In the sequel, for any $x \in \X$, if $\hat{P}_{x_1^n}(x) = 0$, then we interpret $\hat{P}_{x_1^n}(x) \hat{P}_{Z|X = x}^{n}$ as the zero vector. Using these definitions, we have
\begin{align}
& H(\hat{P}_{Y_1^n}|X_1^n = x_1^n) \nonumber \\
& \overset{\eqmakebox[F][c]{\footnotesize (a)}}{=} H(\hat{P}_{Z_1^n}|X_1^n = x_1^n) \nonumber \\
& \overset{\eqmakebox[F][c]{\footnotesize (b)}}{=} H\!\left(\sum_{x \in \X}{\hat{P}_{x_1^n}(x) \hat{P}_{Z|X = x}^{n}} \middle| X_1^n = x_1^n\right) \nonumber \\
& \overset{\eqmakebox[F][c]{\footnotesize (c)}}{\geq} H\!\left(\hat{P}_{x_1^n}(x^*) \hat{P}_{Z|X = x^*}^{n} \middle| X_1^n = x_1^n\right) \nonumber \\
& \overset{\eqmakebox[F][c]{\footnotesize (d)}}{=} H(N_1,\dots,N_{|\Y|-1} | X_1^n = x_1^n) \nonumber \\
& \overset{\eqmakebox[F][c]{\footnotesize (e)}}{=} H(N_1|X_1^n = x_1^n) + \sum_{i = 2}^{|\Y|-1}{H(N_i|N_1,\dots,N_{i-1}, X_1^n = x_1^n)} \nonumber \\
& \overset{\eqmakebox[F][c]{\footnotesize (f)}}{\geq} \sum_{i = 1}^{|\Y|-1}{H(N_i|\{N_j : j \in \Y\backslash\!\{i,|\Y|\}\}, X_1^n = x_1^n)} \, , 
\label{Eq: Intermediate Conditional Entropy Bound}
\end{align}
where (a) holds because $\hat{P}_{Z_1^n} = \hat{P}_{Y_1^n}$ (almost surely), (b) uses \eqref{Eq: Empirical Total Probability Law}, (c) follows from \cite[Problem 2.14]{CoverThomas2006} because $\big\{\hat{P}_{Z|X = x}^{n} \in \Simplex_{\Y}: x \in \X\big\}$ are mutually (conditionally) independent random variables given $X_1^n = x_1^n$, (d) holds because $\hat{P}_{Z|X = x^*}^{n}$ sums to unity and we let $\Y = \{1,\dots,|\Y|\}$ without loss of generality, (e) follows from the chain rule for Shannon entropy (and the summation is $0$ when $|\Y| = 2$), and (f) holds because conditioning reduces Shannon entropy (and it equals $H(N_1|X_1^n = x_1^n)$ when $|\Y| = 2$). 

We next lower bound $H(N_1|N_2,\dots,N_{|\Y|-1}, X_1^n = x_1^n)$; the other terms in the sum in \eqref{Eq: Intermediate Conditional Entropy Bound} can be lower bounded similarly. Let $p \triangleq P_{Z|X}(1|x^*)/\big(P_{Z|X}(1|x^*) + P_{Z|X}(|\Y||x^*)\big)$. Then, the conditional distribution of $N_1$ given $N_2,\dots,N_{|\Y|-1}$ and $X_1^n = x_1^n$ is given by
\begin{align*}
& \P\!\left(N_1 = k_1 \, \middle| \, N_2 = k_2,\dots,N_{|\Y|-1}=k_{|\Y|-1},X_1^n = x_1^n\right) \\
& = \frac{\P\!\left(N_1 = k_1,\dots,N_{|\Y|-1}=k_{|\Y|-1} \, \middle| \, X_1^n = x_1^n\right)}{\P\!\left(N_2 = k_2,\dots,N_{|\Y|-1}=k_{|\Y|-1} \, \middle| \, X_1^n = x_1^n\right)} \\
& = \frac{\displaystyle{\prod_{j \in \Y}{\frac{P_{Z|X}(j|x^*)^{k_j}}{k_{j}!} } }}{\displaystyle{\frac{\left(P_{Z|X}(1|x^*) + P_{Z|X}(|\Y||x^*)\right)^{k_1 + k_{|\Y|}}}{(k_1 + k_{|\Y|})!} \!\prod_{j = 2}^{|\Y|-1}{\!\!\frac{P_{Z|X}(j|x^*)^{k_j}}{k_j !}} }} \\
& = \binom{k_1 + k_{|\Y|}}{k_1} p^{k_1} (1 - p)^{k_{|\Y|}} 
\end{align*}
for every $k_1,\dots,k_{|\Y|} \in \N \cup \! \{0\}$ such that $\sum_{j \in \Y}{k_j} = n \hat{P}_{x_1^n}(x^*)$. Therefore, $N_1 \sim \bin\big(n \hat{P}_{x_1^n}(x^*) - \sum_{j = 2}^{|\Y|-1}{k_j} , p\big)$ given $N_2 = k_2,\dots,N_{|\Y|-1} = k_{|\Y|-1}$ and $X_1^n = x_1^n$. (We remark that this calculation also holds for the $|\Y| = 2$ case because $\sum_{j = 2}^{|\Y|-1}{k_j} = 0$.) Now, for some fixed constant $\tau \in \big(1-P_{Z|X}(1|x^*)-P_{Z|X}(|\Y||x^*),1\big)$, let $E \in \{0,1\}$ be a binary random variable defined by
$$ E \triangleq \I\!\left\{\sum_{j = 2}^{|\Y|-1}{N_j} \leq n \hat{P}_{x_1^n}(x^*) \tau\right\} . $$
Observe using the Bienaym\'{e}-Chebyshev inequality that
\begin{align}
& \P(E = 0|X_1^n = x_1^n) \nonumber \\
& \overset{\eqmakebox[H][c]{}}{=} \P\!\left(\hat{P}_{Z|X = x^*}^{n}(1) + \hat{P}_{Z|X = x^*}^{n}(|\Y|) < 1 - \tau \middle| X_1^n = x_1^n\right) \nonumber \\
& \overset{\eqmakebox[H][c]{\footnotesize (a)}}{\leq} \frac{\VAR\!\left(\hat{P}_{Z|X = x^*}^{n}(1) + \hat{P}_{Z|X = x^*}^{n}(|\Y|) \middle| X_1^n = x_1^n \right)}{\left(\tau + P_{Z|X}(\{1,|\Y|\}|x^*) - 1\right)^2} \nonumber \\
& \overset{\eqmakebox[H][c]{\footnotesize (b)}}{=} \frac{P_{Z|X}(\{1,|\Y|\}|x^*) \!\left(1-P_{Z|X}(\{1,|\Y|\}|x^*)\right)}{n \hat{P}_{x_1^n}(x^*) \!\left(\tau + P_{Z|X}(\{1,|\Y|\}|x^*) - 1\right)^2} \nonumber \\
& \overset{\eqmakebox[H][c]{\footnotesize (c)}}{\leq} \frac{|\X| P_{Z|X}(\{1,|\Y|\}|x^*) \!\left(1-P_{Z|X}(\{1,|\Y|\}|x^*)\right)}{n \!\left(\tau + P_{Z|X}(\{1,|\Y|\}|x^*) - 1\right)^2} 
\label{Eq: Previous Chebyshev Bound} \\
& \overset{\eqmakebox[H][c]{}}{=} O\!\left(\frac{1}{n}\right) ,
\label{Eq: Chebyshev Bound}
\end{align}
where (a), (b), and (c) use the notation $P_{Z|X}(\{1,|\Y|\}|x^*) = P_{Z|X}(1|x^*)+P_{Z|X}(|\Y||x^*)$, and (c) also utilizes the fact that $\hat{P}_{x_1^n}(x^*) \geq \frac{1}{|\X|}$ by definition of $x^*$. Then, we can apply Lemma \ref{Lemma: Approximation of Binomial Entropy} and get
\begin{align}
& H(N_1|N_2,\dots,N_{|\Y|-1}, X_1^n = x_1^n) \nonumber \\
& \overset{\eqmakebox[I][c]{\footnotesize (a)}}{\geq} H(N_1|N_2,\dots,N_{|\Y|-1},E, X_1^n = x_1^n) \nonumber \\
& \overset{\eqmakebox[I][c]{\footnotesize (b)}}{\geq} \left(\!1 - O\!\left(\frac{1}{n}\right)\!\right) H(N_1|N_2,\dots,N_{|\Y|-1},E = 1,X_1^n = x_1^n) \nonumber \\ 
& \overset{\eqmakebox[I][c]{}}{=} \left(\!1 - O\!\left(\frac{1}{n}\right)\!\right) \cdot \nonumber \\
& \quad \enspace \E\!\left[H\!\left(\!\bin\!\left(\! n\hat{P}_{x_1^n}(x^*) - \sum_{j = 2}^{|\Y|-1}{N_j},p\!\right)\!\!\right)\middle|E = 1,X_1^n = x_1^n\right] \nonumber \\ 
& \overset{\eqmakebox[I][c]{\footnotesize (c)}}{\geq} \left(\!1 - O\!\left(\frac{1}{n}\right)\!\right) H\!\left(\bin\!\left({\bigg \lfloor} \frac{n (1-\tau)}{|\X|} {\bigg \rfloor},p\right)\right) \nonumber \\
& \overset{\eqmakebox[I][c]{\footnotesize (d)}}{\geq} \left(\!1 - O\!\left(\frac{1}{n}\right)\!\right) \frac{1}{2} \log(2 \pi e p (1-p) \!\left(\frac{n (1-\tau)}{|\X|} - 1\right)) \nonumber \\
& \quad \, \, - \left(\!1 - O\!\left(\frac{1}{n}\right)\!\right) \frac{c(p) |\X|}{n (1-\tau) - |\X|} \, ,
\label{Eq: Intermediate Conditional Entropy Bound 2}
\end{align}
where (a) holds because conditioning reduces Shannon entropy, (b) follows from \eqref{Eq: Chebyshev Bound} and the non-negativity of Shannon entropy, (c) follows from \cite[Problem 2.14]{CoverThomas2006} and the facts that $E = 1$ and $\hat{P}_{x_1^n}(x^*) \geq \frac{1}{|\X|}$, and (d) employs Lemma \ref{Lemma: Approximation of Binomial Entropy}. Here, to employ Lemma \ref{Lemma: Approximation of Binomial Entropy} and obtain \eqref{Eq: Intermediate Conditional Entropy Bound 2}, we implicitly utilize the assumption that all conditional distributions in $\{P_{Z|X}(\cdot|x) \in \Simplex_{\Y} : x \in \X\}$ are strictly positive, which ensures that $p \in (0,1)$.\footnote{Note that since we do not know a priori which value $x^*$ takes and we have to prove \eqref{Eq: Intermediate Conditional Entropy Bound 2} for every term in \eqref{Eq: Intermediate Conditional Entropy Bound}, we have to assume that $P_{Z|X}(y|x) > 0$ for all $x \in \X$ and $y \in \Y$.} We also note that when $|\Y| = 2$, the above argument mutatis mutandis yields
\begin{align*}
& H(N_1|X_1^n = x_1^n) \\
& \quad \quad \quad = H\!\left(\bin\!\left(n \hat{P}_{x_1^n}(x^*),p\right)\right) \\
& \quad \quad \quad \geq \frac{1}{2} \log(2 \pi e p (1-p) \!\left(\frac{n}{|\X|} - 1\right)) - \frac{c(p) |\X|}{n - |\X|} \, ,
\end{align*}
which is lower bounded by \eqref{Eq: Intermediate Conditional Entropy Bound 2}. So, the bound in \eqref{Eq: Intermediate Conditional Entropy Bound 2} is valid for all $|\Y| \geq 2$.

Next, to upper bound $I(X_1^n;\hat{P}_{Y_1^n})$, let $\tau^*$ be any fixed constant such that
$$ 1 - \min_{\substack{x \in \X, \, y,y^{\prime} \in \Y: \\ y \neq y^{\prime}}}{P_{Z|X}(\{y,y^{\prime}\}|x)} < \tau^* < 1 \, , $$
where $P_{Z|X}(\{y,y^{\prime}\}|x) = P_{Z|X}(y|x) + P_{Z|X}(y^{\prime}|x)$ for any $x \in \X$ and $y,y^{\prime} \in \Y$ such that $y \neq y^{\prime}$. Notice that when we analyze other conditional entropy terms $H(N_i|\{N_j : j \in \Y\backslash\!\{i,|\Y|\}\}, X_1^n = x_1^n)$ for $i \in \{1,\dots,|\Y|-1\}$ akin to our analysis of $H(N_1|N_2,\dots,N_{|\Y|-1}, X_1^n = x_1^n)$ above with $\tau = \tau^*$, the maximum bound of the form \eqref{Eq: Previous Chebyshev Bound} is
\begin{equation}
\label{Eq: Max O(1/n) bound}
\max_{\substack{x \in \X, \\ y,y^{\prime} \in \Y : \\ y \neq y^{\prime}}}{\frac{|\X| P_{Z|X}(\{y,y^{\prime}\}|x) \!\left(1-P_{Z|X}(\{y,y^{\prime}\}|x)\right)}{n \left(\tau^* + P_{Z|X}(\{y,y^{\prime}\}|x) - 1\right)^2}} = O\!\left(\frac{1}{n}\right)
\end{equation}
which remains $O\big(\frac{1}{n}\big)$. Furthermore, let $p^*$ be the optimal value of
$$ p = \frac{P_{Z|X}(y|x)}{P_{Z|X}(\{y,y^{\prime}\}|x)} > 0 $$
that minimizes \eqref{Eq: Intermediate Conditional Entropy Bound 2}, with $\tau = \tau^*$ and the $O\big(\frac{1}{n}\big)$ term given by \eqref{Eq: Max O(1/n) bound}, over all $x \in \X$ and $y,y^{\prime} \in \Y$ such that $y \neq y^{\prime}$. Then, following the derivation of \eqref{Eq: Intermediate Conditional Entropy Bound 2}, for all $i \in \{1,\dots,|\Y|-1\}$, we obtain the lower bound
\begin{align}
& H(N_i|\{N_j : j \in \Y\backslash\!\{i,|\Y|\}\}, X_1^n = x_1^n) \nonumber \\
& \geq \left(\!1 - O\!\left(\frac{1}{n}\right)\!\right) \frac{1}{2} \log(2 \pi e p^* (1-p^*) \!\left(\frac{n (1-\tau^*)}{|\X|} - 1\right)) \nonumber \\
& \quad \, \, - \left(\!1 - O\!\left(\frac{1}{n}\right)\!\right) \frac{c(p^*) |\X|}{n (1-\tau^*) - |\X|} \, ,
\label{Eq: Intermediate Conditional Entropy Bound 3}
\end{align}
where the $O\big(\frac{1}{n}\big)$ term is given by \eqref{Eq: Max O(1/n) bound}. Hence, we can combine \eqref{Eq: Upper bound on MI}, \eqref{Eq: Intermediate Conditional Entropy Bound}, and \eqref{Eq: Intermediate Conditional Entropy Bound 3} to produce
\begin{align}
& I(X_1^n;\hat{P}_{Y_1^n}) \nonumber \\
& \leq (|\Y| - 1) \Bigg(\log(n+1) - \left(\!1 - O\!\left(\frac{1}{n}\right)\!\right) \cdot \nonumber \\ 
& \quad \quad \quad \quad \quad \quad \bigg(\frac{1}{2} \log(2 \pi e p^* (1-p^*) \!\left(\frac{n (1-\tau^*)}{|\X|} - 1\right)) \nonumber \\
& \quad \quad \quad \quad \quad \quad \enspace - \frac{c(p^*) |\X|}{n (1-\tau^*) - |\X|}\bigg)\Bigg) \, .
\label{Eq: Upper bound on MI 2}
\end{align}

Finally, combining \eqref{Eq: Fano step} and \eqref{Eq: Upper bound on MI 2}, and dividing by $\log(n)$, yields
\begin{align*}
R & \leq \frac{1}{\log(n)} + P_{\mathsf{error}}^{(n)} R + (|\Y| - 1) \Bigg(\frac{\log(n+1)}{\log(n)} \, - \\
&  \quad \left(\!1 - O\!\left(\frac{1}{n}\right)\!\right) \! \bigg(\frac{\log(2 \pi e p^* (1-p^*) (1-\tau^*)/|\X|)}{2 \log(n)} \, + \\
& \qquad \qquad \qquad \quad \enspace \, \, \frac{\log(n - (|\X|/(1-\tau^*)))}{2 \log(n)} \, - \\
& \qquad \qquad \qquad \quad \enspace \, \, \frac{c(p^*) |\X|}{\left(n (1-\tau^*) - |\X|\right)\log(n)}\bigg)\Bigg) \, ,
\end{align*}
where letting $n \rightarrow \infty$ produces
$$ R \leq \frac{|\Y|-1}{2} \, . $$
This completes the proof. 
\end{proof}

The proofs of Lemmata \ref{Lemma: Second Moment Method} and \ref{Lemma: Testing between Converging Hypotheses} in appendices \ref{Proof of Lemma Second Moment Method} and \ref{Proof of Lemma Testing between Converging Hypotheses} (along with the discussion following Lemmata \ref{Lemma: Second Moment Method} and \ref{Lemma: Testing between Converging Hypotheses}) and the proof of Proposition \ref{Prop: Achievability Bound for DMCs with Rank 2} portray that the distinguishability between two ``consecutive'' (encoded) messages can be determined by a careful comparison of the difference between their means and a variance (at least in the rank $2$ case). This suggests that the CLT can be used to obtain the correct scaling of $|\M|$ with $n$ in general. We remark that the CLT is in fact implicitly used in the above converse proof when we apply Lemma \ref{Lemma: Approximation of Binomial Entropy}, because estimates for the entropy of a binomial distribution are typically obtained using the CLT. 

We conclude this section by using Theorem \ref{Thm: Converse Bound I}, Proposition \ref{Prop: Degradation by Symmetric Channels}, and Lemma \ref{Lemma: Equivalent Model} to establish the alternative converse bound on the noisy permutation channel capacity of strictly positive DMCs given in Theorem \ref{Thm: Converse Bound II}.

\begin{proof}[Proof of Theorem \ref{Thm: Converse Bound II}]
As in the proof of Theorem \ref{Thm: Converse Bound I}, consider any sequence of encoder-decoder pairs $\{(f_n,g_n)\}_{n \in \N}$ on message sets of size $|\M| = n^R$ such that $\lim_{n \rightarrow \infty}{P_{\mathsf{error}}^{(n)}} = 0$. This defines the Markov chain $M \rightarrow X_1^n \rightarrow Z_1^n \rightarrow Y_1^n$, and the standard argument from \cite[Section 7.9]{CoverThomas2006}, which yielded \eqref{Eq: Fano step} earlier, easily produces
\begin{equation}
\label{Eq: Fano step 2}
R \log(n) \leq 1 + P_{\mathsf{error}}^{(n)} R \log(n) + I(X_1^n;Y_1^n) \, . 
\end{equation}
We proceed to upper bounding $I(X_1^n;Y_1^n)$ using a degradation argument. 

First, we reduce the cardinality of the input alphabet $\X$ of the DMC $P_{Z|X}$. In particular, we let $\X^{\prime} \subseteq \X$ be any (fixed) subset of $\X$ such that $q \triangleq |\X^{\prime}| = \ext(P_{Z|X})$ and the set of conditional distributions $\{P_{Z|X}(\cdot|x) \in \Simplex_{\Y} : x \in \X^{\prime}\}$ (as vectors in $\R^{|\Y|}$) are the extreme points of the convex hull of $\{P_{Z|X}(\cdot|x) \in \Simplex_{\Y}: x \in \X\}$.\footnote{We note that when there are multiple copies of an extreme point of the convex hull of $\{P_{Z|X}(\cdot|x) \in \Simplex_{\Y}: x \in \X\}$ in $\{P_{Z|X}(\cdot|x)\in \Simplex_{\Y} : x \in \X\}$, we only add one of these conditional distributions to $\{P_{Z|X}(\cdot|x) \in \Simplex_{\Y} : x \in \X^{\prime}\}$.} Moreover, we let $P_{Z|\tilde{X}} \in \R^{q \times |\Y|}$ denote the row stochastic matrix whose rows are given by $\{P_{Z|X}(\cdot|x) \in \Simplex_{\Y} : x \in \X^{\prime}\}$, where the random variable $\tilde{X} \in \X^{\prime}$. Since the convex hulls of $\{P_{Z|X}(\cdot|x) \in \Simplex_{\Y}: x \in \X^{\prime}\}$ and $\{P_{Z|X}(\cdot|x)\in \Simplex_{\Y} : x \in \X\}$ are equivalent, for every $x \in \X$, we have
\begin{equation}
\label{Eq: Convex Combinations}
\forall z \in \Y, \enspace P_{Z|X}(z|x) = \sum_{\tilde{x} \in \X^{\prime}}{Q_{x}(\tilde{x}) P_{Z|X}(z|\tilde{x})} 
\end{equation}
for some convex weights $\{Q_{x}(\tilde{x}) \geq 0: \tilde{x} \in \X^{\prime}\}$ such that $\sum_{\tilde{x} \in \X^{\prime}}{Q_x(\tilde{x})} = 1$. Observe that for every probability distribution $P_{X_1^n} \in \Simplex_{\X^n}$, we can construct the probability distribution $P_{\tilde{X}_1^n} \in \Simplex_{(\X^{\prime})^n}$ given by
\begin{equation}
\label{Eq: Existence of Distribution}
\forall \tilde{x}_1^n \in (\X^{\prime})^n, \enspace P_{\tilde{X}_1^n}(\tilde{x}_1^n) \triangleq \sum_{x_1^n \in \X^n}{P_{X_1^n}(x_1^n) \prod_{i = 1}^{n}{Q_{x_i}(\tilde{x}_i)}} \, , 
\end{equation}
where the random variables $\tilde{X}_1,\dots,\tilde{X}_n \in \X^{\prime}$. This distribution has the property that it induces the marginal distribution $P_{Z_1^n}$ of $Z_1^n$, namely, for all $z_1^n \in \Y^n$,
\begin{align}
P_{Z_1^n}(z_1^n) & \triangleq \sum_{x_1^n \in \X^n}{P_{X_1^n}(x_1^n) \prod_{i = 1}^{n}{P_{Z|X}(z_i|x_i)}} \nonumber \\
& = \sum_{x_1^n \in \X^n}{P_{X_1^n}(x_1^n) \prod_{i = 1}^{n}{\sum_{\tilde{x} \in \X^{\prime}}{Q_{x_i}(\tilde{x}) P_{Z|X}(z_i|\tilde{x})}}} \nonumber \\
& = \sum_{x_1^n \in \X^n}{P_{X_1^n}(x_1^n) \sum_{\tilde{x}_1^n \in (\X^{\prime})^n}{\prod_{i = 1}^{n}{Q_{x_i}(\tilde{x}_i) P_{Z|X}(z_i|\tilde{x}_i)}}} \nonumber \\
& = \sum_{\tilde{x}_1^n \in (\X^{\prime})^n}{\sum_{x_1^n \in \X^n}{P_{X_1^n}(x_1^n) \prod_{i = 1}^{n}{Q_{x_i}(\tilde{x}_i) P_{Z|X}(z_i|\tilde{x}_i)}}} \nonumber  \\
& = \sum_{\tilde{x}_1^n \in (\X^{\prime})^n}{P_{\tilde{X}_1^n}(\tilde{x}_1^n) \prod_{i = 1}^{n}{P_{Z|X}(z_i|\tilde{x}_i)}} \, , 
\label{Eq: Equivalent Form of Z_1^n}
\end{align} 
where the first equality defines the distribution $P_{Z_1^n}$ of $Z_1^n$ (and uses the memorylessness of $P_{Z|X}$), the second equality follows from \eqref{Eq: Convex Combinations}, the third equality follows from the distributive property, the fourth equality follows from swapping the order of summations, and the fifth equality follows from \eqref{Eq: Existence of Distribution}. Furthermore, notice that
\begin{align}
& I(X_1^n;Y_1^n) \nonumber \\
& \overset{\eqmakebox[J][c]{}}{=} \sum_{x_1^n \in \X^n}{P_{X_1^n}(x_1^n) D(P_{Y_1^n|X_1^n}(\cdot|x_1^n)||P_{Y_1^n})} \nonumber \\
& \overset{\eqmakebox[J][c]{\footnotesize (a)}}{=} \sum_{x_1^n \in \X^n}{P_{X_1^n}(x_1^n) D(P_{Z_1^n|X_1^n}(\cdot|x_1^n) \cdot \Pi||P_{Y_1^n})} \nonumber \\
& \overset{\eqmakebox[J][c]{\footnotesize (b)}}{=} \sum_{x_1^n \in \X^n} P_{X_1^n}(x_1^n)  \, \cdot \nonumber \\
& \qquad \quad D\!\left(\sum_{\tilde{x}_1^n \in (\X^{\prime})^n}{\prod_{i = 1}^{n}{Q_{x_i}(\tilde{x}_i)} \big(P_{Z_1^n|X_1^n}(\cdot|\tilde{x}_1^n)} \cdot \Pi\big) \middle|\middle|P_{Y_1^n}\!\!\right) \nonumber \\
& \overset{\eqmakebox[J][c]{\footnotesize (c)}}{=} \sum_{x_1^n \in \X^n} P_{X_1^n}(x_1^n) \, \cdot \nonumber \\
& \qquad \quad D\!\left(\sum_{\tilde{x}_1^n \in (\X^{\prime})^n}{\prod_{i = 1}^{n}{Q_{x_i}(\tilde{x}_i)} P_{Y_1^n|X_1^n}(\cdot|\tilde{x}_1^n)} \middle|\middle|P_{Y_1^n}\!\!\right) \nonumber \\
& \overset{\eqmakebox[J][c]{\footnotesize (d)}}{\leq} \!\!\sum_{x_1^n \in \X^n} \!\!\! P_{X_1^n}(x_1^n) \!\!\! \sum_{\tilde{x}_1^n \in (\X^{\prime})^n} \prod_{i = 1}^{n}{Q_{x_i}(\tilde{x}_i)} D(P_{Y_1^n|X_1^n}(\cdot|\tilde{x}_1^n) || P_{Y_1^n}) \nonumber \\
& \overset{\eqmakebox[J][c]{\footnotesize (e)}}{=} \!\!\sum_{\tilde{x}_1^n \in (\X^{\prime})^n} \!\! D(P_{Y_1^n|X_1^n}(\cdot|\tilde{x}_1^n) || P_{Y_1^n}) \!\! \sum_{x_1^n \in \X^n}  \!\!\! P_{X_1^n}(x_1^n) \prod_{i = 1}^{n}{Q_{x_i}(\tilde{x}_i)} \nonumber \\
& \overset{\eqmakebox[J][c]{\footnotesize (f)}}{=} \sum_{\tilde{x}_1^n \in (\X^{\prime})^n}{P_{\tilde{X}_1^n}(\tilde{x}_1^n) D(P_{Y_1^n|X_1^n}(\cdot|\tilde{x}_1^n) || P_{Y_1^n})} \nonumber \\
& \overset{\eqmakebox[J][c]{\footnotesize (g)}}{=} I(\tilde{X}_1^n;Y_1^n) \, ,
\label{Eq: Reduced Alphabet MI}
\end{align}
where (a) and (c) hold because the conditional distribution $P_{Y_1^n|X_1^n}(\cdot|x_1^n) = P_{Z_1^n|X_1^n}(\cdot|x_1^n) \cdot \Pi \in \Simplex_{\Y^n}$ is the output of pushing the conditional distribution $P_{Z_1^n|X_1^n}(\cdot|x_1^n) \in \Simplex_{\Y^n}$ through the random permutation channel $\Pi = P_{Y_1^n|Z_1^n}$ (defined in \eqref{Eq: The Random Perm Channel}) for every $x_1^n \in \X^n$, (b) follows from \eqref{Eq: Equivalent Form of Z_1^n} after substituting Kronecker delta distributions $P_{X_1^n}$ into \eqref{Eq: Existence of Distribution} (and uses the memorylessness of $P_{Z|X}$), (d) follows from the convexity of KL divergence, (e) follows from swapping the order of summations, (f) follows from \eqref{Eq: Existence of Distribution}, and (g) holds because \eqref{Eq: Equivalent Form of Z_1^n} conveys that $P_{Y_1^n} \in \Simplex_{\Y^n}$, which is the marginal distribution of $Y_1^n$ in the original Markov chain $X_1^n \rightarrow Z_1^n \rightarrow Y_1^n$, is also the marginal distribution of $Y_1^n$ in the Markov chain $\tilde{X}_1^n \rightarrow Z_1^n \rightarrow Y_1^n$.\footnote{Note that we abuse notation here and use the same random variable labels $Z_1^n$ and $Y_1^n$ for the Markov chains $X_1^n \rightarrow Z_1^n \rightarrow Y_1^n$ and $\tilde{X}_1^n \rightarrow Z_1^n \rightarrow Y_1^n$, because the two chains can be coupled so that $Z_1^n$ and $Y_1^n$ are shared random variables.} 

Second, we construct an equivalent model of the Markov chain $\tilde{X}_1^n \rightarrow Z_1^n \rightarrow Y_1^n$, which has reduced input alphabet $\X^{\prime}$, $\tilde{X}_1^n \sim P_{\tilde{X}_1^n}$, $P_{Z_1^n|\tilde{X}_1^n}$ given by the DMC $P_{Z|\tilde{X}}$, and $P_{Y_1^n|Z_1^n} = \Pi$ given by the random permutation channel in \eqref{Eq: The Random Perm Channel}. Employing Lemma \ref{Lemma: Equivalent Model} (also see Figure \ref{Figure: Permutation Channel 2}), we can swap the random permutation channel $\Pi$ and the DMC $P_{Z|\tilde{X}}$ to get a Markov chain $\tilde{X}_1^n \rightarrow V_1^n \rightarrow W_1^n$ such that the channel $P_{W_1^n|\tilde{X}_1^n}$ is equivalent to the channel $P_{Y_1^n|\tilde{X}_1^n}$. In this alternative Markov chain, $V_1^n \in (\X^{\prime})^n$ is an independent random permutation of $\tilde{X}_1^n \in (\X^{\prime})^n$, and $W_1^n \in \Y^n$ is the output of passing $V_1^n$ through a DMC $P_{W|V}$, which satisfies $P_{W|V} = P_{Z|\tilde{X}}$ (as shown in \eqref{Eq: Equivalence Condition}). Hence, we have
\begin{equation}
\label{Eq: Swapped MI}
I(\tilde{X}_1^n;Y_1^n) = I(\tilde{X}_1^n;W_1^n) \, , 
\end{equation}
since $\tilde{X}_1^n$ is common to both Markov chains $\tilde{X}_1^n \rightarrow Z_1^n \rightarrow Y_1^n$ and $\tilde{X}_1^n \rightarrow V_1^n \rightarrow W_1^n$. 

Third, we construct a $q$-ary symmetric channel that dominates the DMC $P_{W|V} = P_{Z|\tilde{X}}$ in the degradation sense. To this end, define the parameter
$$ \delta = \frac{\nu}{1 - \nu + \frac{\nu}{q-1}} > 0 $$
in terms of the minimum entry, $\nu$, of $P_{Z|X}$, viz.,
$$ \nu = \min_{x \in \X, \, y\in \Y}{P_{Z|X}(y|x)} = \min_{x \in \X^{\prime}, \, y\in \Y}{P_{Z|\tilde{X}}(y|x)} > 0 \, , $$
where the second equality follows from \eqref{Eq: Convex Combinations}. (Note that $\delta > 0$ because $\nu > 0$, and $\nu > 0$ since $P_{Z|X}$ is strictly positive.) Then, applying Proposition \ref{Prop: Degradation by Symmetric Channels}, we get that $P_{W|V} = P_{Z|\tilde{X}}$ is a degraded version of $\qSC(\delta)$ (see Definition \ref{Def: Symmetric Channel}). Let the input random variable of the $\qSC(\delta)$ be $V \in \X^{\prime}$, and the output random random variable be $\tilde{W} \in \X^{\prime}$, so that we can write $P_{\tilde{W}|V} = S_{\delta}$. Now consider the Markov chain $\tilde{X}_1^n \rightarrow V_1^n \rightarrow \tilde{W}_1^n$, where $V_1^n$ is a random permutation of $\tilde{X}_1^n$ as before, and the channel $P_{\tilde{W}_1^n|V_1^n}$ is given by the DMC $P_{\tilde{W}|V} = S_{\delta}$. Since the degradation preorder tensorizes, the channel $P_{W|V}^{\otimes n}$ is a degraded version of the channel $P_{\tilde{W}|V}^{\otimes n} = S_{\delta}^{\otimes n}$, where $A^{\otimes n}$ denotes the $n$-fold Kronecker product (or tensor product) of a row stochastic matrix $A$, which corresponds to $n$ uses of the memoryless channel $A$.\footnote{The tensorization property of the degradation preorder is well-known in information theory. For a proof, notice that given any three row stochastic matrices $A,B,C$ (with consistent dimensions so that the ensuing products are legal), if $A = B C$, then $A \otimes A = (B \otimes B)(C \otimes C)$ using the \textit{mixed-product property}, where $\otimes$ denotes the Kronecker product.} Thus, we have a Markov chain $\tilde{X}_1^n \rightarrow V_1^n \rightarrow \tilde{W}_1^n \rightarrow W_1^n$ using Definition \ref{Def: Degradation Preorder} (where we neglect the difference between physical and stochastic degradation since it is inconsequential in this context). By the data processing inequality, this implies that
\begin{equation}
\label{Eq: Upper Bound MI}
I(\tilde{X}_1^n;W_1^n) \leq I(\tilde{X}_1^n;\tilde{W}_1^n) \, .
\end{equation}

Fourth, we again swap the random permutation and DMC blocks in the Markov chain $\tilde{X}_1^n \rightarrow V_1^n \rightarrow \tilde{W}_1^n$ using Lemma \ref{Lemma: Equivalent Model}. As we argued earlier, this produces an equivalent Markov chain $\tilde{X}_1^n \rightarrow \tilde{Z}_1^n \rightarrow \tilde{Y}_1^n$ such that the channel $P_{\tilde{Y}_1^n|\tilde{X}_1^n}$ is equivalent to the channel $P_{\tilde{W}_1^n|\tilde{X}_1^n}$. Moreover, in this alternative Markov chain, $\tilde{X}_1^n \sim P_{\tilde{X}_1^n}$ as before, the product channel $P_{\tilde{Z}_1^n|\tilde{X}_1^n}$ is defined by the DMC $P_{\tilde{Z}|\tilde{X}} = S_{\delta}$ (i.e., a $\qSC(\delta)$) with input and output alphabet $\X^{\prime}$, and the channel $P_{\tilde{Y}_1^n|\tilde{Z}_1^n}$ is defined by a random permutation channel. Hence, we have
\begin{equation}
\label{Eq: Re-swapped MI}
I(\tilde{X}_1^n;\tilde{W}_1^n) = I(\tilde{X}_1^n;\tilde{Y}_1^n) \, .
\end{equation}

Finally, combining \eqref{Eq: Reduced Alphabet MI}, \eqref{Eq: Swapped MI}, \eqref{Eq: Upper Bound MI}, and \eqref{Eq: Re-swapped MI}, we get
$$ I(X_1^n;Y_1^n) \leq I(\tilde{X}_1^n;\tilde{Y}_1^n) \, , $$  
which implies that the right hand side of \eqref{Eq: Fano step 2} can be upper bounded as
$$ R \log(n) \leq 1 + P_{\mathsf{error}}^{(n)} R \log(n) + I(\tilde{X}_1^n;\tilde{Y}_1^n) \, . $$
Executing the Fisher-Neyman factorization argument from the outset of the proof of Theorem \ref{Thm: Converse Bound I}, we obtain that $\hat{P}_{\tilde{Y}_1^n}$ is a sufficient statistic of $\tilde{Y}_1^n$ for $\tilde{X}_1^n$. So, we have
\begin{equation}
\label{Eq: Fano step 3}
R \log(n) \leq 1 + P_{\mathsf{error}}^{(n)} R \log(n) + I(\tilde{X}_1^n;\hat{P}_{\tilde{Y}_1^n}) \, ,
\end{equation}
much like the bound in \eqref{Eq: Fano step}. At this stage, noting that the DMC $P_{\tilde{Z}|\tilde{X}} = S_{\delta}$ is strictly positive, we can upper bound $I(\tilde{X}_1^n;\hat{P}_{\tilde{Y}_1^n})$ by following the proof of Theorem \ref{Thm: Converse Bound I} mutatis mutandis. Indeed, starting from \eqref{Eq: Fano step 3} and proceeding with the proof of Theorem \ref{Thm: Converse Bound I} yields
$$ R \leq \frac{q-1}{2} = \frac{|\X^{\prime}|-1}{2} = \frac{\ext(P_{Z|X})-1}{2} \, . $$
This completes the proof.
\end{proof}

\section{Noisy Permutation Channel Capacity}
\label{Permutation Channel Capacity}

To complement our main result in Theorem \ref{Thm: Strictly Positive and Full Rank Channels}, we characterize and bound the noisy permutation channel capacities of several other simple classes of DMCs in this section. 

\subsection{Unit Rank Channels}

We start with what is perhaps the simplest setting\textemdash that of an ``independent channel.'' In this case, the noisy permutation channel capacity is obviously zero, and a standard Fano's inequality argument rigorously justifies this. 

\begin{proposition}[$\Cperm$ of Unit Rank Stochastic Matrices]
\label{Prop: Unit Rank Stochastic Matrices}
For a unit rank DMC $P_{Z|X} \in \R^{|\X| \times |\Y|}$ such that all rows of $P_{Z|X}$ are equal, we have
$$ \Cperm(P_{Z|X}) = 0 \, . $$
\end{proposition}

\begin{proof}
We need only prove a converse bound to establish this. Since all rows of $P_{Z|X}$ are equal, the output $Z$ of the DMC $P_{Z|X}$ is independent of the input $X$. Recalling the setup in subsection \ref{Permutation Channel Model}, this implies that
\begin{equation}
\label{Eq: Zero MI}
I(X_1^n;\hat{P}_{Y_1^n}) = 0 \, . 
\end{equation}
Notice that \eqref{Eq: Fano step} (in the proof of Theorem \ref{Thm: Converse Bound I}) holds for any DMC. So, dividing both sides of \eqref{Eq: Fano step} by $\log(n)$ and applying \eqref{Eq: Zero MI} yields
$$ R \leq \frac{1}{\log(n)} + P_{\mathsf{error}}^{(n)} R \, , $$
where letting $n \rightarrow \infty$ produces $R \leq 0$. Therefore, we have $\Cperm(P_{Z|X}) = 0$.
\end{proof}

\subsection{Permutation Transition Matrices}
\label{Permutation Transition Matrices}

Next, we consider the straightforward dual setting of a ``perfect channel.'' In this case, the DMC is an identity channel without loss of generality, which means that the corresponding noisy permutation channel just permutes its input codewords randomly (see subsection \ref{Permutation Channel Model}). So, intuitively, the maximum number of decodable messages that can be reliably communicated is equal to the number of possible empirical distributions over the alphabet of the DMC (see \eqref{Eq: Order for Permutation Channel 2} below). Thus, the noisy permutation channel capacity is clearly characterized by the alphabet size of the DMC. The formal proof is again straightforward, but we include it here for completeness.

\begin{proposition}[$\Cperm$ of Permutation Stochastic Matrices]
\label{Prop: Permutation Stochastic Matrices}
For a DMC $P_{Z|X}$ such that $|\X| = |\Y| = k \geq 2$ and $P_{Z|X} \in \R^{k \times k}$ is a permutation matrix, we have
$$ \Cperm(P_{Z|X}) = k-1 \, . $$
\end{proposition}

\begin{proof} ~\newline
\indent
\underline{Achievability:} Under the setup of subsection \ref{Permutation Channel Model}, consider the obvious encoder-decoder pair:
\begin{enumerate}
\item The message set $\M = \big\{m = \big[m_1 \, \cdots \, m_k\big]^{\T} \in (\N\cup\!\{0\})^k : m_1 + \cdots + m_k = n\big\}$ with cardinality
\begin{equation}
\label{Eq: Order for Permutation Channel 2}
|\M| = \binom{n+k-1}{k-1} = \Theta\!\left(n^{k-1}\right) , 
\end{equation}
where the elements of $\M$ have been re-indexed for convenience.
\item The encoder $f_n : \M \rightarrow \X^n$ is given by
$$ \forall m \in \M, \enspace f_n(m) = (\underbrace{1,\dots,1}_{m_1 \, 1\text{'s}},\underbrace{2,\dots,2}_{m_2 \, 2\text{'s}},\dots,\underbrace{k,\dots,k}_{m_k \, k\text{'s}}) \, , $$
where $\X = \{1,\dots,k\}$ without loss of generality.
\item The decoder $g_n : \Y^n \rightarrow \M$ is given by
$$ \forall y_1^n \in \Y^n, \enspace g_n(y_1^n) = n \, P_{Z|X} \big[ \hat{P}_{y_1^n}(1) \, \cdots \, \hat{P}_{y_1^n}(k)\big]^{\T} , $$
where $\Y = \X$ without loss of generality (and $P_{Z|X} \in \R^{k \times k}$ is a permutation matrix).
\end{enumerate}
Clearly, this encoder-decoder pair achieves $P_{\mathsf{error}}^{(n)} = 0$. Hence, using \eqref{Eq: Order for Permutation Channel 2}, the rate
$$ R = \lim_{n \rightarrow \infty}{\frac{\displaystyle{\log(\binom{n+k-1}{k-1})}}{\log(n)}} = k-1 $$ 
is achievable, and $\Cperm(P_{Z|X}) \geq k-1$.

\underline{Converse:} Recall that \eqref{Eq: Fano step} (in the proof of Theorem \ref{Thm: Converse Bound I}) holds for any DMC. We bound the mutual information term in \eqref{Eq: Fano step} with
\begin{align*}
I(X_1^n;\hat{P}_{Y_1^n}) & = H(\hat{P}_{Y_1^n}) \\
& \leq (k-1)\log(n+1) \, ,
\end{align*}
where the equality holds because $H(\hat{P}_{Y_1^n}|X_1^n) = 0$, and the inequality uses the upper bound on the number of possible empirical distributions given in \eqref{Eq: Number of Emp Dists}. Then, as before, combining \eqref{Eq: Fano step} with the above bound on mutual information and dividing by $\log(n)$ yields
$$ R \leq  \frac{1}{\log(n)} + P_{\mathsf{error}}^{(n)} R + \frac{(k-1)\log(n+1)}{\log(n)}\, , $$
where letting $n \rightarrow \infty$ produces $R \leq k-1$. Therefore, we have $\Cperm(P_{Z|X}) \leq k-1$, which completes the proof.
\end{proof}

\subsection{Strictly Positive Channels}

Recall that Theorem \ref{Thm: Strictly Positive and Full Rank Channels} in section \ref{Main Results} presents the main contribution of this paper\textemdash a closed-form expression for the noisy permutation channel capacity of strictly positive DMCs with full rank. For general strictly positive DMCs, we complement Theorem \ref{Thm: Strictly Positive and Full Rank Channels} by proposing the following conjecture.

\begin{conjecture}[$\Cperm$ of Strictly Positive Channels]
\label{Conj: Strictly Positive DMC Conjecture}
For any strictly positive DMC $P_{Z|X}$, we have
\begin{align*}
\Cperm(P_{Z|X}) & = \liminf_{n \rightarrow \infty}{ \, \frac{1}{\log(n)} \, \sup_{P_{X_1^n} \in \Simplex_{\X^n}}{ \, I(\hat{P}_{X_1^n};\hat{P}_{Y_1^n})}} \\
& = \frac{\rank(P_{Z|X})-1}{2} \, ,
\end{align*}
where the supremum in the first equality is over all probability distributions in $\Simplex_{\X^n}$, or equivalently, over all probability distributions of $\hat{P}_{X_1^n}$. 
\end{conjecture}

While Definition \ref{Def: Permutation Channel Capacity} provides an operational definition of $\Cperm(P_{Z|X})$, the first equality in Conjecture \ref{Conj: Strictly Positive DMC Conjecture} can be construed as the corresponding notion of ``information capacity'' (analogous to the definition of multi-letter information capacity in, e.g., \cite[Definition 18.6]{PolyanskiyWu2017Notes}), and the second equality in Conjecture \ref{Conj: Strictly Positive DMC Conjecture} is a closed-form expression for the noisy permutation channel capacity. As Conjecture \ref{Conj: Strictly Positive DMC Conjecture} reveals, we believe that our achievability bound in Theorem \ref{Thm: Achievability Bound} is most likely tight. To briefly elaborate on this further, the first equality is inspired by the modified Fano's inequality argument in \eqref{Eq: Fano step}, where we also use Lemma \ref{Lemma: Equivalent Model} to obtain $I(\hat{P}_{X_1^n};\hat{P}_{Y_1^n})$ instead of $I(X_1^n;\hat{P}_{Y_1^n})$, and the second equality is suggested by the first equality via intuition from the multivariate CLT.

Next, as a concrete and canonical illustration of Propositions \ref{Prop: Unit Rank Stochastic Matrices} and \ref{Prop: Permutation Stochastic Matrices} and Theorem \ref{Thm: Strictly Positive and Full Rank Channels}, we present the noisy permutation channel capacity of binary symmetric channels below (see Definition \ref{Def: Symmetric Channel} for a definition of $\qSC$s). This result was first proved in \cite[Theorem 3]{Makur2018}. (Note that in the context of the work in \cite{KovacevicVukobratovic2013}, \cite{KovacevicVukobratovic2015}, and \cite{KovacevicTan2018a}, this BSC setting corresponds to noisy permutation channels with substitution errors.)

\begin{proposition}[$\Cperm$ of BSCs {\cite[Theorem 3]{Makur2018}}]
\label{Prop: Permutation Channel Capacity of BSC}
$$ \Cperm(\BSC(\delta)) =
\begin{cases}
1, & \text{for } \delta \in \{0,1\} \\ 
\frac{1}{2}, & \text{for } \delta \in \left(0,\frac{1}{2}\right)\cup\left(\frac{1}{2},1\right) \\
0, & \text{for } \delta = \frac{1}{2}
\end{cases} . $$
\end{proposition}

\begin{proof}
The $\delta = \frac{1}{2}$ case follows from Proposition \ref{Prop: Unit Rank Stochastic Matrices}, the $\delta \in \{0,1\}$ case follows from Proposition \ref{Prop: Permutation Stochastic Matrices}, and the remaining case follows from Theorem \ref{Thm: Strictly Positive and Full Rank Channels}.
\end{proof}

We remark that Proposition \ref{Prop: Permutation Channel Capacity of BSC} illustrates a few somewhat surprising facts about noisy permutation channel capacity. While traditional channel capacity is convex as a function of the channel (with fixed dimensions), noisy permutation channel capacity is clearly non-convex and discontinuous as a function the channel. Moreover, for the most part, the noisy permutation channel capacity of a BSC does not depend on $\delta$. Looking at the proof of Proposition \ref{Prop: Achievability Bound for DMCs with Rank 2} in subsection \ref{Achievability Bounds for DMCs}, this is because the scaling with $n$ of the $\ell^2$-distance between two encoded messages does not change after passing through the memoryless BSC (see, e.g., \eqref{Eq: Scaling of Message Diff} and footnote \ref{Footnote: Scaling of Message Diff}). However, \eqref{Eq: Overall probability of error} and \eqref{Eq: Overall probability of error 2} suggest that $\delta$ does affect the rate of decay of $P_{\mathsf{error}}^{(n)}$. Finally, we note that in a manner similar to Proposition \ref{Prop: Permutation Channel Capacity of BSC}, we can also determine the noisy permutation channel capacity of any $\qSC(\delta)$ for $\delta \in [0,1)$ using Propositions \ref{Prop: Unit Rank Stochastic Matrices} and \ref{Prop: Permutation Stochastic Matrices} and Theorem \ref{Thm: Strictly Positive and Full Rank Channels}.

\subsection{Erasure Channels and Doeblin Minorization}
\label{Erasure Channels}

In this subsection, we consider the important class of $q$-ary erasure channels. Indeed, in the context of communication networks, networks where packets can be dropped are typically modeled as noisy permutation channels with possible deletions, or equivalently, erasures, cf. \cite{KovacevicVukobratovic2013}, \cite{KovacevicVukobratovic2015}, \cite[Remark 1]{KovacevicTan2018a}. Since the transition kernels of $q$-ary erasure channels contain zero entries, our converse results in Theorems \ref{Thm: Converse Bound I} and \ref{Thm: Converse Bound II} do not hold. So, we will present some bounds on the their noisy permutation channel capacities. First, let us recall the definition of $q$-ary erasure channels.

\begin{definition}[$q$-ary Erasure Channel]
\label{Def: Erasure Channel}
Under the formalism presented in subsection \ref{Permutation Channel Model}, we define a \emph{$q$-ary erasure channel} $P_{Z|X}$ with erasure probability $\eta \in [0,1]$, input alphabet $\X$ with $|\X| = q \in \N\backslash\!\{1\}$, and output alphabet $\Y = \X\cup\!\{\mathsf{E}\}$, where $\mathsf{E}$ denotes the erasure symbol, using the conditional distributions
$$ \forall z \in \Y, \forall x \in \X, \enspace P_{Z|X}(z|x) = 
\begin{cases}
1-\eta, & \text{for } z = x \\ 
\eta, & \text{for } z = \mathsf{E} \\
0, & \text{otherwise}
\end{cases} \, .
$$
Moreover, we represent such a channel $P_{Z|X}$ as $\qEC(\eta)$ for convenience. 
\end{definition}

We note that in the special case where $q = 2$, $\X = \{0,1\}$, and $\eta$ is the probability that the input bit is erased, we refer to the $2\text{-}\mathsf{SC}(\eta)$ as a \textit{binary erasure channel} (BEC), denoted $\BEC(\eta)$. 

Next, in order to present our bounds on the noisy permutation channel capacity of erasure channels, we introduce a classical concept from the Markov process literature. As we will see, one approach to proving our achievability bound entails using a symmetric channel that is degraded by the erasure channel under consideration. While we have introduced degradation in subsection \ref{Auxiliary Lemmata}, the specific setting of degradation by erasure channels has been studied extensively in the Markov process literature under the guise of ``Doeblin minorization.'' We next introduce the concept of Doeblin minorization in an information theoretic light (within the formalism of subsection \ref{Permutation Channel Model}), cf. \cite[Section 3]{BhattacharyaWaymire2001}.

\begin{definition}[Doeblin Minorization]
\label{Def: Doeblin Minorization}
A row stochastic matrix $P_{Z|X} \in \R^{|\X| \times |\Y|}$ satisfies the \emph{Doeblin minorization condition} if there exists a probability distribution $Q_Z \in \Simplex_{\Y}$ and a constant $\eta \in (0,1)$ such that
$$ \forall z \in \Y, \forall x \in \X, \enspace P_{Z|X}(z|x) \geq \eta \, Q_Z(z) \, , $$
and we say that $P_{Z|X}$ satisfies $\mathsf{Doeblin}(Q_Z,\eta)$. Furthermore, we say that $P_{Z|X}$ satisfies $\mathsf{Doeblin}(Q_Z,0)$ when $P_{Z|X}$ does not satisfy the Doeblin minorization condition (since the above condition is trivially true when $\eta = 0$). 
\end{definition}

Definition \ref{Def: Doeblin Minorization} of Doeblin minorization is less general than its definition in a finite state space Markov chain context, where one often studies ``local minorization'' of multi-step Markov transition kernels, cf. \cite[Section 4]{BhattacharyaWaymire2001}. On the other hand, our definition applies to more general (rectangular) transition kernels. While the Doeblin minorization condition was originally developed to study the ergodicity of Markov processes,\footnote{As a historical remark, it is worth mentioning that as stated in \cite[Section 3]{BhattacharyaWaymire2001}, ``two of the most powerful ideas in the modern theory of Markov processes were introduced [by Doeblin in \cite{Doeblin1937} and \cite{Doeblin1938}]; namely minorization and coupling, respectively.''} as we alluded to earlier, it turns out to be equivalent to degradation by an erasure channel. The next lemma depicts this known, but seemingly overlooked, connection.

\begin{lemma}[Doeblin Minorization and Degradation {\cite{BhattacharyaWaymire2001,GohariGunluKramer2019}}]
\label{Lemma: Doeblin Minorization and Degradation}
Consider any DMC $P_{Z|X}$ with input alphabet $\X$ and output alphabet $\Y$ with $|\X| = q$. Then, the following are true:
\begin{enumerate}
\item \emph{(Equivalence \cite[Theorem 3.1]{BhattacharyaWaymire2001})} For any constant $\eta \in (0,1)$, $P_{Z|X}$ satisfies $\mathsf{Doeblin}(Q_Z,\eta)$ for some distribution $Q_Z \in \Simplex_{\Y}$ if and only if $P_{Z|X}$ is a degraded version of the $q$-ary erasure channel $\qEC(\eta)$.
\item \emph{(Extremality \cite[Lemma 4]{GohariGunluKramer2019})} The extremal erasure probability $\eta_{*} = \eta_{*}(P_{Z|X})$ such that $P_{Z|X}$ is a degraded version of $\qEC(\eta_{*})$ is given by
\begin{align*}
\eta_{*} & \triangleq \max\!\left\{\eta \in [0,1] : \parbox[]{8.4em}{$P_{Z|X}$ is a degraded\\version of $\qEC(\eta)$}\right\} \\
& = \sup\!\left\{\eta \in [0,1) : \parbox[]{12.7em}{$P_{Z|X}$ satisfies $\mathsf{Doeblin}(Q_Z,\eta)$\\for some distribution $Q_Z \in \Simplex_{\Y}$}\right\} \\
& = \sum_{z \in \Y}{\min_{x \in \X}{P_{Z|X}(z|x)}} \, ,
\end{align*}
where the second equality follows from part 1 and the quantity in the final equality is known as \emph{Doeblin's coefficient of ergodicity}, cf. \cite[Definition 5.1]{Cohenetal1993}.
\end{enumerate}
\end{lemma}

Although Lemma \ref{Lemma: Doeblin Minorization and Degradation} is known in the literature, we provide a proof of part 1 in appendix \ref{Proof of Lemma Doeblin Minorization and Degradation} for completeness. Moreover, we note that the equivalent description of Doeblin minorization as degradation by an erasure channel can also be viewed as a specialization of the so called \textit{regeneration} or \textit{Nummelin splitting} technique in the theory of Harris chains \cite{AthreyaNey1978,Nummelin1978}. 

In the ensuing theorem, we derive Theorem \ref{Thm: Comparison Bound via Degradation}, which uses the notion of degradation to prove a comparison bound for noisy permutation channel capacities, as well as a related bound pertaining to Doeblin minorization, which specializes Theorem \ref{Thm: Comparison Bound via Degradation} for erasure channels (as revealed by our discussion heretofore). As outlined in subsection \ref{Degradation and Permutation Channel Capacity}, this result concurs with the intuition that degraded channels are ``more noisy,'' and therefore, have smaller noisy permutation channel capacity. 

\begin{theorem}[Comparison Bounds via Degradation]
\label{Thm: Comparison Bounds via Degradation}
Consider any two DMCs $P_{Z_1|X} \in \R^{|\X| \times |\mathcal{Z}_1|}$ and $P_{Z_2|X} \in \R^{|\X| \times |\mathcal{Z}_2|}$, with common input alphabet $\X$ and output alphabets $\mathcal{Z}_1$ and $\mathcal{Z}_2$, respectively. Then, the following are true:
\begin{enumerate}
\item If $P_{Z_2|X}$ is a degraded version of $P_{Z_1|X}$, then we have
$$ \Cperm(P_{Z_2|X}) \leq \Cperm(P_{Z_1|X}) \, . $$
\item If $P_{Z_2|X}$ satisfies $\mathsf{Doeblin}(Q_Z,\eta)$ for some distribution $Q_Z \in \Simplex_{\Y}$ and some constant $\eta \in (0,1)$, then we have
$$ \Cperm(P_{Z_2|X}) \leq \Cperm(\qEC(\eta)) \, , $$
where we let $q = |\X|$.
\end{enumerate}
\end{theorem}

\begin{proof} ~\newline
\indent
\underline{Part 1:} Recalling the formalism introduced in subsection \ref{Permutation Channel Model}, fix any (small) $\epsilon > 0$ such that $R \triangleq \Cperm(P_{Z_2|X}) - \epsilon \geq 0$ is an achievable rate for the DMC $P_{Z_2|X}$. Then, for the noisy permutation channel model with DMC $P_{Z_2|X}$, consider the Markov chain $M \rightarrow f_n(M) = X_1^n \rightarrow (Z_2)_1^n \rightarrow (Y_2)_1^n \rightarrow g_n((Y_2)_1^n)$, defined by a message set $\M$ with cardinality $|\M| = n^R$, a sequence of possibly randomized encoders $\{f_n : \M \rightarrow \X^n\}_{n \in \N}$, and a sequence of associated possibly randomized decoders $\{g_n : \mathcal{Z}_2^n \rightarrow \M\cup\!\{\mathsf{e}\}\}_{n \in \N}$, where $(Y_2)_1^n \in \mathcal{Z}_2^n$ denotes a random permutation of the output codeword $(Z_2)_1^n$ of the DMC $P_{Z_2|X}$. Let us define
$$ P_{\mathsf{error}}^{(n)}(P_{Z_2|X},f_n,g_n) \triangleq \P\!\left(M \neq g_n((Y_2)_1^n)\right) $$
as the average probability of error for the noisy permutation channel model corresponding to $P_{Z_2|X}$, $f_n$, and $g_n$. Since $R$ is an achievable rate, we further assume that $f_n$ and $g_n$ are chosen such that $\lim_{n \rightarrow \infty}{P_{\mathsf{error}}^{(n)}(P_{Z_2|X},f_n,g_n)} = 0$. By our assumption in the theorem statement, we know using Definition \ref{Def: Degradation Preorder} that there exists some DMC $P_{Z_2|Z_1} \in \R^{|\mathcal{Z}_1| \times |\mathcal{Z}_2|}$ (with input alphabet $\mathcal{Z}_1$ and output alphabet $\mathcal{Z}_2$) such that $P_{Z_2|X} = P_{Z_1|X} P_{Z_2|Z_1}$. We will use this degradation relation to construct a ``good'' encoder-decoder pair for the DMC $P_{Z_1|X}$.

To this end, for the noisy permutation channel model with DMC $P_{Z_1|X}$, consider the Markov chain $M \rightarrow f_n(M) = X_1^n \rightarrow (Z_1)_1^n \rightarrow (Y_1)_1^n$, where we use the same message set (with cardinality $n^R$) and encoder sequence as before, and $(Y_1)_1^n \in \mathcal{Z}_1^n$ is a random permutation of the output codeword $(Z_1)_1^n$ of the DMC $P_{Z_1|X}$. (In fact, the random variables $M$ and $X_1^n$ are coupled to be equal for the two models.) Now, for every $n \in \N$, construct the decoder $\tilde{g}_n : \mathcal{Z}_1^n \rightarrow \M\cup\!\{\mathsf{e}\}$ so that 
\begin{equation}
\label{Eq: Degradation Decoder}
\forall y_1^n \in \mathcal{Z}_1^n, \enspace \tilde{g}_n(y_1^n) \triangleq g_n(Z_1^n) \, ,
\end{equation}
where $Z_i \sim P_{Z_2|Z_1}(\cdot|y_i)$ for every $i \in \{1,\dots,n\}$, and $Z_1^n$ are mutually independent. This produces the Markov chain $M \rightarrow X_1^n \rightarrow (Z_1)_1^n \rightarrow (Y_1)_1^n \rightarrow \tilde{g}_n((Y_1)_1^n)$. We note that the decoder in \eqref{Eq: Degradation Decoder} essentially ``simulates'' the auxiliary DMC $P_{Z_2|Z_1}$ so that its output is statistically equivalent to the output of the decoder $g_n$ with DMC $P_{Z_2|X}$.

We next prove the intuitively straightforward relation
\begin{equation}
\label{Eq: Equality in Perror}
P_{\mathsf{error}}^{(n)}(P_{Z_2|X},f_n,g_n) = P_{\mathsf{error}}^{(n)}(P_{Z_1|X},f_n,\tilde{g}_n) \, . 
\end{equation}
To establish this, consider yet another Markov chain, $M \rightarrow f_n(M) = X_1^n \rightarrow V_1^n \rightarrow W_1^n \rightarrow U_1^n$, where the message set and $f_n$ are the same as before, $V_1^n \in \X^n$ is a random permutation of $X_1^n$, $W_1^n \in \mathcal{Z}_1^n$ is the output of the DMC $P_{Z_1|X}$ with input $V_1^n$, and $U_1^n \in \mathcal{Z}_2^n$ is the output of the DMC $P_{Z_2|Z_1}$ with input $W_1^n$. Using Lemma \ref{Lemma: Equivalent Model}, notice that the conditional distribution $P_{(Y_1)_1^n|X_1^n}$ is equivalent to the conditional distribution $P_{W_1^n|X_1^n}$. Furthermore, Lemma \ref{Lemma: Equivalent Model} and the degradation relation $P_{Z_2|X} = P_{Z_1|X} P_{Z_2|Z_1}$ imply that the conditional distribution $P_{(Y_2)_1^n|X_1^n}$ is equivalent to the conditional distribution $P_{U_1^n|X_1^n}$. Thus, the joint distribution of $(M,\tilde{g}_n((Y_1)_1^n))$ is equal to the joint distribution of $(M,g_n((Y_2)_1^n))$, because the conditional distributions of $\tilde{g}_n(W_1^n)$ and $g_n(U_1^n)$ given $M$ are equivalent, where we may perceive $U_1^n$ as the intermediate random variables used by $\tilde{g}$ in \eqref{Eq: Degradation Decoder} so that $\tilde{g}_n(W_1^n) = g_n(U_1^n)$. This produces the relation \eqref{Eq: Equality in Perror}.

Lastly, we conclude this proof by realizing that \eqref{Eq: Equality in Perror} reveals that $R = \Cperm(P_{Z_2|X}) - \epsilon$ is an achievable rate for the DMC $P_{Z_1|X}$. Therefore, we have
$$ \Cperm(P_{Z_2|X}) - \epsilon \leq \Cperm(P_{Z_1|X}) \, , $$
and we can let $\epsilon \rightarrow 0$ to obtain the desired inequality. 

\underline{Part 2:} This follows immediately from part 1 of this theorem and Lemma \ref{Lemma: Doeblin Minorization and Degradation}.
\end{proof}

We are now in a position to present bounds on the noisy permutation channel capacity of $q$-ary erasure channels. Although the achievability bound in the ensuing proposition can be obtained as a direct consequence of Theorem \ref{Thm: Achievability Bound}, we will elucidate an alternative coding scheme that establishes this bound using Doeblin minorization. Similarly, although the converse bound in the ensuing proposition is just the trivial bound given in \eqref{Eq: General Combined Converse Bound}, we will provide an alternative proof for it. 

\begin{proposition}[Bounds on $\Cperm$ of $\qEC$]
\label{Prop: Permutation Channel Capacity of q-EC}
For a $q$-ary erasure channel $\qEC(\eta)$ with $\eta \in (0,1)$, we have
$$ \frac{q-1}{2} \leq \Cperm(\qEC(\eta)) \leq q - 1 \, . $$
Furthermore, the extremal noisy permutation channel capacities are $\Cperm(\qEC(0)) = q-1$ and $\Cperm(\qEC(1)) = 0$.
\end{proposition}

\begin{proof} ~\newline
\indent 
\underline{Achievability for $\eta \in (0,1)$:} 
To derive a lower bound on $\Cperm(\qEC(\eta))$, we will employ a useful representation of $\qSC$'s using $\qEC$'s. Observe that the row stochastic transition probability matrix of a $\qSC(\eta(q-1)/q)$ can be decomposed as
\begin{equation}
\label{Eq: Fortuin-Kasteleyn}
S_{\eta(q-1)/q} = \left(1-\eta\right) I + \eta \left(\frac{1}{q} \1 \1^{\T}\right) ,
\end{equation}
where $S_{\eta(q-1)/q}$ is the $q$-ary symmetric channel matrix defined in \eqref{Eq: q-SC Matrix}, $I \in \R^{q\times q}$ is the identity matrix representing a channel that exactly copies its input, $\1 = [1 \, \cdots \, 1]^{\T} \in \R^q$ denotes the column vector with all elements equal to unity, and $\frac{1}{q} \1 \1^{\T}$ represents a channel whose output is an independent uniform random variable. Hence, a $\qSC(\eta(q-1)/q)$ can be equivalently construed as a channel that either copies its input random variable with probability $1-\eta$, or generates a completely independent output random variable that is uniformly distributed on the input alphabet $\X$ (where $|\X| = q$) with probability $\eta$. A consequence of this interpretation, or equivalently, the decomposition \eqref{Eq: Fortuin-Kasteleyn}, is that a $\qSC(\eta(q-1)/q)$ satisfies $\mathsf{Doeblin}(\1^{\T}\!/q,\eta)$, where $\1^{\T}\!/q$ is a row vector representing the uniform distribution. Using part 1 of Lemma \ref{Lemma: Doeblin Minorization and Degradation}, this means that a $\qSC(\eta(q-1)/q)$ is a degraded version of a $\qEC(\eta)$; in particular, a $\qSC(\eta(q-1)/q)$ is statistically equivalent to a $\qEC(\eta)$ followed by a channel that outputs an independent uniformly distributed random variable for the input erasure symbol $\mathsf{E}$, and copies all other input symbols. Moreover, part 2 of Theorem \ref{Thm: Comparison Bounds via Degradation} conveys that
$$ \Cperm\!\left(\qSC\!\left(\frac{\eta(q-1)}{q}\right)\right) \leq \Cperm(\qEC(\eta)) \, . $$
Since $\frac{\eta(q-1)}{q} \in \big(0,\frac{q-1}{q}\big)$, it is straightforward to verify that the $q$-ary symmetric channel matrix $S_{\eta(q-1)/q}$ is strictly positive and non-singular. So, we have
$$ \Cperm(\qEC(\eta)) \geq \Cperm\!\left(\qSC\!\left(\frac{\eta(q-1)}{q}\right)\right) = \frac{q-1}{2} $$
using Theorem \ref{Thm: Strictly Positive and Full Rank Channels}, which proves the desired result.

We remark that according to the proof of part 1 of Theorem \ref{Thm: Comparison Bounds via Degradation}, an appropriately altered coding scheme from the achievability proof of Theorem \ref{Thm: Achievability Bound} in subsection \ref{Achievability Bounds for DMCs}, which has:
\begin{enumerate}
\item A randomized encoder described by \eqref{Eq: Randomized Encoder 1},
\item A decoder that first generates independent uniform random output letters to replace every erasure symbol in the output codeword, and then applies the decoder \eqref{Eq: Decoder 1}, which is characterized by \eqref{Eq: Element-wise Decoder} specialized to a $\qSC(\eta(q-1)/q)$,
\end{enumerate}
achieves the lower bound on $\Cperm(\qEC(\eta))$. Alternatively, if we directly use Theorem \ref{Thm: Achievability Bound} and the fact that $\rank(\qEC(\eta)) = q$ to obtain the lower bound on $\Cperm(\qEC(\eta))$ (as mentioned earlier), then this corresponds to using the same encoder \eqref{Eq: Randomized Encoder 1}, but an alternative decoder \eqref{Eq: Decoder 1}, which is characterized by \eqref{Eq: Element-wise Decoder} specialized to a $\qEC(\eta)$.

\underline{Converse for $\eta \in (0,1)$:} 
As mentioned earlier, the upper bound in the proposition statement is immediate from \eqref{Eq: General Combined Converse Bound} and the fact that $\ext(\qEC(\eta)) = q$. However, as outlined after \eqref{Eq: General Combined Converse Bound} in subsection \ref{Converse Bounds}, the inequality in terms of $\ext(\cdot)$ in \eqref{Eq: General Combined Converse Bound} is proved using a degradation argument akin to the proof of Theorem \ref{Thm: Converse Bound II}. Here, we provide a simpler alternative proof of the (intuitively obvious) upper bound in the proposition statement by exploiting specific properties of $q$-ary erasure channels.

Recall that \eqref{Eq: Fano step} (from the proof of Theorem \ref{Thm: Converse Bound I}) holds for a $\qEC(\eta)$, and we can bound the mutual information term in \eqref{Eq: Fano step} via
\begin{align}
I(X_1^n;\hat{P}_{Y_1^n}) & \overset{\eqmakebox[K][c]{\footnotesize (a)}}{=} I(X_1^n;\hat{P}_{Y_1^n}(\mathsf{E}),\hat{P}_{Y_1^n}(1),\dots,\hat{P}_{Y_1^n}(q-1)) \nonumber \\
& \overset{\eqmakebox[K][c]{\footnotesize (b)}}{=} I(X_1^n;\hat{P}_{Y_1^n}(1),\dots,\hat{P}_{Y_1^n}(q-1)|\hat{P}_{Y_1^n}(\mathsf{E})) \nonumber \\
& \quad \, + I(X_1^n;\hat{P}_{Y_1^n}(\mathsf{E})) \nonumber \\
& \overset{\eqmakebox[K][c]{\footnotesize (c)}}{=} I(X_1^n;\hat{P}_{Y_1^n}(1),\dots,\hat{P}_{Y_1^n}(q-1)|\hat{P}_{Y_1^n}(\mathsf{E})) \nonumber \\
& \overset{\eqmakebox[K][c]{}}{=} H(\hat{P}_{Y_1^n}(1),\dots,\hat{P}_{Y_1^n}(q-1)|\hat{P}_{Y_1^n}(\mathsf{E})) \nonumber \\
& \quad \, - H(\hat{P}_{Y_1^n}(1),\dots,\hat{P}_{Y_1^n}(q-1)|X_1^n,\hat{P}_{Y_1^n}(\mathsf{E})) \nonumber \\
& \overset{\eqmakebox[K][c]{\footnotesize (d)}}{\leq} (q-1)\log(n+1) \nonumber \\
\label{Eq: BEC Pre-Upper bound on MI} 
& \quad \, - H(\hat{P}_{Y_1^n}(1),\dots,\hat{P}_{Y_1^n}(q-1)|X_1^n,\hat{P}_{Y_1^n}(\mathsf{E})) \\
& \overset{\eqmakebox[K][c]{\footnotesize (e)}}{\leq} (q-1)\log(n+1) \, ,
\label{Eq: BEC Upper bound on MI}
\end{align}
where (a) holds because $\hat{P}_{Y_1^n}$ sums to unity and we let $\X = \{1,\dots,q\}$ without loss of generality, (b) follows from the chain rule, (c) uses the fact that $I(X_1^n;\hat{P}_{Y_1^n}(\mathsf{E})) = 0$ since the codeword $X_1^n$ is independent of the number of erasures $n\hat{P}_{Y_1^n}(\mathsf{E})$, (d) holds because $n\hat{P}_{Y_1^n}(i) \in \{0,\dots,n\}$ for every $i \in \{1,\dots,q-1\}$, and (e) follows from the non-negativity of Shannon entropy. Therefore, as before, substituting \eqref{Eq: BEC Upper bound on MI} into \eqref{Eq: Fano step}, dividing by $\log(n)$, and letting $n \rightarrow \infty$, we get that any achievable rate $R \geq 0$ satisfies $R \leq q-1$. This proves that $\Cperm(\qEC(\eta)) \leq q-1$. 

\underline{Case $\eta = 0$:} In this case, the $\qEC(0)$ is just the deterministic identity channel. Hence, $\Cperm(\qEC(0)) = q-1$ using Proposition \ref{Prop: Permutation Stochastic Matrices}.

\underline{Case $\eta = 1$:} In this case, the $\qEC(1)$ erases all its input symbols so that we obtain an ``independent channel.'' Hence, $\Cperm(\qEC(1)) = 0$ using Proposition \ref{Prop: Unit Rank Stochastic Matrices}.
\end{proof}

We finally make several pertinent remarks. Firstly, in the special case of $q = 2$ and $\eta \in (0,1)$, the elegant interpretation of a $\BSC\big(\frac{\eta}{2}\big)$ as a $\BEC(\eta)$ which additionally replaces any output erasure symbol $\mathsf{E}$ with an independent $\Ber\big(\frac{1}{2}\big)$ bit, or equivalently, the decomposition \eqref{Eq: Fortuin-Kasteleyn}, is a notion that originates from \textit{Fortuin-Kasteleyn random cluster representations of Ising models} in the study of percolation, cf. \cite{Grimmett1997}. Moreover, this notion has been exploited in various other discrete probability contexts such as reliable computation using noisy circuits \cite[p.570]{Feder1989}, broadcasting on trees \cite[p.412]{Evansetal2000}, and broadcasting on directed acyclic graphs \cite[Appendix C]{MakurMosselPolyanskiy2020} (also see \cite[p.1634]{MakurMosselPolyanskiy2019}).

Secondly, in the special case of $q = 2$ and $\eta \in (0,1)$, we propose the following conjecture.

\begin{conjecture}[$\Cperm$ of BECs]
\label{Conj: BEC Conjecture}
For any erasure probability $\eta \in (0,1)$, the noisy permutation channel capacity of the $\BEC(\eta)$ is given by
$$ \Cperm(\BEC(\eta)) = \frac{1}{2} \, . $$
\end{conjecture}

Specifically, we believe that the achievability bound presented in Proposition \ref{Prop: Permutation Channel Capacity of q-EC} is tight. Consequently, unlike traditional channel capacity, we believe that the noisy permutation channel capacities of BSCs and BECs are equal in the non-trivial regimes of their parameters. Indeed, the converse bound in Proposition \ref{Prop: Permutation Channel Capacity of q-EC} for the case $q = 2$, $\Cperm(\BEC(\eta)) \leq 1$, is intuitively trivial as there are only $n+1$ distinct empirical distributions of codewords in $\{0,1\}^n$. So, we anticipate that this bound can be significantly tightened. One approach towards tightening the converse bound would be to consider the mutual information bound in \eqref{Eq: BEC Pre-Upper bound on MI}, and much like the proof of Theorem \ref{Thm: Converse Bound I} in subsection \ref{Converse Bounds for Strictly Positive DMCs}, derive a lower bound on $H(\hat{P}_{Y_1^n}(1)|X_1^n,\hat{P}_{Y_1^n}(\mathsf{E}))$ of the form
\begin{equation}
\label{Eq: BEC Conditional entropy lower bound}
H(\hat{P}_{Y_1^n}(1)|X_1^n,\hat{P}_{Y_1^n}(\mathsf{E})) \geq \frac{1}{2}\log(n) + o(\log(n)) \, .
\end{equation}
Clearly, combining \eqref{Eq: Fano step}, \eqref{Eq: BEC Pre-Upper bound on MI}, and \eqref{Eq: BEC Conditional entropy lower bound} would yield the desired bound $\Cperm(\BEC(\eta)) \leq \frac{1}{2}$. As explained in \cite[Equation (26)]{Makur2018}, finding a bound of the form \eqref{Eq: BEC Conditional entropy lower bound} corresponds to analyzing the Shannon entropy of hypergeometric distributions in non-trivial regimes of their parameters.  

Thirdly, when $q \geq 3$, we also postulate that for any $\eta \in (0,1)$,
\begin{equation}
\label{Eq: Conjecture 3}
\Cperm(\qEC(\eta)) \leq \frac{q}{2} \, .
\end{equation}
While this upper bound does not exactly determine the noisy permutation channel capacity of $q$-ary erasure channels, it does have the following useful corollaries (if proven to be true):
\begin{enumerate}
\item In the limit of asymptotically large input alphabet size (i.e., as $q \rightarrow \infty$), the noisy permutation channel capacity of $q$-ary erasure channels is characterized by
\begin{equation}
\forall \eta \in (0,1), \enspace \lim_{q \rightarrow \infty}{\frac{\Cperm(\qEC(\eta))}{q}} = \frac{1}{2} \, .
\end{equation}
\item Applying part 2 of Theorem \ref{Thm: Comparison Bounds via Degradation}, if a DMC $P_{Z|X}$ satisfies the Doeblin minorization condition, then we obtain the converse bound
\begin{equation}
\Cperm(P_{Z|X}) \leq \frac{|\X|}{2} \, . 
\end{equation}
This bound is clearly weaker than that in Theorem \ref{Thm: Converse Bound II} for strictly positive DMCs. However, it also holds for certain DMCs that have zero entries, and therefore, extends Theorem \ref{Thm: Converse Bound II} for such DMCs. 
\end{enumerate}
We believe \eqref{Eq: Conjecture 3} could be true, because we can upper bound the mutual information term in \eqref{Eq: Fano step} so that
\begin{align}
I(X_1^n;\hat{P}_{Y_1^n}) & \overset{\eqmakebox[G][c]{}}{=} H(\hat{P}_{Y_1^n}) - H(\hat{P}_{Y_1^n}|X_1^n) \nonumber \\
& \overset{\eqmakebox[G][c]{\footnotesize (a)}}{\leq} q \log(n + 1) - H(\hat{P}_{Y_1^n}(1),\dots,\hat{P}_{Y_1^n}(q)|X_1^n) \nonumber \\
& \overset{\eqmakebox[G][c]{\footnotesize (b)}}{=} q \log(n + 1) - H(N_\mathsf{E}(1),\dots,N_\mathsf{E}(q)|X_1^n) \nonumber \\
& \overset{\eqmakebox[G][c]{\footnotesize (c)}}{=} q \log(n + 1) - \sum_{i = 1}^{q}{H(N_\mathsf{E}(i)|X_1^n)} \nonumber \\
& \overset{\eqmakebox[G][c]{\footnotesize (d)}}{=} q \log(n + 1) - \sum_{i = 1}^{q}{\E\!\left[H(\bin(n \hat{P}_{X_1^n}(i),\eta))\right]} \, , 
\label{Eq: Another MI Bound}
\end{align}
where (a) follows from the bound in \eqref{Eq: Number of Emp Dists}, the fact that $\hat{P}_{Y_1^n}$ sums to unity, and by letting $\X = \{1,\dots,q\}$ without loss of generality, (b) follows from defining the number of erasures that occur on the input symbol $i \in \X$ as the random variable $N_\mathsf{E}(i) \triangleq n \hat{P}_{X_1^n}(i) - n \hat{P}_{Y_1^n}(i) \in \{0,\dots,n \hat{P}_{X_1^n}(i)\}$, (c) holds because $\{N_\mathsf{E}(i) : i \in \X\}$ are conditionally independent given $X_1^n$, (d) holds because each $N_\mathsf{E}(i)$ is a binomial random variable with $n \hat{P}_{X_1^n}(i)$ trials and success probability $\eta$ conditioned on $X_1^n$, and each term $\E\big[H(\bin(n \hat{P}_{X_1^n}(i),\eta))\big]$ represents the expectation of a binomial entropy with respect to the distribution of $X_1^n$. If it can be shown that any encoder, which achieves vanishing probability of error for rates ``close to'' the noisy permutation channel capacity, must satisfy $\hat{P}_{X_1^n}(i) \geq \alpha$ for all $i \in \X$ for some constant $\alpha \in (0,1)$ with high probability, then \eqref{Eq: Fano step}, \eqref{Eq: Another MI Bound}, and Lemma \ref{Lemma: Approximation of Binomial Entropy} would yield \eqref{Eq: Conjecture 3}. However, proving such a lower bound on $\hat{P}_{X_1^n}$ appears to be challenging (if at all possible), since we essentially have to develop a ``probabilistic pigeonhole principle'' to argue that ``good'' encoders need to utilize all the symbols in $\X$ significantly.

\section{Conclusion}
\label{Conclusion}

In closing, we first briefly reiterate our main contributions. Propelled by existing literature in coding theory, communication networks, and molecular and biological communications, we formulated the information theoretic notion of noisy permutation channel capacity for the problem of reliably transmitting information through a noisy permutation channel, i.e., a DMC followed by an independent random permutation transformation. We then derived achievability and converse bounds on noisy permutation channel capacities in Theorems \ref{Thm: Achievability Bound}, \ref{Thm: Converse Bound I}, and \ref{Thm: Converse Bound II} (as well as in \eqref{Eq: General Combined Converse Bound}). These results gave rise to an exact characterization of the noisy permutation channel capacity of strictly positive and full rank DMCs in Theorem \ref{Thm: Strictly Positive and Full Rank Channels}. Furthermore, in our effort to prove these results and acquire a deeper understanding of noisy permutation channel capacity, we elucidated a simple construction of symmetric channels that dominate given DMCs in the degradation sense in Proposition \ref{Prop: Degradation by Symmetric Channels}, and established an intuitive monotonicity relation between noisy permutation channel capacity and degradation in Theorem \ref{Thm: Comparison Bounds via Degradation}.

We next propose some directions for future research. Evidently, addressing any of the open problems explicated in Conjectures \ref{Conj: Strictly Positive DMC Conjecture}, \ref{Conj: BEC Conjecture}, and \eqref{Eq: Conjecture 3} is an excellent starting point to furthering this line of work. After these conjectures are resolved, our ultimate objective is to establish the noisy permutation channel capacity of general DMCs (whose row stochastic matrices can have zero entries). We remark that determining the noisy permutation channel capacities of DMCs with zero entries appears to be more intractable than strictly positive DMCs, because zero entries introduce a combinatorial flavor to the problem.\footnote{This combinatorial aspect of the problem is similar to (but not exactly the same as) the zero error capacity problem, cf. \cite{Shannon1956}. It is well-known that calculating the zero error capacity of channels is very challenging, and the best known approaches use semidefinite programming relaxations such as the Lov\'{a}sz $\vartheta$ function, cf. \cite{Lovasz1979}.} So, completely settling the noisy permutation channel capacity question for general DMCs is likely to be quite challenging. Finally, there are several other open questions that parallel aspects of classical information theoretic development such as:
\begin{enumerate}
\item Finding tight bounds on the average probability of error (akin to classical \textit{error exponent analysis}), cf. \cite[Chapter 5]{Gallager1968}.
\item Developing \textit{strong converse} results, cf. \cite[Section 22.1]{PolyanskiyWu2017Notes}, \cite[Theorem 5.8.5]{Gallager1968}.
\item Establishing \textit{exact asymptotics} for the maximum achievable value of $|\M|$ (akin to ``finite blocklength analysis''), cf. \cite{PolyanskiyPoorVerdu2010}, \cite[Chapter II.4]{Tan2014}, and the references therein.
\item Extending the noisy permutation channel model by replacing DMCs with other kinds of memoryless channels or networks, e.g., additive white Gaussian noise (AWGN) channels or multiple-access channels (MACs), and by using more general or ``realistic'' algebraic operations that are applied to the output codewords, e.g., random permutations that belong to subgroups of the symmetric group. (For example, when modeling out-of-order delivery of packets in a communication network, all permutations of the packets are not equally likely; indeed, the first two transmitted packets are quite likely to arrive swapped at the receiver, but the first and last transmitted packets are very unlikely to change their relative ordering.) 
\end{enumerate} 
Altogether, our main results and these future directions illustrate that the study of noisy permutation channel capacity begets a fairly rich, relevant, and seemingly solvable class of new problems.

\appendices

\section{Proof of Proposition \ref{Prop: Degradation by Symmetric Channels}}
\label{Proof of Proposition Degradation by Symmetric Channels}

\begin{proof}
To prove this result, we seek to find $\qSC(\delta)$'s with $\delta \in \big[0,\frac{q-1}{q}\big]$ such that $P_{Z|X}$ is a degraded version of $\qSC(\delta)$. Indeed, it is straightforward to see that the upper bound on $\delta$ in the proposition statement satisfies
\begin{equation}
\label{Eq: Upper Bound Constraint}
\frac{\nu}{1 - \nu + \frac{\nu}{q-1}} \leq \frac{q-1}{q} \, ,  
\end{equation}
because \eqref{Eq: Upper Bound Constraint} is equivalent to
$$ \frac{\nu (q-1)}{(q-1) - \nu (q-2)} \leq \frac{q-1}{q} \quad \Leftrightarrow \quad \nu \leq \frac{1}{2} $$
for any $q \in \N\backslash\!\{1\}$, and the latter bound is always true since $|\Y| \geq 2$. Furthermore, we have equality in \eqref{Eq: Upper Bound Constraint} if and only if $\nu = \frac{1}{2}$, which happens precisely when $|\Y| = 2$ and all rows of $P_{Z|X}$ are equal to $\big(\frac{1}{2},\frac{1}{2}\big)$. In this case, $P_{Z|X} = S_{\delta} P_{Z|X}$ for all $\delta \in [0,1]$, and the sufficient condition for degradation in the proposition statement holds trivially. So, we will assume without loss of generality that $\nu < \frac{1}{2}$ in the rest of the proof.

To construct $\qSC(\delta)$'s with $\delta \in \big[0,\frac{q-1}{q}\big)$ such that $P_{Z|X}$ is a degraded version of $\qSC(\delta)$, we must ensure that $P_{Z|X} = S_{\delta} Q$ for some row stochastic matrix $Q \in \R^{q \times |\Y|}$. Equivalently, $S_{\delta}^{-1} P_{Z|X}$ must be a row stochastic matrix. A direct calculation yields $S_{\delta}^{-1} = S_{\tau}$ with (cf. \cite[Proposition 4]{MakurPolyanskiy2018})
\begin{equation}
\label{Eq: Inverse Parameter}
\tau = \frac{-\delta}{1 - \delta - \frac{\delta}{q-1}} \, , 
\end{equation}
i.e., $S_{\delta}^{-1} = S_{\tau}$ has the structure shown in \eqref{Eq: q-SC Matrix} (but with $\tau$ replacing $\delta$). Since the rows of $S_{\tau}$ sum to unity, the rows of $S_{\delta}^{-1} P_{Z|X} = S_{\tau} P_{Z|X}$ also sum to unity. Thus, it suffices to verify that the minimum entry of $S_{\tau} P_{Z|X}$ is non-negative. 

For $\delta \in \big[0,\frac{q-1}{q}\big)$, we have $\tau \leq 0$, which means that the principal diagonal elements of $S_{\tau}$ are at least unity, and the off-diagonal elements of $S_{\tau}$ are non-positive. Hence, the minimum entry of $S_{\tau} P_{Z|X}$ is lower bounded by
\begin{align*}
\min_{\substack{i \in \{1,\dots,q\}\\j\in \{1,\dots,|\Y|\}}}{\left[ S_{\tau} P_{Z|X} \right]_{i,j}} & \geq (1-\tau)\nu + \tau (1-\nu) \\
& = \frac{\nu - \delta\!\left(1 - \nu + \frac{\nu}{q-1}\right)}{1 - \delta - \frac{\delta}{q-1}} \, ,
\end{align*}
where the inequality uses the fact that the maximum entry of $P_{Z|X}$ is upper bounded by $1-\nu$ (because $P_{Z|X}$ is a stochastic matrix), and the equality follows from substituting \eqref{Eq: Inverse Parameter}. So, a sufficient condition that ensures that the minimum entry of $S_{\tau} P_{Z|X}$ is non-negative is
$$ \frac{\nu - \delta\!\left(1 - \nu + \frac{\nu}{q-1}\right)}{1 - \delta - \frac{\delta}{q-1}} \geq 0 \quad \Leftrightarrow \quad \delta \leq \frac{\nu}{1 - \nu + \frac{\nu}{q-1}} \, . $$
This completes the proof.
\end{proof}

\section{Proof of Lemma \ref{Lemma: Second Moment Method}}
\label{Proof of Lemma Second Moment Method}

\begin{proof}
For the binary hypothesis problem in \eqref{Eq: BHT Problem}, define the ``translated empirical distribution of $X_1^n$'' random vector
\begin{equation}
\label{Eq: SS}
T_n \triangleq \hat{P}_{X_1^n} - C_n \in \mathcal{T}_{n} \, , 
\end{equation}
where the constant vector $C_n = (c_1,\dots,c_{|\X|}) \in \R^{|\X|}$ (which can depend on $n$) will be chosen later, and $\mathcal{T}_n = \big\{\big(\frac{k_1}{n} - c_1,\dots,\frac{k_{|\X|}}{n} - c_{|\X|}\big) : k_1,\dots,k_{|\X|} \in \N\cup\!\{0\}, \, k_1 + \cdots + k_{|\X|} = n\big\}$. Moreover, for ease of exposition, let $T_n^-$ and $T_n^+$ denote versions of the random variable $T_n$ with probability distributions $P_{T_n}^-$ and $P_{T_n}^+$ induced by $P_X^{\otimes n}$ and $Q_X^{\otimes n}$, respectively, such that
\begin{align*}
P_{T_n}^-(t) & \triangleq \P\!\left(T_n^- = t\right) = \P\!\left(\hat{P}_{X_1^n} = t + C_n \middle|H = 0\right) , \\
P_{T_n}^+(t) & \triangleq \P\!\left(T_n^+ = t\right) = \P\!\left(\hat{P}_{X_1^n} = t + C_n \middle|H = 1\right) ,
\end{align*}
for all $t \in \mathcal{T}_{n}$. Then, we clearly have
$$ P_{T_n} = \frac{1}{2} P_{T_n}^{-} + \frac{1}{2} P_{T_n}^{+} \, . $$
It is straightforward to verify that $T_n$ is a sufficient statistic of $X_1^n$ for performing inference about $H$. So, the ML decoder of $H$ based on $X_1^n$, $\hat{H}_{\mathsf{ML}}^n(X_1^n)$, is a function of $T_n$ without loss of generality (see \eqref{Eq: ML decoder}), and we denote it as $\hat{H}_{\mathsf{ML}}^n: \mathcal{T}_n \rightarrow \{0,1\}$, $\hat{H}_{\mathsf{ML}}^n(T_n)$ with abuse of notation. Thus, we have $P_{\mathsf{ML}}^{(n)} = \P(\hat{H}_{\mathsf{ML}}^n(T_n) \neq H)$, and \eqref{Eq: Le Cam relation} implies that
\begin{equation}
\label{Eq: TV Equality}
\left\|P_X^{\otimes n}  - Q_X^{\otimes n} \right\|_{\mathsf{TV}} = \left\|P_{T_n}^+ - P_{T_n}^-\right\|_{\mathsf{TV}} . 
\end{equation}
It therefore suffices to lower bound the right hand side.

Similar to the proof of \cite[Lemma 4.2(iii)]{Evansetal2000}, observe that
\begin{align}
& \left\|\E\!\left[T_n^+\right] - \E\!\left[T_n^-\right]\right\|_2^2 \nonumber \\
& \overset{\eqmakebox[A][c]{}}{=} \left\|\sum_{t \in \mathcal{T}_n}{t \left(P_{T_n}^+(t) - P_{T_n}^-(t)\right)}\right\|_2^2 \nonumber \\
& \overset{\eqmakebox[A][c]{\footnotesize (a)}}{=} \sum_{i = 1}^{|\X|}{\left(\sum_{t \in \mathcal{T}_n}{\left(\frac{P_{T_n}^+(t) - P_{T_n}^-(t)}{\sqrt{P_{T_n}(t)}}\right) t_i \sqrt{P_{T_n}(t)}}\right)^{\!2}} \nonumber \\
& \overset{\eqmakebox[A][c]{\footnotesize (b)}}{\leq} 4 \underbrace{\left(\frac{1}{4} \sum_{t \in \mathcal{T}_n}{\frac{\left(P_{T_n}^+(t) - P_{T_n}^-(t)\right)^{\!2}}{P_{T_n}(t)}}\right)}_{\triangleq \, \mathsf{LC}(P_{T_n}^+||P_{T_n}^-)} \!\left(\sum_{i = 1}^{|\X|}{\sum_{t \in \mathcal{T}_n}{t_i^2 P_{T_n}(t)}}\right) \nonumber \\
& \overset{\eqmakebox[A][c]{}}{=} 4 \, \mathsf{LC}\big(P_{T_n}^+ \big|\big|P_{T_n}^-\big) \,  \E\!\left[\left\|T_n\right\|_2^2\right] 
\label{Eq: HCR Bound}
\\ 
& \overset{\eqmakebox[A][c]{\footnotesize (c)}}{\leq} 4 \, \E\!\left[\left\|T_n\right\|_2^2\right] \left\|P_{T_n}^+ - P_{T_n}^-\right\|_{\mathsf{TV}} ,
\label{Eq: CS Bound}
\end{align}
where we let $t = (t_1,\dots,t_{|\X|})$ in (a), (b) follows from the Cauchy-Schwarz-Bunyakovsky inequality, $\mathsf{LC}(\cdot||\cdot)$ denotes the \textit{Vincze-Le Cam distance} or \textit{triangular discrimination} between two probability distributions \cite{Vincze1981,LeCam1986}, and (c) upper bounds Vincze-Le Cam distance using TV distance (via the observation that $\big|P_{T_n}^+(t) - P_{T_n}^-(t)\big| \leq P_{T_n}^+(t) + P_{T_n}^-(t)$ for all $t \in \mathcal{T}_n$). We note that \eqref{Eq: HCR Bound} is precisely a vector version of \cite[Lemma 4.2(iii)]{Evansetal2000}. Hence, combining \eqref{Eq: TV Equality} and \eqref{Eq: CS Bound}, we get
\begin{equation}
\label{Eq: General SMM Bound}
\left\|P_X^{\otimes n}  - Q_X^{\otimes n} \right\|_{\mathsf{TV}} \geq \frac{\left\|\E\!\left[T_n^+\right] - \E\!\left[T_n^-\right]\right\|_{2}^2}{4 \, \E\!\left[\left\|T_n\right\|_{2}^2 \right]} \, . 
\end{equation}

We now select the vector $C_n$. Since the numerator of the bound in \eqref{Eq: General SMM Bound} is invariant to the value of $C_n$, the best bound of the form \eqref{Eq: General SMM Bound} is obtained by minimizing the second moment $\E\big[\big\|T_n\big\|_2^2\big]$. Thus, $C_n$ is given by
\begin{equation}
\label{Eq: Constant Shift}
C_n = \E\!\left[\hat{P}_{X_1^n}\right] = \frac{1}{2} P_X + \frac{1}{2} Q_X \, ,
\end{equation}
using the binary hypothesis testing model \eqref{Eq: BHT Problem}, where we employ the well-known fact that mean-squared error is minimized by the mean (see, e.g., \cite[Section 1.7, Example 7.17]{LehmannCasella1998}). With this choice of $C_n$, notice that
\begin{align*}
\E\!\left[T_n^-\right] & = \frac{1}{2} P_X - \frac{1}{2} Q_X \, , \\
\E\!\left[T_n^+\right] & = \frac{1}{2} Q_X - \frac{1}{2} P_X \, , \\
\E\!\left[\left\|T_n\right\|_2^2\right] & = \sum_{x \in \X}{\VAR\big(\hat{P}_{X_1^n}(x)\big)} \, .
\end{align*}
Using these expressions, we can simplify the second moment method bound in \eqref{Eq: General SMM Bound} and obtain the bound in the lemma statement.
\end{proof} 

We remark that with the choice of $C_n$ in \eqref{Eq: Constant Shift}, \eqref{Eq: HCR Bound} can be perceived as a vector version of the HCR bound in statistics \cite{Hammersley1950,ChapmanRobbins1951}, where the Vincze-Le Cam distance replaces the usual $\chi^2$-divergence.

\section{Proof of Lemma \ref{Lemma: Testing between Converging Hypotheses}}
\label{Proof of Lemma Testing between Converging Hypotheses}

\begin{proof}
To upper bound the ML decoding probability of error $P_{\mathsf{ML}}^{(n)}$, we combine \eqref{Eq: Le Cam relation} and Lemma \ref{Lemma: Second Moment Method} to get
\begin{equation}
\label{Eq: Bound on Probability of Error}
P_{\mathsf{ML}}^{(n)} \leq \frac{1}{2}\left(1 - \frac{\left\|P_X - Q_X\right\|_{2}^2}{\displaystyle{4 \sum_{x \in \X}{\VAR\big(\hat{P}_{X_1^n}(x)\big)}}}\right) . 
\end{equation}
We now compute the right hand side of this bound explicitly. Observe using \eqref{Eq: Constant Shift} that for every $x \in \X$,
\begin{align*}
\VAR\big(\hat{P}_{X_1^n}(x)\big) & = \frac{1}{n^2} \, \E\!\left[ \left(\sum_{i = 1}^{n}{\I\{X_i = x\}}\right)^{\! 2}\right] \\
& \quad \, - \left(\frac{P_X(x) + Q_X(x)}{2}\right)^{\! 2} \\
& = \frac{P_X(x) + Q_X(x)}{2 n} - \left(\frac{P_X(x) + Q_X(x)}{2}\right)^{\! 2} \\
& \quad \, + \frac{1}{n^2} \sum_{\substack{1 \leq i,j \leq n\\i \neq j}}{\E\!\left[\I\{X_i = x\} \I\{X_j = x\}\right]} \\
& = \frac{P_X(x) + Q_X(x)}{2 n} - \left(\frac{P_X(x) + Q_X(x)}{2}\right)^{\! 2} \\
& \quad \, + \left(\frac{n-1}{2 n}\right) \!\left(P_X(x)^2 + Q_X(x)^2\right) \\
& = \frac{P_X(x) \left(1 - P_X(x)\right)}{2 n} \\
& \quad \, + \frac{Q_X(x) \left(1 - Q_X(x)\right)}{2 n} \\
& \quad \, + \frac{\left(P_X(x) - Q_X(x)\right)^2}{4} \\
& \leq \frac{1}{4n} + \frac{\left(P_X(x) - Q_X(x)\right)^2}{4} \, ,
\end{align*}
where the equalities follow from straightforward algebraic manipulations, and the final inequality holds because $t (1 - t) \leq \frac{1}{4}$ for all $t \in [0,1]$. Plugging this inequality into \eqref{Eq: Bound on Probability of Error} yields
\begin{align*}
P_{\mathsf{ML}}^{(n)} & \leq \frac{1}{2}\left(1 - \frac{\left\|P_X - Q_X\right\|_{2}^2}{\frac{|\X|}{n} + \left\|P_X - Q_X \right\|_2^2}\right) \\
& = \frac{|\X|}{2 |\X| + 2 n \left\|P_X - Q_X \right\|_2^2} \\
& \leq \frac{|\X|}{2 |\X| + 2 n^{2\epsilon_n}} \, ,
\end{align*}
where the final inequality follows from applying \eqref{Eq: 2-norm condition}. This completes the proof. 
\end{proof}

\section{Proof of Lemma \ref{Lemma: Doeblin Minorization and Degradation}}
\label{Proof of Lemma Doeblin Minorization and Degradation}

\begin{proof} ~\newline
\indent
\underline{Part 1:} For the convenience of readers unfamiliar with the notion of \textit{iterated random maps}, we translate the proofs of \cite[Theorem 3.1, Proposition 4.1]{BhattacharyaWaymire2001} into information theoretic language. (We also refer readers to \cite[Remark III.2]{Raginsky2016}, which shows the forward direction.) 

Suppose $P_{Z|X}$ satisfies $\mathsf{Doeblin}(Q_Z,\eta)$. Then, construct the channel $P_{Z|X^{\prime}}$ with input alphabet $\X \cup\{\mathsf{E}\}$ and output alphabet $\Y$ such that
$$ P_{Z|X^{\prime}}(z|x) = 
\begin{cases}
\displaystyle{\frac{P_{Z|X}(z|x) - \eta Q_Z(z)}{1 - \eta}} , & \text{for } x \in \X \\
Q_Z(z) , & \text{for } x = \mathsf{E}
\end{cases}
$$
for all $z \in \Y$ and $x \in \X \cup\{\mathsf{E}\}$, where $P_{Z|X}(z|x) - \eta Q_Z(z) \geq 0$ due to Definition \ref{Def: Doeblin Minorization}, and $\sum_{z \in \Y}{P_{Z|X}(z|x) - \eta Q_Z(z)} = 1 - \eta$. It follows via a direct calculation that $P_{Z|X} = \qEC(\eta) \cdot P_{Z|X^{\prime}}$ (i.e., $P_{Z|X}$ is the product of the stochastic matrices $\qEC(\eta)$ and $P_{Z|X^{\prime}}$), which means that $P_{Z|X}$ is a degraded version of $\qEC(\eta)$.

To prove the reverse direction, suppose $P_{Z|X}$ is a degraded version of $\qEC(\eta)$. Then, using Definition \ref{Def: Degradation Preorder}, there exists a channel $P_{Z|X^{\prime}}$ with input alphabet $\X \cup\!\{\mathsf{E}\}$ and output alphabet $\Y$ such that $P_{Z|X} = \qEC(\eta) \cdot P_{Z|X^{\prime}}$. Hence, it is straightforward to show that for every $x \in \X$ and $y \in \Y$,
\begin{align*}
P_{Z|X}(z|x) & = (1-\eta) P_{Z|X^{\prime}}(z|x) + \eta P_{Z|X^{\prime}}(z|\mathsf{E}) \\
& \geq \eta P_{Z|X^{\prime}}(z|\mathsf{E}) \, ,
\end{align*}
where the inequality holds because $(1-\eta) P_{Z|X^{\prime}}(z|x) \geq 0$. Thus, employing Definition \ref{Def: Doeblin Minorization}, this implies that $P_{Z|X}$ satisfies $\mathsf{Doeblin}(P_{Z|X^{\prime}}(\cdot|\mathsf{E}),\eta)$. This completes the proof of part 1. 
 
\underline{Part 2:} We refer readers to \cite[Lemma 4]{GohariGunluKramer2019} for a proof of this part. (It is worth juxtaposing $\eta_{*}(P_{Z|X})$ with \cite[Equations (58) and (102)]{MakurZheng2020}, which state that contraction coefficients of operator convex $f$-divergences characterize the extremal erasure probability $\eta$ such that $P_{Z|X}$ is dominated by a $\qEC(\eta)$ in the ``less noisy'' sense; see \cite{MakurZheng2020} for details.)  
\end{proof}

\bibliographystyle{IEEEtran}
\bibliography{finalrefs}

\end{document}